\documentclass[12pt]{article}
\usepackage{mathrsfs}
\usepackage{graphicx}
\usepackage{amsmath}
\usepackage{amsthm}
\usepackage{amssymb}
\usepackage{amsfonts}
\usepackage{epsfig}
\usepackage{hyperref,color}
\usepackage{epsf}
\usepackage{epstopdf}
\usepackage[title,titletoc]{appendix}
\def\Im {\mathop{\rm Im}\nolimits}
\def\arg {\mathop{\rm arg}\nolimits}
\def\Re {\mathop{\rm Re}\nolimits}
\def\Ai {{\rm Ai}}

\newcommand{\D}{\displaystyle}

\newtheorem{pro}{PROPOSITION}
\newtheorem{cor}{COROLLARY}
\newtheorem{lem}{LEMMA}
\newtheorem{thm}{THEOREM}
\newtheorem{rem}{REMARK}

\newtheorem{rhp}{RH problem}[section]

\numberwithin{equation}{section}

\begin{document}

\title{Tracy-Widom distributions in critical unitary random matrix ensembles and the coupled Painlev\'{e} II system}
\author{Shuai-Xia Xu\footnotemark[1] ~and Dan Dai\footnotemark[2] ~}

\renewcommand{\thefootnote}{\fnsymbol{footnote}}
\footnotetext[1]{Institut Franco-Chinois de l'Energie Nucl\'{e}aire, Sun Yat-sen University,
Guangzhou 510275, China. E-mail: \texttt{xushx3@mail.sysu.edu.cn}}
\footnotetext[2]{Department of Mathematics, City University of Hong Kong, Tat Chee
Avenue, Kowloon, Hong Kong. E-mail: \texttt{dandai@cityu.edu.hk} (Corresponding author)}
\date{}
\maketitle

\maketitle

\noindent \hrule width 6.27in \vskip .3cm

\noindent {\bf{Abstract }} We study Fredholm determinants of the Painlev\'{e} II and Painlev\'{e} XXXIV kernels. In certain critical unitary random matrix ensembles, these determinants describe special gap probabilities of eigenvalues. We obtain Tracy-Widom formulas for the Fredholm determinants, which are explicitly given in terms of integrals involving a family of distinguished solutions to the coupled Painlev\'{e} II system  in dimension four. Moreover, the large gap asymptotics for these Fredholm determinants are derived, where the constant terms are given explicitly in terms of the Riemann zeta-function.


\vspace{.2in}


\vskip .3cm

\noindent \hrule width 6.27in\vskip 1.3cm

\tableofcontents

\section{Introduction} \setcounter{section} {1}

We consider the space of $n \times n$ Hermitian matrices $M$ with probability distribution
\begin{equation}\label{GUE}
\frac{1}{\mathcal{Z}_{n}}|\det M|^{2\alpha}e^{-n \mathrm{ Tr}V(M)}dM \qquad \textrm{for } \alpha>-\frac {1}{2},
\end{equation}
where $V:  \mathbb{R} \to \mathbb{R}$ is a real analytic function and satisfies
\begin{equation}
  \lim_{x\to\pm\infty}\frac {V(x)}{\log(1+x^2)}=+\infty.
\end{equation}
Here $dM$  is the Lebesgue measure for Hermitian matrices and $\mathcal{Z}_{n}$ is the normalization constant. It is well-known that the joint probability density function for the eigenvalues of $M$ is given by
\begin{equation}\label{pdf}
p_n(\lambda_1,\cdots,\lambda_n)=\frac{1}{Z_{n}}
\prod_{i=1}^n |\lambda_i|^{2\alpha}e^{- nV(\lambda_i)}\prod_{i<j}(\lambda_i-\lambda_j)^2,
\end{equation}
which can be put into a determinantal form
\begin{equation}\label{GUE-pdf-d}
p_n(\lambda_1,\cdots,\lambda_n)=\frac{1}{n!}
\det(K_{n}(\lambda_i,\lambda_j))_{i,j=1}^n,
\end{equation}
with the correlation kernel
\begin{equation}\label{Weighted OP kernel}
K_{n}(x,y)=|xy|^\alpha  e^{-\frac {n}{2}V(x)} e^{-\frac {n}{2}
V(y)}\sum_{k=0}^{n-1}P_{k}(x)P_{k}(y);
\end{equation}
see for example \cite{deift} and \cite{m}. The above kernel is so-called \emph{orthogonal polynomial kernel}, and $P_{k}(x)$ is the $k$-th degree orthonormal polynomial with
respect to the weight $|x|^{2\alpha}e^{-nV(x)}$. 

In the global regime, the limiting mean density of eigenvalues is
\begin{equation}\label{circle law}
\rho_V(x)=\lim_{n\to \infty }\frac {1}{n}K_{n}(x,x),
\end{equation}
which depends on the exact potential $V(x)$; see \cite{deift} and \cite{dkm:Jat1998}. However, the local statistics of eigenvalues only rely on some general characteristics of the density function $\rho_V(x)$ and satisfy fascinating universal behaviors. For example, given a general real analytic potential $V$ and $\alpha=0$, the \emph{bulk university} holds for any point $x^*$  in the bulk of the spectrum, which implies that the limiting correlation kernel is the sine kernel
\begin{equation}\label{Sine kernel limit}
 \lim_{n\rightarrow\infty}
\frac {1}{n\rho_V(x^*)}K_n(x^*+\frac {x}{n\rho_V(x^*)}, x^*+\frac {y}{n\rho_V(x^*)})=  K^{\sin}(x,y) =\frac{\sin\pi(x-y)}{\pi(x-y)},
\end{equation}
uniformly for $x$ and $y$ in compact subsets of $\mathbb{R}$ whenever $\rho_V(x^*)>0$. Moreover, the \emph{soft edge universality} holds at a right regular edge point $b$ of the support of $\rho_V(x)$. This means that, when $\rho_V(x) \sim \frac{c}{\pi} (b-x)^{\frac{1}{2}}$ for $c>0$ as $x \to b-$, the limiting correlation kernel is the Airy kernel
\begin{equation}\label{Airy kernel limit}
 \lim_{n\rightarrow\infty}
\frac {1} {(cn)^{2/3}}K_n(b+\frac {x} {(cn)^{2/3}}, b+\frac{y} {(cn)^{2/3}})=K^{\Ai}(x,y)=\frac{\mathrm{Ai}(x)\mathrm{Ai}'(y)-\mathrm{Ai}'(x) \mathrm{Ai}(y)}{x-y},
\end{equation}
uniformly for $x$ and $y$ in compact subsets of $\mathbb{R}$. The above results for the bulk and soft edge universality were proved in \cite{bi,dkmv1,Pas:Shc1997}.

Like other determinantal point processes, all information of unitary random matrix ensembles is contained in the correlation kernel $K$. If we consider the \emph{gap probability} that there is no eigenvalue near the point $x^*$ in the bulk of the spectrum, it is given in terms of a Fredholm determinant as follows
\begin{equation}\label{gap pro-sine}
\lim_{n\rightarrow\infty} \mbox{Prob}\left[\mbox{$M$ has no eigenvalues in }(x^*-\frac{s}{n\rho_V(x^*)}, x^*+\frac {s}{n\rho_V(x^*)} )
\right]=\det[I-K^{\sin}_{s}],
\end{equation}
where $K^{\sin}_s$ is the trace-class operator acting on $L^2(-s,s)$ with the sine kernel $K^{\sin}(x,y)$ in \eqref{Sine kernel limit} and
the Fredholm determinant $\det[I-K^{\sin}_{s}]$ is given by the following series
\begin{equation}\label{FD-expansiion}
\det[I-K^{\sin}_s]=\sum_{k=0}^{+\infty}\frac {(-1)^k}{k!}\int_{-s}^{s}...\int_{-s}^{s}\det(K^{\sin}(x_i,x_j))_{i,j=1}^{k}dx_1 \cdots dx_k.
\end{equation}
We may also consider the gap probability near the rightmost edge point, i.e., the distribution of the largest eigenvalue. Let $\lambda_n$ be the largest eigenvalue of the matrix $M$ and $b$ be the rightmost regular edge point. Then, the limiting distribution of $\lambda_n$ is given by the following Fredholm determinant
\begin{equation} \label{eq:gap:larges-eigen}
\lim_{n\rightarrow\infty} \mbox{Prob}[(c n)^{\frac{2}{3}}(\lambda_n-b)<s
]=\det[I-K^{\mathrm{Ai}}_s ],
\end{equation}
where $K^{\Ai}_s$ is the trace-class operator acting on $L^2(s,\infty)$ with the Airy kernel $K^{\Ai}(x,y)$ in \eqref{Airy kernel limit} and the Fredholm determinant $\det[I-K^{\mathrm{Ai}}_s ]$
has a similar series expansion as that in \eqref{gap pro-sine}.

In \cite{TW}, Tracy and Widom discovered that the Fredholm determinant $\det[I-K^{\mathrm{Ai}}_s ]$ has a more explicit form as follows
 \begin{equation}\label{Tracy-Widom formula}
\det[I-K^{\Ai}_s]=F_{\mathrm{TW}}(s) := \exp \left( -\int_s^{+\infty}(x-s)y^2(x;0)dx \right),
 \end{equation}
where $y(x;0)$ is the Hastings-McLeod solution to the homogeneous Painlev\'{e} II ($\mathrm{P}_{2}$) equation ($\alpha = 0$)
\begin{equation}\label{Painleve II}
y''(x;\alpha)=xy(x;\alpha)+2y^3(x;\alpha)-\alpha,
\end{equation}
satisfying the following asymptotic behaviors
\begin{align}
  y(x;0)&\sim \mathrm{Ai} (x), & \textrm{as } x\rightarrow+\infty,  \label{HM solution-infty} \\
  y(x;0)&\sim\sqrt{ \frac{-x}{2}} \left(1+\frac{1}{8x^3}+O(x^{-6}) \right), & \textrm{as } x\rightarrow-\infty. \label{HM solution--infty}
\end{align}
It is remarkable that the Tracy-Widom distribution $F_{\mathrm{TW}}(s)$ appears not only in random matrices, but also in random permutations \cite{b}, totally asymmetric simple exclusion process \cite{j} and many other areas.

\subsection*{Painlev\'{e} II universality}

For the unitary ensembles \eqref{GUE}, when the limiting density function $\rho_V(x)$ vanishes quadratically at an interior point $x^*$, the \emph{Painlev\'{e} II universality} emerges; see \cite{bi,bi1,ck,ckv}. For example, consider the unitary ensemble \eqref{GUE} with the following quartic potential
\begin{equation}\label{quartic-UE}
V(x)=\frac {x^4}{4}+\frac {g}{2}x^2.
\end{equation}
When $g_{cr}=-2$, the limiting density function is
\begin{equation*}
  \rho_V(x) = \frac{x^2}{2 \pi} \sqrt{4-x^2}, \qquad \textrm{for } x \in [-2,2].
\end{equation*}
And there is a one-cut to two-cut transition near the point $x=0$ when the parameter $g$ varies in the neighbourhood of $g_{cr}$. This type of phase transition is described by the Painlev\'{e} II kernels. More precisely, if $n\to \infty$ in the way such that $2^{-1/3}n^{1/3}(g+2)\to t$,
the double scaling limit of the correlation kernel near the origin is given by
\begin{equation}\label{P2 kernel-asymptotic}
\lim_{n\to\infty}\frac {2^{2/3}}{n^{1/3}} K_n(\frac {x}{2^{-2/3}n^{1/3}},\frac {y}{2^{-2/3}n^{1/3}})=K_{\alpha}^{P2}(x,y;t),
\end{equation}
uniformly for $x$ and $y$ in compact subsets of $\mathbb{R}$; see \cite{bi1,ckv}. The limiting kernel is constructed out of the $\psi$-functions associated with the Hastings-McLeod solution to the $\textrm{P}_2$ equation \eqref{Painleve II}. The precise description of the $\mathrm{P}_{2}$ kernel will be given later.

Similar to \eqref{gap pro-sine}, once the limiting kernel is obtained, the gap probability near the origin is given as follows
\begin{equation}\label{gap pro zero}
\lim_{n\rightarrow\infty} \mbox{Prob}\left[\mbox{$M$ has no eigenvalues in }(-\frac {s}{2^{-2/3}n^{1/3}},\frac {s}{2^{-2/3}n^{1/3}})
\right]=\det[I-K^{P2}_{\alpha,s}],
\end{equation}
where $K^{P2}_{\alpha,s}$ is the trace-class operator acting on $L^2(-s,s)$ with the $\textrm{P}_2$ kernel in \eqref{P2 kernel-asymptotic}.

\subsection*{Painlev\'{e} XXXIV universality}

Now we turn to effect of the algebraic singular term $|\det M|^{2\alpha}$ in \eqref{GUE} near the soft edge. Although this singular term does not change the eigenvalue distributions in the global regime, it modifies the local eigenvalue statistics. Indeed, if there is a potential $V(x)$ such that the origin is a right regular edge point, then instead of the Airy kernel in \eqref{Airy kernel limit}, the limiting eigenvalue correlation kernel becomes the Painlev\'{e} XXXIV ($\mathrm{P}_{34}$ for short) kernel; see Its, Kuijlaars and \"{O}stensson \cite{ik1}. Later, a more general $\mathrm{P}_{34}$ kernel with two parameters was obtained in the critical situation where a Fisher-Hartwig singularity of both root and jump types appears near the soft edge of a perturbed Gaussian unitary ensemble (GUE). More precisely, the joint probability density function for the eigenvalues in this model is given by
\begin{equation}\label{pG-root type}
p_n(\lambda_1,\cdots,\lambda_n)=\frac{1}{Z_{n,\alpha,\omega}}
\prod_{i=1}^n |\lambda_i-\mu|^{2\alpha}\chi(\lambda_i-\mu)e^{- 2n\lambda_i^2}\prod_{i<j}(\lambda_i-\lambda_j)^2 \qquad \textrm{for }  \alpha>-\frac {1}{2},
\end{equation}
where $\chi(\lambda) = \begin{cases} \omega, & \lambda >0 \\ 1, & \lambda < 0 \end{cases} $  and $\omega\in \mathbb{C}\setminus(-\infty,0)$; see \cite{wxz} and \cite{bci,xz}.

When $\alpha \in \mathbb{N}$ and $\omega \in [0,1]$, it is interesting to note that the above model
 can be interpreted as a thinned and conditioned GUE. More precisely, let us consider a thinned process for the GUE by removing each eigenvalue independently with probability $\omega \in [0,1]$; see \cite{Boh:Pato2004,Boh:Pato2006}. Then, the eigenvalue distribution under the conditions that $\mu$ is an eigenvalue with multiplicity $\alpha$ in GUE and all other thinned eigenvalues are smaller than $\mu$ is given by \eqref{pG-root type}. Recently, the thinning and conditioning models have appeared in many situations. For example, gap and conditional probabilities for the thinned unitary ensembles are derived in \cite{bci,Char:Claeys2017}; the asymptotic behavior of mesoscopic fluctuations in the thinned CUE is studied in \cite{Ber:Duits2016}; the transition between the Tracy-Widom distribution and the Weibull distribution as the probability $\omega \downarrow 0$ is considered in \cite{Bot:Buck2017};  see also  \cite{Bot:Dei:Its:Kra2015} for another interesting transition. A nice application in the study of the Riemann zeros can be found in \cite{Born:Forr2017}.

Now let us consider the limiting kernel and take $n\to \infty$ in a way such that
$$\lim_{n\to \infty}2n^{2/3}(\mu_n-1)=t.$$
The double scaling limit of the correlation kernel near the soft edge $\lambda = 1$ is given by
\begin{equation}\label{kernel-asymptotic}
\lim_{n\to\infty}\frac 1{2n^{2/3}}K_n(\mu_n+\frac {x}{2n^{2/3}},\mu_n+\frac {y}{2n^{2/3}})=K_{\alpha, \omega}^{P34}(x,y;t),
\end{equation}
uniformly for $x$ and $y$ in compact subsets of $\mathbb{R}\setminus \{0\}$, with
\begin{equation}\label{psi-kernel-int}
 K_{\alpha, \omega}^{P34}(x,y;t)=\frac{\psi_2(x;t)\psi_1(y;t)-\psi_1(x;t)\psi_2(y;t)}{2\pi i(x-y)},
\end{equation}
where $(\psi_1(x;t),\psi_2(x;t))^{T}$ satisfies the Lax pair associated with the $\mathrm{P}_{34}$ equation. (The detailed information about the functions $(\psi_1(x;t),\psi_2(x;t))^{T}$ will be provided later in Section \ref{sec-p34}.) For $\alpha=0$ and $\omega = 1$, as the density function $p_n(\lambda_1,\cdots,\lambda_n)$ in \eqref{pG-root type} is reduced to that of GUE, the $\mathrm{P}_{34}$ kernel becomes the shifted Airy kernel accordingly
\begin{equation}\label{shifted Airy kernel}
K_{0,1}^{P34}(x,y;t)=\frac {\Ai(x+t)\Ai'(y+t)-\Ai'(x+t)\Ai(y+t)}{x-y};
\end{equation}
see \cite[Eq. (1.11)]{ik1}. So, the $\mathrm{P}_{34}$ kernel furnishes as a generalization of the Airy kernel. Moreover, the distribution of the largest eigenvalue in \eqref{eq:gap:larges-eigen} is replaced by
\begin{equation}\label{dist largest ev}
\lim_{n\rightarrow\infty} \mbox{Prob} [2n^{2/3}(\lambda_n-1)<s
]=\det[I-K^{P34}_{\alpha, \omega, s}],
\end{equation}
where $K^{P34}_{\alpha, \omega,s}$ is the trace-class operator acting on $L^2(s,\infty)$ with the $\textrm{P}_{34}$ kernel in \eqref{psi-kernel-int}. When $\alpha \in \mathbb{N}$ and $\omega \in [0,1]$, the above formula describe the largest eigenvalue distribution of the conditional GUE \eqref{pG-root type}.

\bigskip

In the past a few years, various Painlev\'{e} kernels have been adopted to characterize new universality classes in different critical
random matrix models; for example, see \cite{acm,cik,ck,ckv,xdz,xz-2015}. However, there are very few results about the Tracy-Widom type formulas for these kernels or their large gap asymptotics. To the best of our knowledge, the only higher-order analogues of the Tracy-Widom formula for the Fredholm determinant associated with the Painlev\'{e} I hierarchy was obtained by Claeys, Its and Krasovsky in \cite{cik2}, where the density function $\rho_V(x)$ in \eqref{circle law} vanishes with the order $2k + \frac{1}{2}, k \in \mathbb{N}$ at an endpoint of its support. Besides the large gap asymptotics in \cite{cik2}, the only other asymptotic result is obtained by Bothner and Its \cite{bothner-i} in the study of the $\mathrm{P}_{2}$ kernel in \eqref{P2 kernel-asymptotic} with the parameter $\alpha=0$. However, an analogous expression of the Tracy-Widom formula for the $\mathrm{P}_{2}$ kernel is still to be discovered.

In the present paper, we study the Fredholm determinant of  the $\textrm{P}_2$ and $\textrm{P}_{34}$ kernels with general parameters. We aim to find analogous expressions of the Tracy-Widom formula for these determinants and evaluate their large gap asymptotics.  We also plan to study the $\textrm{P}_3$ kernel at the hard edge with pole singularities in the potential in a forthcoming publication \cite{dai:xu:zhang}.


\subsection{Expressions of the $\textrm{P}_{2}$ and $\textrm{P}_{34}$ kernels}\label{sec-p34}

Before the statement of our main results, let us fist give the specific representation of the $\textrm{P}_{34}$ kernel in \eqref{psi-kernel-int}.
The $\textrm{P}_{2}$ kernel will be provided through its relation with the $\textrm{P}_{34}$ kernel.

The functions $\psi_1(x;t)$ and $\psi_2(x;t)$ in the $\textrm{P}_{34}$ kernel \eqref{psi-kernel-int} appear as solutions of the Lax pair associated with the $\textrm{P}_{34}$ equation. In addition, it is convenient to characterize this kernel in terms of the following  Riemann-Hilbert (RH) problem; see Its, Kuijlaars  and \"{O}stensson \cite{ik1}.

\begin{figure}[h]
 \begin{center}
   \includegraphics[width=6 cm]{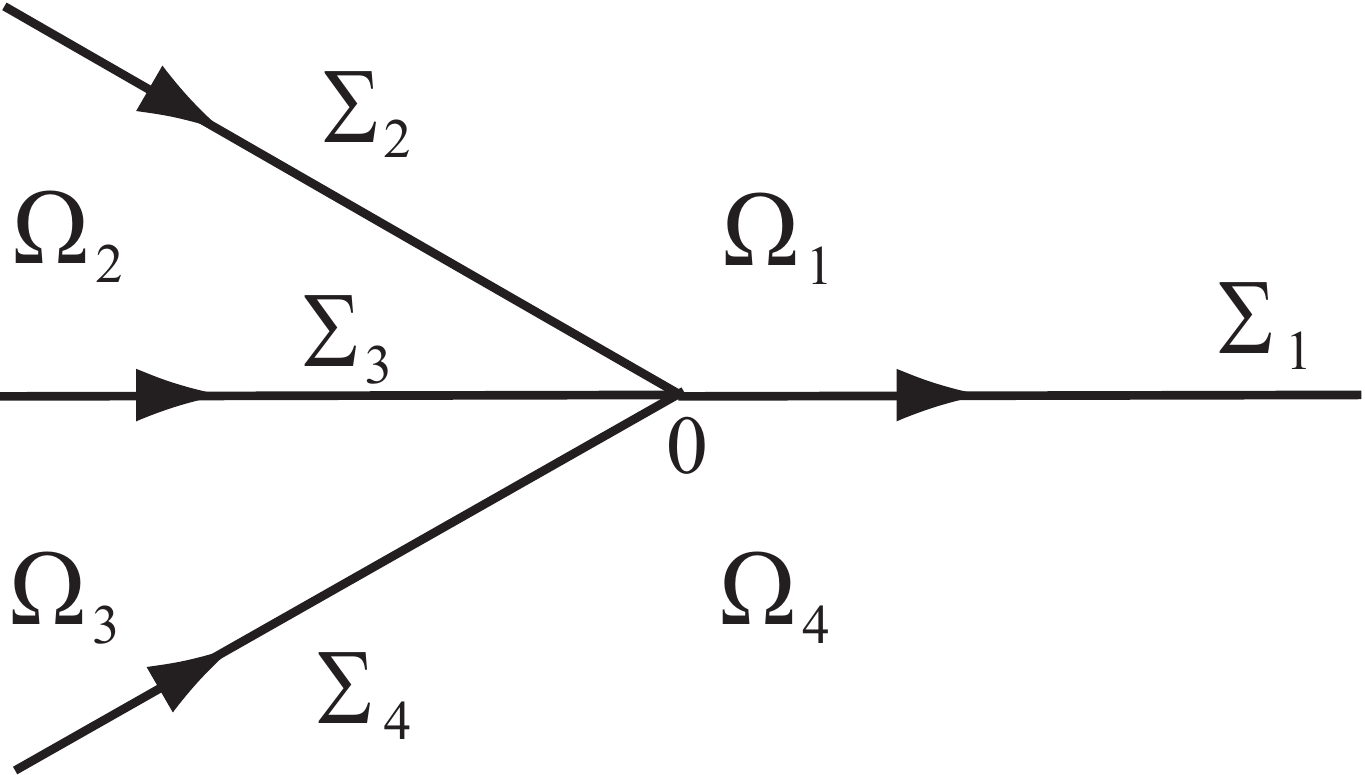} \end{center}
 \caption{\small{Contours  for the RH problem for $\Psi$: $\Sigma_1=\{\arg \zeta=0\}$, $\Sigma_2=\{\arg \zeta=\frac{2}{3}\pi\}$, $\Sigma_3=\{\arg \zeta=\pi\}$ and $\Sigma_4=\{\arg \zeta=-\frac{2}{3}\pi\}$.}}
 \label{fig-contour for psi}
\end{figure}

\medskip
\noindent\textbf{RH problem for $\Psi$:}
\begin{itemize}
  \item[(a)] $\Psi(\zeta):=\Psi(\zeta;t;\alpha,\omega)$ is analytic for $\zeta \in \mathbb{C}\backslash \cup_{i=1}^4 \Sigma_i$, where  the contours $\Sigma_i$ are depicted in Fig. \ref{fig-contour for psi};

  \item[(b)] $\Psi(\zeta)$  satisfies the following jump conditions
  \begin{eqnarray}
    \Psi_+ (\zeta)&=&\Psi_- (\zeta)
 \left(
                               \begin{array}{cc}
                                 1 & \omega \\
                                 0 & 1 \\
                                 \end{array}
                             \right),
 \qquad   \zeta \in \Sigma_1, \label{psi-alpha-jumps-1}  \\
 \Psi_+(\zeta)&=&\Psi_-(\zeta)\left(
                               \begin{array}{cc}
                                 1 & 0 \\
                                 e^{2\pi i\alpha} & 1 \\
                                 \end{array}
                             \right),  \quad  \zeta \in \Sigma_2, \label{psi-alpha-jumps-2} \\
  \Psi_+ (\zeta)&=&\Psi_- (\zeta)
 \left(
                               \begin{array}{cc}
                                 0 & 1 \\
                                 -1 & 0 \\
                                 \end{array}
                             \right),
 \qquad \zeta \in \Sigma_3 , \label{psi-alpha-jumps-3} \\
 \Psi_+(\zeta)&=&\Psi_-(\zeta)\left(
                               \begin{array}{cc}
                                 1 & 0 \\
                                 e^{-2\pi i\alpha} & 1 \\
                                 \end{array}
                             \right),  \quad  \zeta \in \Sigma_4; \label{psi-alpha-jumps-4}
  \end{eqnarray}

\item[(c)] As $\zeta \to \infty$, the asymptotic behavior of $\Psi(\zeta)$  is given by
  \begin{align}\label{psi-alpha infinity}
\Psi(\zeta)=
  \left(
    \begin{array}{cc}
      1 & 0 \\
     -im_2(t) & 1 \\
    \end{array}
  \right)
     \left [I+\frac 1{\zeta}
     \left(\begin{array}{cc} m_1(t) & im_2(t)\\ -im_3(t) & -m_1(t)\end{array} \right)+O(\zeta^{-2})\right]
   \zeta^{-\frac{1}{4}\sigma_3}\frac{I+i\sigma_1}{\sqrt{2}} e^{-\theta(\zeta,t)\sigma_3} ,
 \end{align}
  where $\theta(\zeta,t)=\frac{2}{3}\zeta^{\frac{3}{2}}+t\zeta^{\frac{1}{2}}$, $\arg \zeta \in(-\pi,\pi)$ and $m_i(t):= m_i(t;\alpha,\omega)$, $i=1,\cdots, 3$;

\item[(d)] The behavior of $\Psi(\zeta)$  at the origin is
\begin{equation}\label{psi-alpha at zero}
\Psi(\zeta)=\left(
              \begin{array}{cc}
                O(\zeta^{\alpha}) & O(\zeta^{\alpha}) \\
                O(\zeta^{\alpha}) & O(\zeta^{\alpha}) \\
              \end{array}
            \right),\quad \textrm{if } -\frac{1}{2}<\alpha<0
\end{equation}
and
\begin{equation}\label{psi-alpha at zero-1}
\Psi(\zeta)= \begin{cases}
  \left(
              \begin{array}{cc}
                O(\zeta^{\alpha}) & O(\zeta^{-\alpha}) \\
                O(\zeta^{\alpha}) & O(\zeta^{-\alpha}) \\
              \end{array}
            \right), &  \zeta\in \Omega_1\cup \Omega_4, \\
   \left(
              \begin{array}{cc}
                O(\zeta^{-\alpha}) & O(\zeta^{-\alpha}) \\
                O(\zeta^{-\alpha}) & O(\zeta^{-\alpha}) \\
              \end{array}
            \right), & \zeta\in \Omega_2\cup \Omega_3,
\end{cases}
\quad \textrm{if } \alpha\geq 0.
\end{equation}

\end{itemize}

Then, the $\textrm{P}_{34}$ kernel
\begin{equation}\label{psi-kernel}
K_{\alpha, \omega}^{P34}(x,y;t)=\frac{\psi_2(x;t)\psi_1(y;t)-\psi_1(x;t)\psi_2(y;t)}{2\pi i(x-y)},
\end{equation}
is given in terms of the functions
\begin{equation}\label{generalized airy kernel}
\left(
                                \begin{array}{c}
                                  \psi_1(x;t) \\
                                  \psi_2(x;t)\\
                                \end{array}
                              \right)
:= \begin{cases}
  \omega^{\frac12}\Psi_{+}(x; t)\left(\begin{array}{c}
                                                         1 \\
                                                         0
                                                       \end{array}\right), & \textrm{if } x>0, \\
  e^{-\alpha\pi i}\Psi_{+}(x;t)\left(\begin{array}{c}
                                                         1 \\
                                                         e^{2\alpha\pi i}
                                                       \end{array}\right), & \textrm{if } x<0.
\end{cases}
\end{equation}
Obviously, from the behavior of $\Psi(\zeta)$ at infinity in \eqref{psi-alpha infinity}, the function $(\psi_1(x;t),\psi_2(x;t))^{T}$ satisfies the following behaviours:
\begin{equation}\label{psi- kernel at positive  infinity}
\left(
                                \begin{array}{c}
                                  \psi_1(x;t) \\
                                  \psi_2(x;t)\\
                                \end{array}
                              \right)
=\frac{\omega^{\frac12}}{\sqrt{2}}e^{-(\frac{2}{3}x^{\frac{3}{2}}+t\sqrt{x})}\left(
   \begin{array}{c}
    x^{-\frac{1}{4}}+O(x^{-\frac 34})\\
     ix^{\frac{1}{4}}+O(x^{-\frac 14})\\
   \end{array}
 \right),
 \end{equation}
as $x\to +\infty$; and
\begin{equation}\label{psi- kernel at negative infinity}\left(
                                \begin{array}{c}
                                  \psi_1(x;t) \\
                                  \psi_2(x;t)\\
                                \end{array}
                              \right)
=\sqrt{2}\left(
   \begin{array}{c}
    |x|^{-\frac{1}{4}}\cos(\frac{2}{3}|x|^{\frac{3}{2}}-t\sqrt{|x|}-\alpha\pi -\pi/4)+O(|x|^{-\frac 34})\\
     -i |x|^{\frac{1}{4}}\sin(\frac{2}{3}|x|^{\frac{3}{2}}-t\sqrt{|x|}-\alpha\pi -\pi/4)+O(|x|^{-\frac 14})\\
   \end{array}
 \right),
\end{equation}
as $x\to -\infty$. When $\alpha=0,\omega=1$, the above RH problem for $\Psi$ is indeed the RH problem for the Airy functions. 
Therefore, $\textrm{P}_{34}$ kernel \eqref{psi-kernel} is reduced to the shifted Airy kernel \eqref{shifted Airy kernel}; see \cite[Eq. (1.11)]{ik1}.

Until now, we have not explained how the kernel in \eqref{psi-kernel} is related to the $\textrm{P}_{34}$ equation. The following proposition reveals their relations.¡¡

{\pro [\cite{ik1,ik2}]{\label{pro- p34} For $\alpha>-\frac 12, \omega\in \mathbb{C}\setminus (-\infty,0)$ and $t\in \mathbb{R}$, the model RH problem for $\Psi$ is uniquely solvable. In addition, let $m_2(t)$ be the function given in \eqref{psi-alpha infinity}
and
\begin{equation} \label{m2t-ut}
u(t;2\alpha,\omega) = m_2'(t) -\frac {t}{2},
\end{equation}
then $u(t;2\alpha,\omega)$ ($u(t)$ for short) satisfies the $\textrm{P}_{34}$ equation
\begin{equation}\label{int-painleve 34}
u''(t)=4u^2(t)+2tu(t) +\frac{u'(t)^2-(2\alpha)^2}{2u(t)}.
\end{equation}
Furthermore, the solution $u(t;2\alpha,\omega)$ is analytic on the real axis and uniquely determined by following asymptotic behaviors:
\begin{equation} \label{thm: painleve  positive infinity}
  \begin{split}
    u(t;2\alpha,\omega)& = \frac {\alpha}{\sqrt{t}}\left(\sum_{k=0}^n\frac {a_k}{t^{3k/2}}+O(t^{-3(n+1)/2}) \right)  \\
  & \quad +(e^{2\pi i\alpha}-\omega)\frac {\Gamma(2\alpha+1)}{2^{2+6\alpha}\pi }t^{-(3\alpha+\frac 12)}e^{-\frac 43 t^{3/2}}(1+O(t^{-1/4}))
  \end{split}
\end{equation}
as $t\to +\infty$,
with $a_0=1$ and $a_1=-\alpha$; and
\begin{equation}\label{thm: painleve negative infinity}
u(t;2\alpha,\omega)= \begin{cases}
  -\frac {t}{2}+\frac {16\alpha^2-1} {8}t^{-2}+O(t^{-7/2}), & \textrm{if } \omega = 0, \\
  \frac{1}{\sqrt{-t}}\biggl[i\beta+\frac 12 \frac{\Gamma(1+\alpha-\beta)}{\Gamma(\alpha+\beta)}
e^{i\vartheta(t;\alpha,\beta)} \\
 \qquad +\frac 12 \frac{\Gamma(1+\alpha+\beta)}{\Gamma(\alpha-\beta)}
e^{-i\vartheta(t;\alpha,\beta)}\biggr]+O(t^{3|\Re \beta|-2}), & \textrm{if } \omega=e^{-2\pi i\beta},
\end{cases}
\end{equation}
as $t\to-\infty$, with $ |\Re \beta|<1/2$ and $\vartheta(t;\alpha,\beta)=\frac 43 |t|^{3/2}-\alpha \pi-6i \beta \ln 2-3i\beta\ln|t|$.
}}

\begin{rem} \label{rmk:p34-tron}
  The $\textrm{P}_{34}$ transcendent $u(t;2\alpha,\omega)$ in the above Proposition satisfies the relation
\begin{equation} \label{p34-p2-relation}
  2^{1/3}u(-2^{-1/3}t;2\alpha, \omega)=y'(t;2\alpha+\frac {1}{2})+y^2(t;2\alpha+\frac {1}{2})+\frac {t}{2},
\end{equation}
where $y(z;2\alpha+\frac {1}{2})$ is one-parameter family of tronque\'e solutions of the $\textrm{P}_{2}$ equation \eqref{Painleve II} with the following asymptotics
\begin{equation}\label{P2-asymptotic}
y(t;\alpha)=\sqrt{\frac{-t}{2}}-\frac {\alpha}{2t} + O(t^{-5/2}) + c(\omega) (-t)^{-3\alpha-1} e^{-\frac{2 \sqrt{2}}{3} (-t)^{\frac{3}{2}}} (1+O(t^{-1/4})),
\end{equation}
$\textrm{as } t  \to \infty$ and $\arg(-t) \in (-\frac{\pi}{3}, \frac{\pi}{3}]$; see \cite[(11.5.56)]{fikn}. When $\omega = 0$,  this solution reduces to the Hastings-McLeod solution for the $\textrm{P}_{2}$ equation \eqref{Painleve II}, which is uniquely determined by the boundary conditions
\begin{equation} \label{HM solution}
  y(t;\alpha) \sim \sqrt{-t/2}, \quad t \to -\infty, \qquad  y(t;\alpha) \sim \alpha/t, \quad t \to +\infty;
\end{equation}
see \cite[Remark 11.12]{fikn}.
\end{rem}

Like the other Painlev\'e equations, the RH problem for $\Psi(\zeta)$ in \eqref{psi-alpha-jumps-1}-\eqref{psi-alpha at zero-1} implies a Lax pair for the $\textrm{P}_{34}$ equation. More precisely, we have
\begin{align}\label{Lax pair}
&\frac {\partial}{\partial\zeta}\Psi(\zeta)=\left(
      \begin{array}{cc}
        \frac {u_t}{2\zeta} & i-i\frac {u}{\zeta} \\
        -i \zeta-i(u+t)-i\frac {(u_t)^2-(2\alpha)^2}{4u \zeta} & -\frac {u_t}{2\zeta} \\
      \end{array}
    \right)\Psi(\zeta),\\\label{Lax pair-t}
&\frac {\partial}{\partial t}\Psi(\zeta)=\left(
      \begin{array}{cc}
       0 & i\\
        -i\zeta-2i(u+\frac t2) & 0 \\
      \end{array}
    \right)\Psi(\zeta),
\end{align}
where $u(t)$ satisfies the $\textrm{P}_{34}$ equation \eqref{int-painleve 34}; see \cite[Theorem 1.5]{ik1}.



One can define the $\textrm{P}_{2}$ kernel via the RH problem in a similar way. On the other hand, the $\textrm{P}_{2}$ kernel in \eqref{P2 kernel-asymptotic}
can be determined by the $\textrm{P}_{34}$ kernel with the parameter $\omega = 0$ as follows:
\begin{equation} \label{P2-kernel}
  \begin{split}
    K_{\alpha}^{P2}(x,\pm y;t)=&(xy)^{1/2}\biggl( K_{\frac{\alpha}{2}-\frac {1}{4},0}^{P34}(-2^{2/3}x^2,-2^{2/3}y^2;-2^{-1/3}t)  \\
  & \qquad \qquad \pm K_{\frac{\alpha}{2}+\frac {1}{4},0}^{P34}(-2^{2/3}x^2,-2^{2/3}y^2;-2^{-1/3}t) \biggr)
  \end{split}
\end{equation}
for $x > 0$. When $x<0$, the $\textrm{P}_{2}$ kernel is obtained from the above formula and the symmetry relation
\begin{equation}\label{P2-kernel-odd}
 K_{\alpha}^{P2}(x,y;t) = -K_{\alpha}^{P2}(-x,-y;t);
\end{equation}
see Claeys and Kuijlaars \cite{ck-2}.





\section{Statement of results}

We will derive Tracy-Widom type formulas for Fredholm determinants of the $\mathrm{P}_{2}$ and $\mathrm{P}_{34}$ kernels, as well as their large gap asymptotics.

\subsection{The coupled Painlev\'{e} II system and their asymptotics}

To express our Tracy-Widom type formulas, we need to introduce the following  coupled $\textrm{P}_2$ systems in dimension four
\renewcommand{\arraystretch}{1.3}
\begin{equation}\label{eq:Hamiltonian system}
\left\{\begin{array}{l}
         \frac {dw_1}{dx}=-\frac{\partial H}{\partial v_1} =2(v_1+v_2+\frac{x}{2})-w_1^2 \\
         \frac {dv_1}{dx}=\frac{\partial H}{\partial w_1} = 2v_1w_1\\
          \frac {dw_2}{dx}=-\frac{\partial H}{\partial v_2} =2(v_1+v_2+\frac{x-s}{2})-w_2^2 \\
         \frac {dv_2}{dx}=\frac{\partial H}{\partial w_2} = 2v_2w_2+2\alpha\\
       \end{array}
\right.
\end{equation}
where $v_i:=v_i (x;s, 2\alpha)$, $w_i:=w_i (x;s, 2\alpha+\frac{1}{2})$ and  the Hamiltonian $H:=H(v_1,v_2,w_1,w_2;x,s,2\alpha)$ is given by
\begin{equation}\label{def:Hamiltonian-CPII}
H(v_1,v_2,w_1,w_2;x,s,2\alpha)=-(v_1+v_2)^2-(v_1+v_2)x+v_1w_1^2+v_2w_2^2+sv_2+2\alpha w_2.
\end{equation}
With the transformations $v_i(x;s,2\alpha)=2^{-\frac {1}{3}}p_i(-2^{\frac {1}{3}}x;-2^{\frac {1}{3}}s,2\alpha)$ and  $w_i(x;s,2\alpha+\frac{1}{2})=-2^{\frac {1}{3}}q_i(-2^{\frac {1}{3}}x;-2^{\frac {1}{3}}s,2\alpha+\frac{1}{2})$, the above Hamiltonian is equivalent to the one studied in Sasano \cite{Sasano} via the following simple relation
\begin{equation}
  H^{Sasano}(p_1,p_2,q_1,q_2;x,s,2\alpha)=-2^{-\frac {1}{3}}H(v_1,v_2,w_1,w_2;-2^{-\frac {1}{3}}x,-2^{-\frac {1}{3}}s,2\alpha).
\end{equation}


The above coupled $\textrm{P}_2$ system \eqref{eq:Hamiltonian system} was first introduced and studied by Sasano \cite{Sasano}. It is regarded as a fourth-order extension of the classical $\textrm{P}_2$ equation. The studies of other coupled Painev\'{e} systems  in dimension four can be found in \cite{NY,S-1,S-2}.
In recent years, the program to classify the four-dimensional Painlev\'{e}-type equations has been carried out by  Kawakami, Nakamura and Sakai.
In \cite{Sakai}, from the  isomonodromic deformation theory of the Fuchsian equations, Sakai derived  four source systems for 4-dimensional Painlev\'e type equations, namely, the Garnier system in two variables, the Fuji-Suzuki system, the Sasano system and the matrix Painlev\'e system. Later, the complete degeneration scheme of these  four source systems was obtained in Kawakami, Nakamura and Sakai \cite{kns-1} and Kawakami \cite{k-1,k-2,k-3}.  The coupled $\textrm{P}_2$ system \eqref{eq:Hamiltonian system}
appears in both of the degeneration schemes of the
Garnier system  in two variables \cite[(3.5)-(3.7))]{k-3} and  the Sasano system \cite[(3.22)-(3.23)]{k-2}. Note that applications of
the coupled $\textrm{P}_2$ system in the study of  the Airy point process was discovered by Claeys and Doeraene \cite{Claeys:Doer2017} very recently.

Eliminating $w_i$ and $v_i$ from the Hamiltonian system \eqref{eq:Hamiltonian system}, respectively, gives us the following nonlinear equations for $v_i$:
\renewcommand{\arraystretch}{1.1}
\begin{equation}\label{int-equation v}
\left\{\begin{array}{l}
         v_{1xx}-\D\frac{v_{1x}^2}{2v_1} -4v_1(v_1+v_2+\frac{x}{2})=0\\
          v_{2xx}-\D\frac{v_{2x}^2-4\alpha^2}{2v_2} -4v_2(v_1+v_2+\frac{x-s}{2})=0\\
       \end{array}
\right.
\end{equation}
and equations for $w_i$:
\begin{equation}\label{int-equation-q }
\left\{\begin{array}{l}
         w_{1xx}-2w_1^3+2xw_1+4v_2(w_1-w_2)-(4\alpha+1)=0 \\
         w_{2xx}-2w_2^3+2(x-s)w_2-4v_1(w_1-w_2)-(4\alpha+1)=0
       \end{array}
\right.
\end{equation}
The coupled equations \eqref{int-equation v} are similar to \cite[(2.1)]{Hone2001}, which was obtained from similarity reduction of the Hirota-Satsuma system. Eliminating either one of the functions $v_1$ or $v_2$ from the above equation, one gets a fourth-order nonlinear differential equation.
Besides, if we set $v_1(x)=y^2(x)$, then the first equation becomes
\begin{equation}\label{GpII}
y''-2y^3-xy=2v_2y.
\end{equation}
When $\alpha=0$, taking the admissible solution $v_2=0$ in \eqref{int-equation v}, then \eqref{GpII} is
reduced to the standard  $\mathrm{P}_{2}$ equation \eqref{Painleve II}.


On the way to our Tracy-Widom type formulas, we need a class of distinguished solutions to the couple $\textrm{P}_2$ system. The solutions satisfy the following properties.


{\thm{\label{theorem-v-1}For $\alpha>-\frac 12$, $\omega\in \mathbb{C}\setminus (-\infty,0)$ and $s \in \mathbb{R}$, there exist real analytic solutions $v_i(x)$ to the coupled $\textrm{P}_{2}$ equations \eqref{int-equation v}.  
Moreover, the solutions $v_i(x,s;2\alpha,\omega)$ satisfy the following asymptotic behaviors:
\begin{equation}\label{int-v-1-asy-+infty}
v_1(x,s;2\alpha,\omega)=\frac {1}{4\pi \sqrt{x}}e^{-\frac {4}{3}x^{\frac {3}{2}}}\frac {c}{2^{4\alpha}}\left|\frac {s}{x}\right|^{2\alpha}
\left(1+\alpha\frac {s}{x}+O\left(x^{-3/2}\right)\right), \quad \textrm{as } x\to+\infty,
\end{equation}
with the constant $c= \begin{cases}
  \omega, & s>0, \\
  1, & s<0;
\end{cases}$ and
\begin{equation}\label{int-v-2-asy-+infty}
v_2(x,s;2\alpha,\omega)=\frac{\alpha}{\sqrt{x-s}}-\frac{\alpha^2}{(x-s)^2}+O(1/x^3), \quad \textrm{as } x\to+\infty. 
\end{equation}
In addition, the functions $w_i(x,s;2\alpha + \frac{1}{2},\omega)$ and the Hamiltonian $H(x;s):=H(v_1,v_2,w_1,w_2;x,s,\alpha)$ in \eqref{eq:Hamiltonian system} satisfy the following asymptotic behaviors:
\begin{equation}\label{eq:asy-w-1}
w_1(x,s;2\alpha+ \frac{1}{2},\omega)=-\sqrt{x}+O(\ln x), \quad \textrm{as } x\to+\infty,
\end{equation}
\begin{equation}\label{eq:asy-w-2}
w_2(x,s;2\alpha+ \frac{1}{2},\omega)=-\sqrt{x-s}-\frac{\alpha+\frac{1}{4}}{x-s}+O(1/x^2), \quad \textrm{as } x\to+\infty,
\end{equation}
\begin{equation}\label{eq:asy-H}
H(x;s)=-2\alpha\sqrt{x-s}-\frac{\alpha^2}{x-s}+O(1/x^2), \quad \textrm{as } x\to+\infty.
\end{equation}
}}

\begin{rem} \label{remark:p2HM}
For $\alpha=0$ and $\omega=1$, then $v_2=0$ and the equation \eqref{GpII} is
reduced to the $\textrm{P}_{2}$   equation \eqref{Painleve II}.
From the asymptotic behaviors in \eqref{int-v-1-asy-+infty}, one immediately sees that $v_1(x)=y^2(x;0)$,  where y(x;0) is the
classical Hastings-McLeod solution to the $\textrm{P}_2$ equation given in \eqref{HM solution-infty}-\eqref{HM solution--infty}.
\end{rem}

{\rem{When $s=0$, we have $v_1\equiv0$ and $v_2(x,0;2\alpha,\omega)=u(x;2\alpha,0)$, where $u(x;2\alpha,0)$ is the solution to the $\textrm{P}_{34}$ equation and satisfies the  asymptotic behaviors given in Proposition \ref{pro- p34}.}}



\subsection{Tracy-Widom type expressions for the Fredholm determinants}

Now we have the following integral representations for the Fredholm determinants of the $\textrm{P}_{2}$ and $\textrm{P}_{34}$ kernels.

\begin{thm}\label{thm:TW-P34}
For $\alpha>-\frac 12$, $\omega\in \mathbb{C}\setminus (-\infty,0)$ and $s,t \in \mathbb{R}$, let $K^{P34}_{\alpha, \omega,s}$ be the trace-class operator acting on $L^2(s,\infty)$ with the $\textrm{P}_{34}$ kernel $K_{\alpha,\omega}^{P34}(x,y;t)$ in \eqref{psi-kernel}, then we have
\begin{equation}\label{int-TW-p34-1}
\det[I-K^{P34}_{\alpha, \omega, s}] = \exp\left(-\int_{t}^{+\infty}(v_1(x+s)+v_2(x+s)-u(x))(x-t)dx\right),
\end{equation}
where $v_i(x)=v_i(x,s;2\alpha,\omega)$ are the smooth solutions to the coupled $\textrm{P}_{2}$ equations \eqref{int-equation v} with the properties specified in Theorem \ref{theorem-v-1} and $u(x)=u(x;2\alpha,\omega)$ is the $\textrm{P}_{34}$ transcendent given in Proposition \ref{pro- p34}.
\end{thm}

%
%

\begin{rem}
For $\alpha=0$ and $\omega=1$, we recover the celebrated Tracy-Widom distribution \eqref{Tracy-Widom formula} from Theorem \ref{theorem-v-1} and \ref{thm:TW-P34}. Indeed, from the properties of the functions $v_i(x)=v_i(x,s;2\alpha,\omega)$ described in Remark \ref{remark:p2HM}, the integral representation \eqref{int-TW-p34-1} becomes
\begin{equation}
\det[I-K^{\Ai}_{s+t}]=\exp\left(-\int_{s+t}^{+\infty}(\tau-t-s) y^2(\tau;0)d\tau \right),
\end{equation}
where $K^{\Ai}_{s+t}$ is the trace-class operator acting on $L^2(s+t,\infty)$ with the Airy kernel in \eqref{Airy kernel limit}
and $y(x;0)$ is the Hastings-McLeod solution to the $\textrm{P}_{2}$ equation given in \eqref{HM solution-infty}-\eqref{HM solution--infty}.
\end{rem}

\begin{rem}
  When $\alpha \in \mathbb{N}$ and $\omega \in [0,1]$, the Tracy-Widom type formula \eqref{int-TW-p34-1} is the largest eigenvalue distribution of the conditional GUE in \eqref{pG-root type}. When $\alpha = 0$, it agrees with
   the results obtained by Claeys and Doeraene in \cite[(2.8)]{Claeys:Doer2017}. Our result in \eqref{int-TW-p34-1} holds for general parameter $\alpha>-\frac 12$, which enables us to establish the Tracy-Widom type formula for the $\textrm{P}_2$ kernel below.
\end{rem}

Next, we establish a relation between the Fredholm determinants for the $\textrm{P}_{2}$ and $\textrm{P}_{34}$ kernels.

\begin{lem} \label{lem:kp2-kp34}
  Let $K^{P2}_{\alpha,s}$ and $K^{P34}_{\alpha,0,s'}$ be the  trace-class operator acting on $L^2(-s,s)$ and $L^2(s',+\infty)$ with the $P_2$ kernel $K_{\alpha}^{P2}(x,y;t)$ in \eqref{P2 kernel-asymptotic} and $\textrm{P}_{34}$ kernel  $ K^{P34}_{\frac {\alpha}{2}\pm\frac 14, 0}(x,y; -2^{-1/3}t)$ in  \eqref{psi-kernel}, respectively, we have
\begin{equation}\label{Fredholm det-P2-P34}
\det[I-K^{P2}_{\alpha,s}]=\det[I-K_{\frac {\alpha}{2}+\frac {1}{4},0,s'}^{P34}]\det[I-K_{\frac {\alpha}{2}-\frac {1}{4},0,s'}^{P34}],
\end{equation}
where $s'=-2^{\frac{2}{3}}s^2$, $s>0$ and $\alpha>-\frac 12$.
\end{lem}

Then, Theorem \ref{thm:TW-P34} and the above lemma gives us the Tracy-Widom type formula for the Fredholm determinant of the $\textrm{P}_{2}$ kernel.

{\thm{\label{thm:TW-P2}For $\alpha>-\frac 12$, $s>0$ and $t \in \mathbb{R}$,  let  $K^{P2}_{\alpha,s}$ be the trace-class operator
acting on $L^2(-s,s)$ with the $\textrm{P}_{2}$ kernel $K_{\alpha}^{P2}(x,y;t)$ given in \eqref{P2-kernel}, then we have
\begin{equation}\label{int-TW-pII}
 \det[I-K^{P2}_{\alpha,s}]= \exp\left(-\int^{t}_{-\infty}(y^2(x;\alpha)-2^{-2/3} w_2^2(-2^{-1/3} x+s',s';\alpha) )(x-t)dx\right),
\end{equation}
where  $s'=-2^{2/3}s^2$,  $y(x;\alpha)$ is the Hastings-McLeod solution to the  $\textrm{P}_{2}$ equation described in \eqref{HM solution}, and
$w_2(x,s;\alpha)$ is the smooth solution to the equation \eqref{int-equation-q } with the asymptotic behaviors given in \eqref{eq:asy-w-2}.}}



\subsection {Large gap asymptotics  for the Fredholm determinants }

From \eqref{psi- kernel at positive  infinity}, it is easy to see that the $\textrm{P}_{34}$ kernel \eqref{psi-kernel} satisfies the following asymptotic behavior as $x,y\to +\infty$
$$K_{\alpha,\omega}^{P34}(x,y)=O(e^{-c(x^{\frac {3}{2}}+y^{\frac {3}{2}})}),$$
for certain constant $c>0$. From the series expansion of the Fredholm determinant (for example, see the expansion for the sine kernel in \eqref{FD-expansiion}), we obtain
$$\ln \det[I-K^{P34}_{\alpha,\omega,s}]=O(e^{-cs^{\frac {3}{2}}}), \quad  \mbox{as} \quad s\to+\infty,$$
where $K^{P34}_{\alpha,\omega,s}$ is  the trace-class operator with kernel $K_{\alpha,\omega}^{P34}$ acting on $L^2(s,\infty)$.
The large gap asymptotics for the Fredholm determinant of the $\textrm{P}_{34}$ kernel as $s\to -\infty$
are much more involved and given in the following theorem.

{\thm{\label{theorem-large gap asy-P34}For $\alpha>-\frac 12$, $\omega\in \mathbb{C}\setminus (-\infty,0)$ and  $t\in \mathbb{R}$,
let $K^{P34}_{\alpha,\omega,s}$ be the trace-class operator  acting on $L^2(s,\infty)$ with the kernel  $K_{\alpha,\omega}^{P34}(x,y;t)$ in \eqref{psi-kernel}, we have the asymptotic
 expansion for the Fredholm determinant as $s\to -\infty$
 \begin{align}
   \ln \det[I-K^{P34}_{\alpha,\omega,s}]=&-\frac 1{12}|s+t|^3 +\frac 23 \alpha|s|^{\frac 32}-2\alpha|s|^{\frac{1}{2}}t-(\alpha^2+\frac {1}{8}) \ln|s+t| \label{lag gap asy} \\
 &+\frac 43 \alpha \, \mathrm{sgn}(t) |t|^{\frac 32}  +\alpha^2\ln|t|+\int_t^{+\infty}(\tau-t)\left(u(\tau)-\frac {\alpha}{|\tau|^{\frac{1}{2}}}+\frac {\alpha^2}{\tau^2}\right )d\tau +c_0 +o(1), \nonumber
 \end{align}
where $u(x)=u(x;2\alpha,\omega)$ is the $\textrm{P}_{34}$ transcendent given in Proposition \ref{pro- p34}. The constant $c_0$ in the above formula is given explicitly as
\begin{equation}\label{def:constant-0}
c_0=\frac{1}{24} \ln 2+\zeta'(-1).
\end{equation}
where $\zeta'(z)$ is the derivative of the Riemann zeta-function. }}

\begin{rem}
  For $\alpha=0$, we recover the large gap asymptotics for the Fredholm determinant associated with the Airy kernel as $s+t\to -\infty$
\begin{equation}\label{large gap asy-airy}
\ln \det[I-K^{\Ai}_{s+t}]=-\frac 1{12}|s+t|^3-\frac 1{8}\ln|s+t|+ c_0 +o(1),
\end{equation}
where $K^{\Ai}_{s+t}$ is the trace-class operator with the Airy kernel \eqref{Airy kernel limit} acting on $L^2(s+t,+\infty)$ and $c_0$ is given in \eqref{def:constant-0}; see \cite{DIK2008,TW}.
\end{rem}

We also have the large gap asymptotics for the Fredholm determinant of the $\textrm{P}_{2}$ kernel.
{\thm{\label{theorem-large gap asy-P2}For $\alpha>-\frac 12$,
let $K^{P2}_{\alpha,s}$ be the trace-class operator  acting on $L^2(-s,s)$ with the kernel  $K_{\alpha}^{P2}(x,y;t)$ given in \eqref{P2-kernel},  we have the asymptotic  expansion for the Fredholm determinant as $s\to +\infty$
 \begin{equation}\label{lag gap asy-P2}
 \begin{split}
 \ln \det[I-K^{P2}_{\alpha,s}]&=-\frac {2}{3}(s^2+\frac {t}{2})^3 +\frac 43 \alpha s^{3}+2\alpha st-(\alpha^2+\frac {3}{4}) \ln s +\frac{2\sqrt{2}}{3}\alpha  \, \mathrm{sgn}(t) |t|^{\frac{3}{2}} \\
 &+(\frac{\alpha^2}{2}+\frac{1}{8})\ln|t| -\int^{t}_{-\infty}(\tau-t) \left( y^2(\tau;\alpha) +\frac{\tau}{2}-\frac{\alpha}{\sqrt{|2\tau|}}+\frac{\frac{\alpha^2}{2}+\frac{1}{8}}{\tau^2}\right)d\tau+c_1+o(1),
 \end{split}
 \end{equation}
where $y(x;\alpha)$ is the Hastings-McLeod solution to the $\textrm{P}_{2}$ equation described in \eqref{HM solution} and  the constant $c_1$ is given explicitly as
\begin{equation} \label{def:constant-1}
c_1=-(\alpha^2+\frac{5}{24}) \ln 2+2\zeta'(-1).
\end{equation}
}}

{\rem{For $\alpha=0$, the above large gap asymptotics can be reduced to that in 
Bothner and Its \cite{bothner-i}
\begin{equation}\label{lag gap asy-P2-0}
 \ln \det[I-K^{P2}_{0,s}]=-\frac {2}{3} s^6-s^4t-\frac {1}{2}(st)^2 -\frac {3}{4} \ln s+\int_{t}^{+\infty}(\tau-t)y^2(\tau;0)d\tau+ c_2 +o(1),
 \end{equation}
where $y(x;0) $ is the Hastings-McLeod solution to the $\textrm{P}_{2}$ equation and the constant term
$$ c_2=-\frac {1}{6} \ln 2+3\zeta'(-1).$$
In order to reduce \eqref{lag gap asy-P2} to \eqref{lag gap asy-P2-0}, one needs the following total integral of the Hastings-McLeod solution $y(x;0)$ to the $\textrm{P}_{2}$ equation
\begin{equation}\label{eq: total integral pii-0}
\int_t^{\infty}(\tau-t)y^2(\tau;0)d\tau+\int_{-\infty}^{t}(\tau-t)(y^2(
\tau;0)+\frac{\tau}{2}+\frac{1}{8\tau^2})d\tau=-\frac{t^3}{12}+\frac{1}{8}\ln|t|-c_0,
\end{equation}
where $c_0$ is defined in \eqref{def:constant-0}. Note that the above formula directly follows from \eqref{Tracy-Widom formula} and \eqref{large gap asy-airy}.
}}

The rest of the paper is organized as follows. In Section \ref{sec:RHP for Fredholm det}, we provide a representation of the Fredholm determinant in terms of a RH problem and derive a differential identity for the Fredholm determinant. Then, we transform this RH problem to a model one to facilitate our future study. In  Section \ref{Sec:Lax Pair}, a Lax pair is derived from the model RH problem obtained in the Section \ref{sec:RHP for Fredholm det}. The compatibility condition of the Lax pair is described by
 the coupled $\textrm{P}_2$ equations. The B\"{a}cklund transformations for the coupled $\textrm{P}_2$ system are also studied.
In Section \ref{sec:asy of v1-infty},  we  study the asymptotic behaviors of
  the functions in the Hamiltonian system \eqref{eq:Hamiltonian system} as $x \to + \infty$ and prove Theorem \ref{theorem-v-1}. Section \ref{sec: tw} is then devoted to the proof of Theorems \ref{thm:TW-P34}-\ref{thm:TW-P2}. In Section \ref{sec:large gap asy}, we evaluate the large gap asymptotics for the Fredholm determinants. Finally, for the sake of completeness and possible interests, the asymptotics of the functions $v_i(x)$ in Theorem \ref{theorem-v-1} as $x \to - \infty$ are derived in Appendix \ref{sec:asy of v1--infty}.

\section{Riemann-Hilbert problem for the Fredholm determinant and differential identity}\label{sec:RHP for Fredholm det}

\subsection{Riemann-Hilbert problem for the Fredholm determinant and differential identity}

 Let  $K^{P34}_{\alpha,\omega,s}$ be the trace-class operator   acting on $L^2(s,\infty)$  with the kernel  $K_{\alpha,\omega}^{P34}$ in \eqref{psi-kernel}, the Fredholm determinant $\det[I-K^{P34}_{\alpha,\omega,s}]$  can be characterized  in terms of the solution of certain RH problems; see Deift, Its and Zhou \cite{diz} and Its et al. \cite{Its:Ize:Koe:Sla1990}.

\medskip
\noindent\textbf{RH problem for $Y$:}
\begin{itemize}
  \item[(a)]   $Y(z)$ is analytic for $ z \in \mathbb{C}\backslash (s, \infty)$;

  \item[(b)]   $Y(z)$  satisfies the jump condition on $(s, \infty)$
  \begin{equation}\label{Y-jump}
  Y_+(x)=Y_-(x)J_{Y}(x),\quad J_{Y}(x)=\left(
                  \begin{array}{cc}
                    1+\psi_1(x)\psi_2(x) &- \psi_1(x)^2 \\
                    \psi_2(x)^2 & 1-\psi_1(x)\psi_2(x) \\
                  \end{array}
                \right),
  \end{equation}
  where $\psi_i(x)$ are defined in \eqref{generalized airy kernel}
  \item[(c)]   The asymptotic behavior of $Y(z)$  at infinity:
  \begin{equation}\label{Y at infinity}
  Y(z)=I+ \frac{Y_{-1}}{z}+O(1/z^2) \qquad \textrm{as } z \rightarrow \infty ;
  \end{equation}

\item[(d)] At possible endpoints $z=0$ and $z = s$, $Y (z)$ is square integrable.
\end{itemize}

{\lem[\cite{diz}]{Let  $K^{P34}_{\alpha,\omega,s}$ be the trace-class operator   acting on $L^2(s,\infty)$  with the kernel  $K_{\alpha,\omega}^{P34}$ in \eqref{psi-kernel} and assume $(I-K^{P34}_{\alpha,\omega,s})^{-1}$ exists, then the solution for the above RH problem is given by
 \begin{equation}\label{Y solution}
  Y(z)=I-\frac{1}{2\pi i}\int_s^{+\infty} \frac{M(x)}{x-z} dx, \quad M(x)=\left(
                                                                      \begin{array}{cc}
                                                                        -F_1(x)\psi_2(x) & F_1(x)\psi_1(x) \\
                                                                        -F_2(x)\psi_2(x) & F_2(x)\psi_1(x) \\
                                                                      \end{array}
                                                                    \right),\end{equation}
  where $\psi_i(x)$ are defined in \eqref{generalized airy kernel} and $F_i(x)=(I-K^{P34}_{\alpha,\omega,s})^{-1}\psi_i(x)$. Conversely, the functions $F_i(x)$ can be expressed in terms of $Y$ as follows
\begin{equation}\label{relation Y+ and F}
 \left(
               \begin{array}{c}
               F_1(x) \\
                F_2(x) \\
               \end{array}
             \right)= Y_+(x)\left(\begin{array}{c}
          \psi_1(x) \\
          \psi_2(x)
        \end{array}\right),\quad x>s, x\neq 0.
\end{equation}}}

From the above RH problem for $Y$, we derive a differential identity for the Fredholm determinant of the $\textrm{P}_{34}$ kernel as follows.
\begin{pro}
For $\alpha>-\frac 12$, $\omega\in \mathbb{C}\setminus (-\infty,0)$ and $s,t \in \mathbb{R}$, we have
\begin{equation}\label{ logrithmic derivaive t related to Y}
\frac{d}{dt}\ln \det[I-K^{P34}_{\alpha,\omega,s}]=i(Y_{-1})_{12},
\end{equation}
where $Y_{-1}$ is the coefficient of the $1/z$ term as $z \to \infty$; cf. \eqref{Y at infinity}.
\end{pro}
\begin{proof}
From \eqref{psi-kernel} and \eqref{Lax pair-t}, we have
\begin{equation}\label{D t of kernel}
\frac{d}{dt}K_{\alpha,\omega}^{P34}(x,y;t)=-\frac {1}{2\pi}\psi_1(x)\psi_1(y).
\end{equation}
Using properties of trace-class operators, we get
\begin{equation}
\frac{d}{dt}\ln \det[I-K^{P34}_{\alpha,\omega,s}]=-\mathrm{tr}((I-K^{P34}_{\alpha,\omega,s})^{-1}\frac{d K^{P34}_{\alpha,\omega,s}}{dt})=\frac {1}{2\pi}\int_s^{\infty}((I-K^{P34}_{\alpha,\omega,s})^{-1}\psi_1)(x)\psi_1(x)dx.
\end{equation}
Then, \eqref{ logrithmic derivaive t related to Y} follows from \eqref{Y solution} and \eqref{relation Y+ and F}.
\end{proof}

\subsection{Model Riemann-Hilbert problem for $\Phi$} \label{sec:rhp-phi}

Next, we transform the original RH problem for $Y$ into a new one with constant jumps. Observe that the jump matrix $J_{Y}(x)$ in \eqref{Y-jump} can be factorized as follows:
\begin{equation}\label{Y-jumps related to psi-alpha-1 }
J_{Y}(x) =\Psi_{+}(x)\left(
                                       \begin{array}{cc}
                                         1 & -\omega \\
                                        0 & 1 \\
                                       \end{array}
                                     \right)
\Psi_{+}(x)^{-1} \qquad \textrm{for } x>0
\end{equation}
and
\begin{equation}\label{Y-jumps related to psi-alpha-2}
J_{Y}(x) =\Psi_{+}(x)\left(\begin{array}{cc}
                                                             1 & 0 \\
                                                             e^{2\pi i \alpha} & 1
                                                           \end{array}\right)
\left(
  \begin{array}{cc}
    1 & -e^{-2\pi i \alpha} \\
    0 & 1\\
  \end{array}
\right)
\left(\begin{array}{cc}
                                                             1 & 0 \\
                                                             -e^{2\pi i \alpha} & 1
                                                           \end{array}\right)\Psi_{+}(x)^{-1} \ \ \textrm{for } x<0.
\end{equation}
This evokes us to introduce the following transformation
\begin{equation}\label{U-1}
    \widetilde{\Phi}(z)= Y(z)\Psi(z), \qquad \textrm{if } s\geq 0
\end{equation}
and
\begin{equation}\label{U}
    \widetilde{\Phi}(z)= \begin{cases}
     Y(z)\Psi(z), & z\in \Omega_1^R\cup \Omega_2\cup\Omega_3\cup\Omega_4^R \\
     Y(z) \Psi(z)\left(\begin{array}{cc}
                                                             1 & 0 \\
                                                             e^{2\pi i \alpha} & 1 \end{array}\right), & z\in \Omega_1^L \\
     Y(z) \Psi(z)\left(\begin{array}{cc}
                                                             1 & 0 \\
                                                             -e^{-2\pi i \alpha} & 1
                                                           \end{array}\right), &   z\in \Omega_4^L,
    \end{cases}, \quad \textrm{if } s<0,
\end{equation}
where the regions are indicated in Fig. \ref{figure-U}.
\begin{figure}[ht]
 \begin{center}
   \includegraphics[width=6cm]{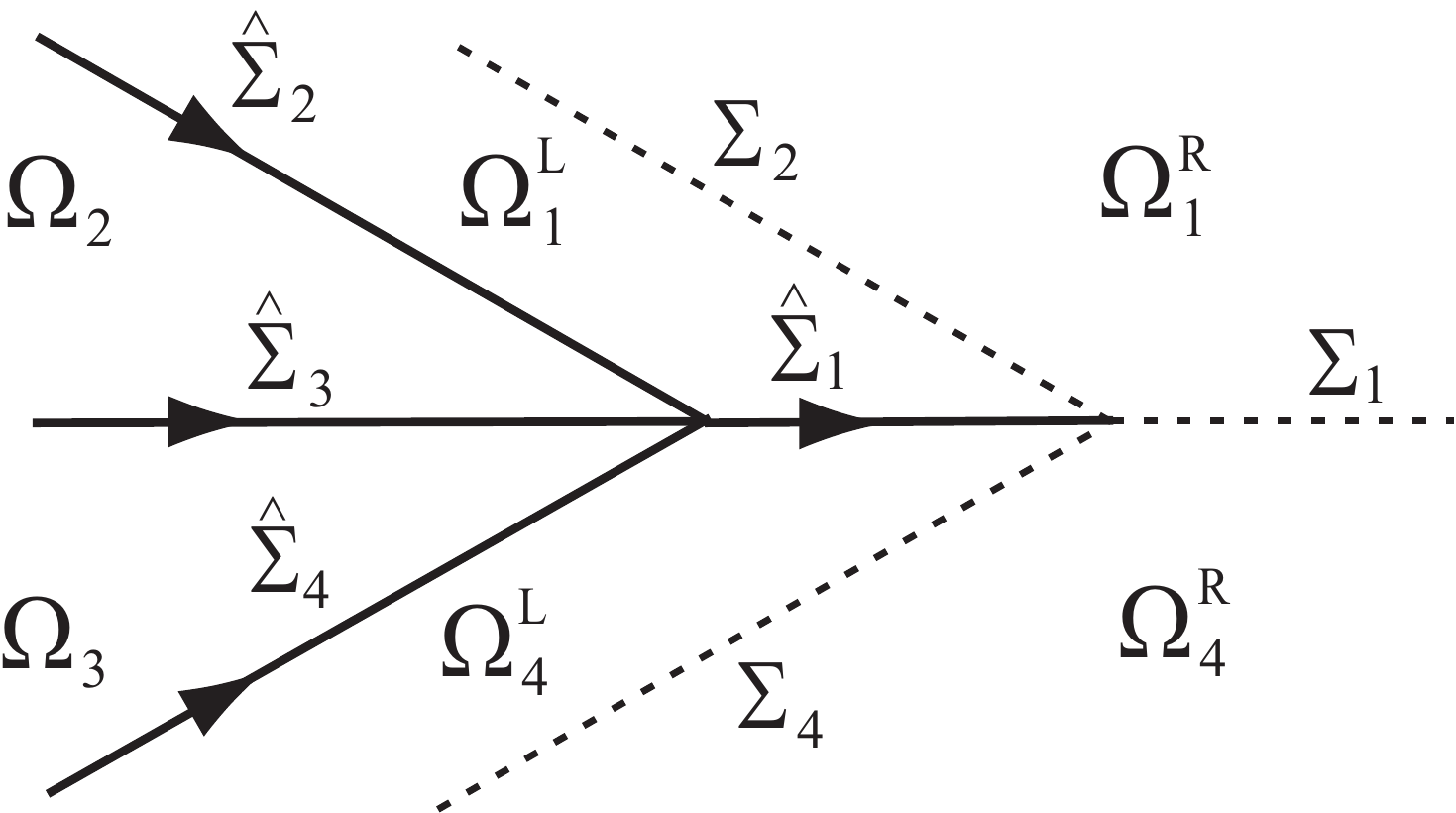} \end{center}
 \caption{\small{Regions and contours  for  $\widetilde{\Phi}$ ($\hat{\Sigma}_1=[0,s]$ for $s>0$ and $\hat{\Sigma}_1=[s,0]$ for $s<0$).}}
 \label{figure-U}
 \end{figure}
Then, $\widetilde{\Phi}$ satisfies a RH problem as follows.
\medskip
\noindent\textbf{RH problem for $\widetilde{\Phi}$:}
\begin{itemize}
  \item[(a)]   $\widetilde{\Phi}(z):=\widetilde{\Phi}(z;x,s;\alpha,\omega)$ is analytic for  $ z \in \mathbb{C}\backslash\hat{ \Sigma}_{i}$; see Fig. \ref{figure-U};

  \item[(b)]   $\widetilde{\Phi}(z)$  satisfies the jump condition
  \begin{equation}\label{U -jumps}
  \widetilde{\Phi}_+(z)=\widetilde{\Phi}_-(z) J_{i}(z), \qquad \textrm{for } z\in \hat{\Sigma}_{i}
  \end{equation}
with
$$J_{1}(z)=\left(\begin{array}{cc}
                                                              e^{2\pi i \alpha} & 0 \\
                                                            0 &  e^{-2\pi i \alpha}
                                                           \end{array}\right),\quad \mbox{if} \quad  s<0,
                                                           \quad J_{1}(z)=\left(\begin{array}{cc}
                                                              1 & \omega \\
                                                            0 &  1
                                                           \end{array}\right), \quad \mbox{if} \quad s>0, $$
$$ J_{2}(z)=\left(\begin{array}{cc}
                                                             1 & 0 \\
                                                             e^{2\pi i \alpha} & 1
                                                           \end{array}\right),
                                                            \quad J_{3}(z)=\left(\begin{array}{cc}
                                                             0 & 1 \\
                                                            -1 & 0
                                                           \end{array}\right),
                                                           \quad  J_{4}(z)=\left(\begin{array}{cc}
                                                             1 & 0 \\
                                                             e^{-2\pi i \alpha} & 1
                                                           \end{array}\right);$$

  \item[(c)]   The asymptotic behavior of $\widetilde{\Phi}(z)$  at infinity:
  \begin{equation}\label{U at infinity}
  \widetilde{\Phi}(z)= \left(
    \begin{array}{cc}
      1 & 0 \\
     -im_2(t) & 1 \\
    \end{array}
  \right)
    \left (I+O\left (\frac 1 {z}\right )\right)
  z^{-\frac{1}{4}\sigma_3}\frac{I+i\sigma_1}{\sqrt{2}}
  e^{-\theta(z,t)\sigma_3}, \qquad \textrm{as } z \rightarrow \infty,
  \end{equation}
  where   $\theta(z,t)=\frac{2}{3}z^{\frac{3}{2}}+tz^{\frac{1}{2}}, \ \arg z\in(-\pi, \pi)$;

\item[(d)] The asymptotic behavior of $\widetilde{\Phi}(z)$  at the node point $s$:
\begin{equation}\label{U at s}
\widetilde{\Phi}(z)=\widetilde{\Phi}^{(s)}(z) \begin{cases}
 \left(
                                           \begin{array}{cc}
                                             1 &\frac{\omega\ln(z-s)}{2\pi i} \\
                                             0 & 1 \\
                                           \end{array}
                                         \right) ,  \quad \textrm{if } s>0 \\
  \left(
                                           \begin{array}{cc}
                                             1 &\frac {\ln(z-s)}{2\pi i} \\
                                             0 & 1 \\
                                           \end{array}
                                         \right)C, \quad \textrm{if } s<0
\end{cases}
 \qquad \textrm{as } z \rightarrow s,
\end{equation}
where $\widetilde{\Phi}^{(s)}(z)=\widetilde{\Phi}^{(s)}_0(I+\widetilde{\Phi}^{(s)}_1(z-s)+ \cdots)$ is analytic at $z=s$ and the constant matrix $C$ is
\begin{equation} \label{matrix-c-formula}
  C=\left\{\begin{array}{cc}
                 (I-\sigma_-) e^{\pi i\alpha\sigma_3} & z\in \Omega_2 \\
                  (I+\sigma_-)  e^{-\pi i\alpha\sigma_3}& z\in \Omega_3 \\
                 e^{\pi i\alpha\sigma_3} & z\in \Omega_1^L\\
                e^{-\pi i\alpha\sigma_3} & z\in \Omega_4^L
               \end{array}
\right.;
\end{equation}

\item[(e)] The asymptotic behavior of $\widetilde{\Phi}(z)$  at $z=0$:
\begin{equation}\label{U at zero}
\widetilde{\Phi}(z)=O(1)\Psi(z), \ s>0, \quad \widetilde{\Phi}(z)=\widetilde{\Phi}^{(0)}(z)z^{\alpha\sigma_3}, \  s<0, \qquad \textrm{as } z \rightarrow 0,
\end{equation}
where $\widetilde{\Phi}^{(0)}(z)$ is analytic at $z=0$.
\end{itemize}

To facilitate our future derivation of the associated Lax pair, let us introduce one more transformation:
\begin{equation}\label{U to Phi}
\Phi(z)=\left(
         \begin{array}{cc}
           1 & 0 \\
          \eta & 1 \\
         \end{array}
       \right)
\widetilde{\Phi}(s+z).
\end{equation}
Then, $\Phi(z)$ solves the following model RH problem.

\medskip
\noindent\textbf{RH problem for $\Phi$:}
\begin{itemize}
  \item[(a)]   $\Phi(z):=\Phi(z;x,s;\alpha,\omega)$ is analytic for $z \in \mathbb{C}\backslash \hat{\Sigma}_{i}$; see Fig. \ref{figure-Phi-model};

  \begin{figure}[h]
 \begin{center}
   \includegraphics[width=5cm]{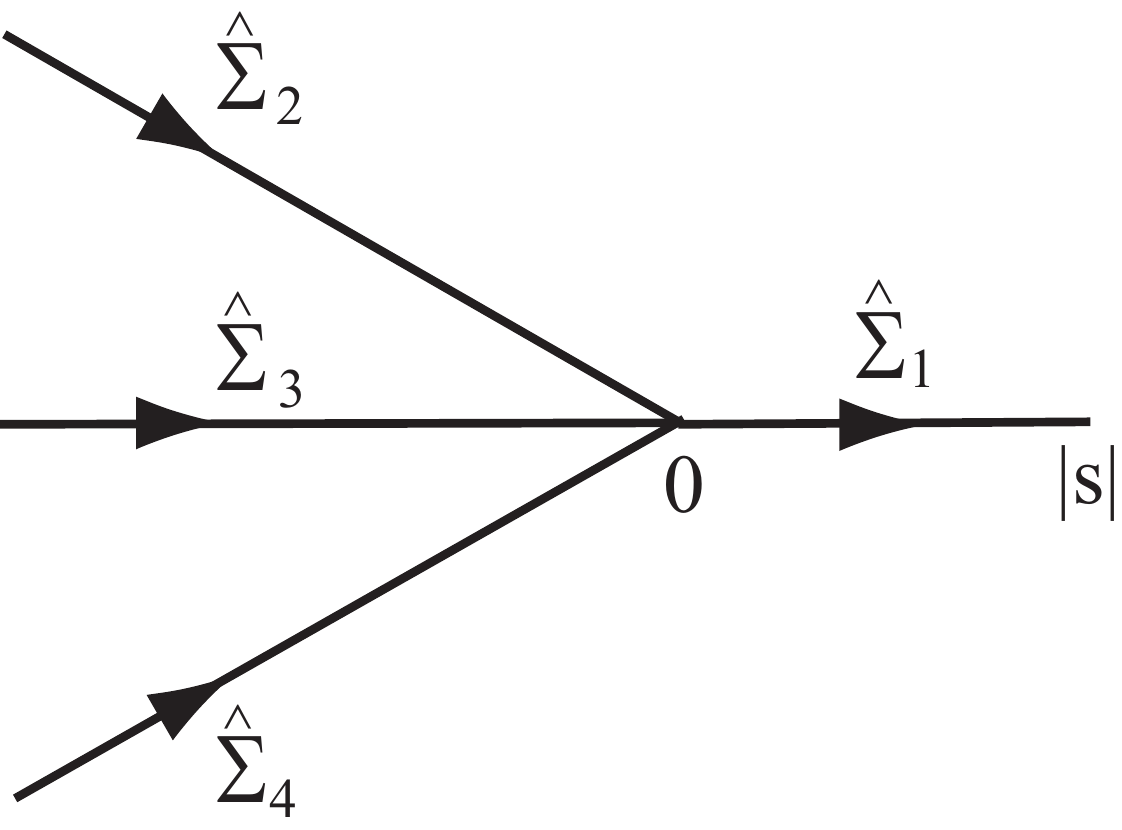} \qquad \includegraphics[width=5cm]{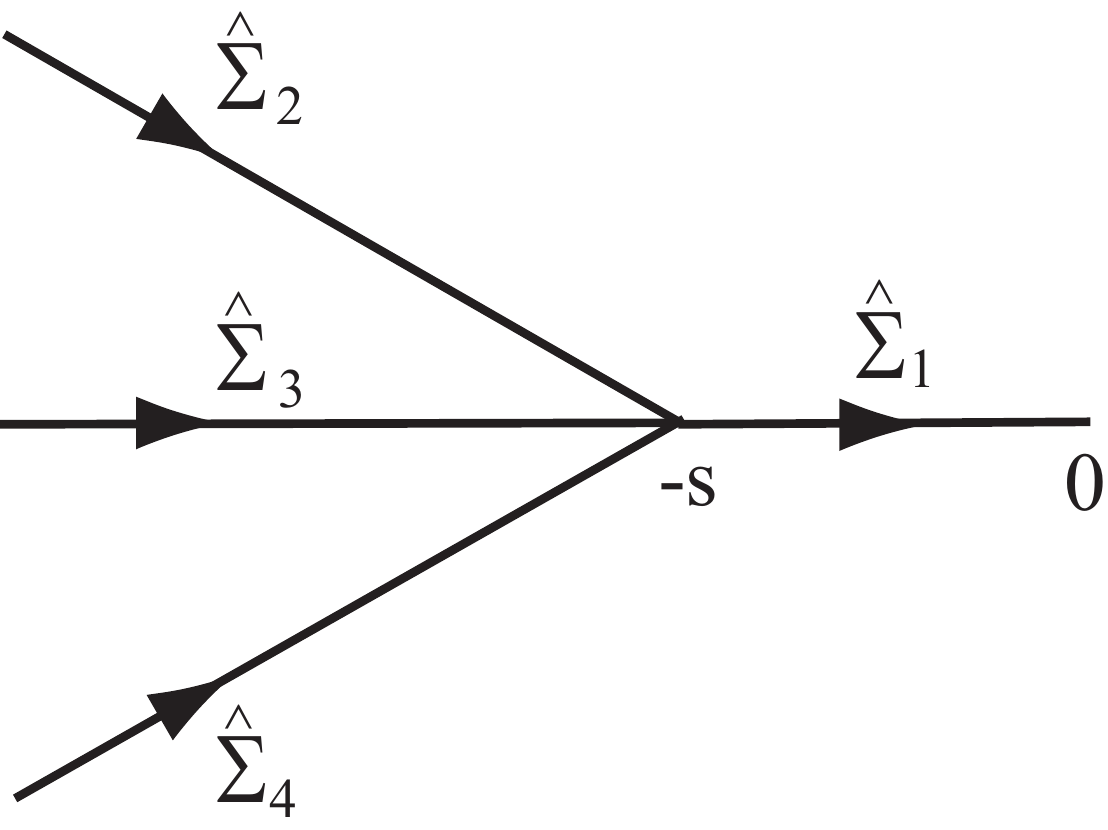} \end{center}
 \caption{\small{The contour  for  $\Phi$ (left: $s<0$; right: $s>0$).}} \label{figure-Phi-model}
  \end{figure}

  \item[(b)]   $\Phi(z)$  satisfies the same jump conditions as $\widetilde{\Phi}(s+z)$ on $\hat{\Sigma}_{i}$;

  \item[(c)] The asymptotic behavior of $\Phi(z)$  at infinity:
\begin{equation}\label{Phi at infinity}
\Phi(z)=\left(
         \begin{array}{cc}
           1 & 0 \\
           -ir_2 & 1 \\
         \end{array}
       \right)
    \left (I+\frac{\Phi_{-1}}{z} + \frac{\Phi_{-2}}{z^2}  + \frac{\Phi_{-3}}{z^3} +O\left (\frac 1 {z^4}\right )\right)
  z^{-\frac{1}{4}\sigma_3}\frac{I+i\sigma_1}{\sqrt{2}}
  e^{-\theta(z,x)\sigma_3} ,
  \end{equation}
  where  $\theta(z,x)=\frac{2}{3}z^{\frac{3}{2}}+x z^{\frac{1}{2}}, \ \arg z\in(-\pi, \pi)$ and
  \begin{equation}\label{Phi at infinity Phi-1}
  \Phi_{-1}=
    \left(
      \begin{array}{cc}
        r_1(x) & ir_2(x) \\
       -ir_3(x) & -r_1(x) \\
      \end{array}
    \right), \qquad  \Phi_{-2}=
    \left(
      \begin{array}{cc}
        k_1(x) & ik_2(x) \\
       -ik_3(x) & -k_1(x) \\
      \end{array}
    \right),
  \end{equation}
 $x=t+s$ and the pre-factor in \eqref{U to Phi} is chosen to be $\eta=im_2(t)-ir_2+\frac {s(2x-s)}{4}i $  to simplify the differential equations for $\Phi$;

 \item[(d)] The asymptotic behavior of $\Phi(z)$  at $z=0$:
\begin{equation}\label{Phi at 0}
\Phi(z)=\Phi^{(0)}(z) \begin{cases}
 \left(
                                           \begin{array}{cc}
                                             1 &\frac{\omega\ln z}{2\pi i} \\
                                             0 & 1 \\
                                           \end{array}
                                         \right) ,  \quad \textrm{if } s>0 \\
  \left(
                                           \begin{array}{cc}
                                             1 &\frac {\ln z}{2\pi i} \\
                                             0 & 1 \\
                                           \end{array}
                                         \right)C, \quad \textrm{if } s<0
\end{cases}
 \qquad \textrm{as } z \rightarrow 0,
\end{equation}
where $\Phi^{(0)}=\Phi^{(0)}_0(I+\Phi^{(0)}_1z+\cdots)$ is analytic at $z=0$ and the constant matrix $C$ is defined in \eqref{matrix-c-formula};

\item[(e)] The asymptotic behavior of $\Phi(z)$  as $z \rightarrow -s$
\begin{equation} \label{Phi at -s}
  \begin{split}
  & \hspace{-0.2cm} \Phi(z)=O(1)\Psi(z+s) \qquad \textrm{for } s>0,  \\
  & \hspace{-0.2cm}  \Phi(z)=\Phi^{(-s)}_0(s) \biggl(I+\Phi^{(-s)}_1 (s ) \,(z+s)+O(z+s)^2 \biggr)(z+s)^{\alpha\sigma_3} \quad \textrm{for } s<0,
  \end{split}
\end{equation}
where the behavior of $\Psi$ near origin is given in \eqref{psi-alpha at zero} and \eqref{psi-alpha at zero-1}.
\end{itemize}

For $s=0$, the jump conditions for $\Phi(z;t,0;\alpha,\omega)$ in \eqref{U to Phi} is the same as that of $\Psi(z;t;\alpha,0)$. From \eqref{psi-alpha at zero} and \eqref{Y solution}, one can see that $\Phi(z;t,0;\alpha,\omega) \Psi(z;t;\alpha,0)^{-1}=Y(z)\Psi(z;t;\alpha,\omega)\Psi(z;t;\alpha,0)^{-1}$ is analytic in the complex plane with a removable singularity at the origin and tends to the identity matrix as $z \to \infty$. Thus, we have the following remark.

\begin{rem}
  For $s=0$, we have
$$
\Phi(z;t,0;\alpha,\omega)=\Psi(z;t;\alpha,0),
$$
where $\Psi(z;t; \alpha,\omega)$ is the solution to the RH problem associated with the $\textrm{P}_{34}$ equation given in Section \ref{sec-p34}.
\end{rem}


\section{Coupled Painlev\'e II equations} \label{Sec:Lax Pair}

In this section, we derive a Lax pair from the model RH problem for $\Phi$, which is a Garnier system in two variables. The compatibility condition of the Lax pair gives us the coupled $\textrm{P}_2$ equations. We also show an important relation between the Hamiltonian $H$ in \eqref{eq:Hamiltonian system} and the RH problem for $\Phi$. The solvability of the model RH problem for $\Phi$ and the existence of real analytic solution to the $\textrm{P}_2$ equations are  justified. The B\"{a}cklund transformations for the coupled $\textrm{P}_2$ system are also studied.

\subsection{Lax pair}

We derive a  Lax pair  from the RH problem for $\Phi$. It furnishes as a
generalization of the Lax pair for the $\textrm{P}_{34}$ equation given in \eqref{Lax pair}.
\begin{pro}\label{Lax pair for Phi}
We have the  following Garnier system in two variables
\begin{equation}\label{Pro:Lax pair-Phi-z}
 \Phi_z(z;x,s)=\left(
       \begin{array}{cc}
       \frac {v_{1x}}{2z}+\frac {v_{2x}}{2(z+s)} & i-\frac {iv_1}{z}-\frac {iv_2}{z+s} \\
        -i(z+x+v_1+v_2+ \frac {v_{1x}^2}{4v_1z} +\frac {v_{2x}^2-4\alpha^2}{4v_2(z+s)})& -\frac {v_{1x}}{2z}-\frac {v_{2x}}{2(z+s)} \\
       \end{array}
     \right)\Phi(z;x,s),
 \end{equation}
 \begin{equation}\label{Pro:Lax pair-Phi-x}
\Phi_{x}(z;x,s)=\left(
      \begin{array}{cc}
        0 & i \\
        -iz-2i(v_1+v_2+\frac {x}{2}) & 0 \\
      \end{array}
    \right)
\Phi(z;x,s),
\end{equation}
and
\begin{equation}\label{Pro:Lax pair-Phi-s}
\Phi_{s}(z;x,s)=\left(
      \begin{array}{cc}
        \frac {v_{2x}}{2(z+s)} &-\frac {iv_2}{z+s}  \\
       iv_2-i \frac {v_{2x}^2-4\alpha^2}{4v_2(z+s)}& -\frac {v_{2x}}{2(z+s)} \\
      \end{array}
    \right)
\Phi(z;x,s).
\end{equation}
Moreover, the compatibility conditions $\Phi_{zx}=\Phi_{x z}$ and $\Phi_{xs}=\Phi_{s x}$ gives us the couple $\textrm{P}_2$ equations \eqref{int-equation v} and
\begin{equation} \label{compatibility2}
  v_{2x} = -( v_{1s} + v_{2s}),
\end{equation}
respectively.
\end{pro}
\begin{proof}
Due to the fact that the jump matrices for $\Phi$ are all constants, then $\Phi_{z}\Phi^{-1}$, $\Phi_{x}\Phi^{-1}$ and $\Phi_{s}\Phi^{-1}$ are meromorphic functions
with possible isolated singular points at $-s$ and $0$. Using the behavior of  $\Phi$ at infinity, $-s$ and $0$, we have
\begin{equation}\label{Lax pair Phi-z}
\Phi_z=\left(
      \begin{array}{cc}
        a(z;x,s) &  b(z;x,s) \\
         c(z;x,s) & - a(z;x,s) \\
      \end{array}
    \right)
\Phi,
\end{equation}
\begin{equation}\label{Lax pair Phi-tao}
\Phi_{x}=\left(
      \begin{array}{cc}
        0 & i \\
        -iz+2ir_1-ir_2^2-ir_{2x} & 0 \\
      \end{array}
    \right)
\Phi,
\end{equation}
and
\begin{equation}\label{Lax pair Phi-s}
\Phi_{s}=\left(
      \begin{array}{cc}
        \frac {a_2(x;s)}{z+s} & -\frac {iv_2(x;s)}{z+s} \\
        -ir_{2s}+\frac {c_2(x;s)}{z+s} & -\frac {a_2(x;s)}{z+s} \\
      \end{array}
    \right)
\Phi,
\end{equation}
where
\begin{align}
\label{a} a(z;x,s)&=\frac {a_1(x;s)}z+\frac {a_2(x;s)}{z+s} \\
\label{b}  b(z;x,s)&=i-\frac {iv_1(x;s)}z-\frac {iv_2(x;s)}{z+s} \\
\label{c}  c(z;x,s)&=-iz-\frac {x i}{2}+2ir_1-ir_2^2+\frac {c_1(x;s)}z+\frac {c_2(x;s)}{z+s}.
\end{align}
Making use of the large-$z$ expansion of $\Phi(z)$ in \eqref{Phi at infinity}, we have from \eqref{Lax pair Phi-tao}
$$\frac {1}{z}\left(\frac d {dx}\Phi_{-1}+i[\Phi_{-1},\sigma_+]-i[\sigma_-\Phi_{-1},\Phi_{-1}]-i[\Phi_{-2},\sigma_-]\right)+O(1/z^2)=0,$$
where $\sigma_{\pm}$ are the constant matrices given below
\begin{equation} \label{sigma+-matrix}
  \sigma_+ = \begin{pmatrix}
    0 & 1 \\ 0 & 0
  \end{pmatrix} \quad \textrm{and} \quad \sigma_- = \begin{pmatrix}
    0 & 0 \\ 1 & 0
  \end{pmatrix}.
\end{equation}
Particularly, the $(1,2)$ entry of the above matrix equation gives us the useful relation
\begin{equation}\label{p-q relation}
 r_1=\frac 12(-r_{2x}+r_2^2).
\end{equation}
Similarly, substituting large-$z$ expansion of $\Phi(z)$ into \eqref{Lax pair Phi-s}, we obtain
\begin{equation}\label{r-2s-v relation}
 r_{2s}=-v_2.
\end{equation}
With the aid of the above two formulas, the equation \eqref{Lax pair Phi-tao} and the expression for $c(z;x,s)$ in \eqref{c} are simplified to
\begin{equation}\label{Lax pair Phi-tao-simplified}
\Phi_{x}=\left(
      \begin{array}{cc}
        0 & i \\
        -iz-2ir_{2x} & 0 \\
      \end{array}
    \right)
\Phi,
\end{equation}
and
\begin{equation}\label{c-simplify}
c(z;x,s)=-iz-\frac {x i}{2}-ir_{2x}+\frac {c_1(x;s)}z+\frac {c_2(x;s)}{z+s}.
\end{equation}

Let us consider the equations \eqref{Lax pair Phi-z} and \eqref{Lax pair Phi-tao}. Their compatibility condition $\Phi_{zx}=\Phi_{x z}$
gives us the following equations:
\begin{align} \label{equation a}a&=\frac i2b_{x}\\
\label{equation c}  c&=-(z+2r_{2x})b+\frac 12 b_{xx}\\
\label{equation c-tau} c_{x}&=-i-2(iz+2ir_{2x})a.
\end{align}
Moreover, from the behavior of $\Phi$ at $0$, $-s$ in \eqref{Phi at 0} and \eqref{Phi at -s}, we have
\begin{equation}\label{relation a-1, b-1,c-1}\det \left(
    \begin{array}{cc}
      a_1 & -iv_1 \\
      c_1 & -a_1 \\
    \end{array}
  \right)
 =0 \quad \textrm{and} \quad
 \det \left(
\begin{array}{cc}
     a_2 & -iv_2 \\
     c_2 & -a_2 \\
   \end{array}
 \right) =-\alpha^2.
\end{equation}
Substituting  \eqref{b} and \eqref{c-simplify} into \eqref{equation c} yields
\begin{equation}\label{equation r and v}
r_{2x}=v_1+v_2+\frac {x}{2}.
\end{equation}
By \eqref{a}, \eqref{r-2s-v relation} and \eqref{relation a-1, b-1,c-1}--\eqref{equation r and v}, we get the  Lax pair \eqref{Pro:Lax pair-Phi-z}--\eqref{Pro:Lax pair-Phi-s}. And the coupled $\textrm{P}_2$ equations \eqref{int-equation v} can be derived
from \eqref{equation a}, \eqref{equation c} and \eqref{relation a-1, b-1,c-1}. Finally, by the compatibility condition $\Phi_{x,s} = \Phi_{s,x}$ and formulas \eqref{r-2s-v relation} and \eqref{equation r and v}, we obtain \eqref{compatibility2}.
\end{proof}




Note that the Lax pair system for $\Phi$ in \eqref{Lax pair Phi-z}-\eqref{Lax pair Phi-s} is equivalent to the one
in \cite[(3.5)-(3.7)]{k-3} by elementary transformation.
The Lax pair for a general couple $\textrm{P}_2$ system with $k$ regular singular points and one irregular singular point can be found in \cite{Claeys:Doer2017}.

The Hamiltonian $H$ in \eqref{eq:Hamiltonian system} will play an important role in the derivation of the large gap asymptotics in Section \ref{sec:large gap asy}. Besides its definition in \eqref{def:Hamiltonian-CPII}, it also has the following simple relation with the RH problem for $\Phi$.
\begin{pro} \label{prop-Hamilton}
  Let $H$ be the Hamiltonian defined in \eqref{def:Hamiltonian-CPII}. We have
 \begin{equation}\label{eq:Hamiltotion -r-2}
    H=\frac{x^2}{4}-r_2,
 \end{equation}
 where $r_2(x;s) = -i(\Phi_{-1})_{12}$ and $\Phi_{-1}$ is the coefficient of $1/z$ term in the large-$z$ expansion of $\Phi(z)$; see \eqref{Phi at infinity Phi-1}.
\end{pro}
\begin{proof}
  From \eqref{Pro:Lax pair-Phi-z}, we have
  \begin{eqnarray}
    \Phi_z(z)\Phi^{-1}(z) & = & -iz \sigma_- + \begin{pmatrix}
      0 & i \\
      -i(x+v_1+v_2) & 0
    \end{pmatrix}  + \frac{1}{z} \begin{pmatrix}
      \frac{v_{1x} + v_{2x}}{2} & -i(v_1+v_2) \\ \frac{v_{1x}^2}{4iv_1} + \frac{v_{2x}^2 - 4 \alpha^2}{4iv_2}  & - \frac{v_{1x} + v_{2x}}{2}
    \end{pmatrix}  \nonumber \\
    & & + \frac{1}{z^2} \begin{pmatrix}
      -\frac{s v_{2x}}{2} & is v_2 \\ is \frac{v_{2x}^2 - 4 \alpha^2}{4v_2}  & \frac{s v_{2x}}{2}
    \end{pmatrix} + O(\frac{1}{z^3}),
  \end{eqnarray}
  as $z \to \infty$. On the other hand, from the large-$z$ expansion of $\Phi$ in \eqref{Phi at infinity}, we get
  \begin{align}\label{eq:expand-z}
\Phi_z(z)\Phi^{-1}(z)&=(I -ir_2\sigma_{-})\big\{-iz\sigma_{-}-i[\Phi_{-1},\sigma_{-}] -\frac{i}{2}x\sigma_{-}+i\sigma_{+}\\\nonumber
                      & + \frac{1}{z}(i[\Phi_{-1},-\frac{1}{2}x\sigma_{-}+\sigma_{+}+\sigma_{-}\Phi_{-1}]-i[\Phi_{-2},\sigma_{-}]+\frac{i}{2}x\sigma_{+}-\frac{1}{4}\sigma_{3})  \\\nonumber
                      & +\frac{1}{z^2} \big(i[\Phi_{-1},\frac{x}{2}\sigma_{+}-\sigma_{+}\Phi_{-1}+\frac{x}{2}\sigma_{-}\Phi_{-1}+\sigma_{-}\Phi_{-2}-\sigma_{-}\Phi_{-1}^2+\frac{i}{4}\sigma_{3}]\\\nonumber
                      & +i[\Phi_{-2},-\frac{1}{2}x\sigma_{-}+\sigma_{+}+\sigma_{-}\Phi_{-1}] -i[\Phi_{-3},\sigma_{-}]-\Phi_{-1}\big)+O(1/z^3) \big\}(I +ir_2\sigma_{-}),
                       \end{align}
where the matrices $\sigma_\pm$ are given in \eqref{sigma+-matrix}. Comparing the above two formulas, one obtains the following system of equations involving $v_i$, $r_i$ and $k_i$ (the entries of $\Phi_{-1}$ and $\Phi_{-2}$; see \eqref{Phi at infinity Phi-1}):
  \begin{equation}\label{eq:expand-z-equations}
\left\{\begin{array}{l}
         r_2^2-2r_1=v_1+v_2+\frac{x}{2}, \\
        r_1r_2-\frac{x}{2}r_2-r_2(v_1+v_2)+r_3+\frac{1}{2}(v_{1x}+v_{2x}) -k_2+\frac{1}{4}=0, \\
        r_2^2(v_1+v_2)-r_2(v_{1x}+v_{2x})-r_2r_3-2r_1^2+xr_1+ 2k_1+\frac{v_{1x}^2}{4v_1} +\frac{v_{2x}^2-4\alpha^2}{4v_2}=0, \\
        \frac{x}{2}r_2^2-r_2r_3+2r_2k_2+\frac{1}{2}r_2-2r_1^2-xr_1-2k_1+sv_2=0.
       \end{array}
\right.
\end{equation}
Note that, in the above system, the first equation comes from the $O(1)$ term; the second and third one are from the $O(1/z)$ term; and the last one comes from the (1,2) entry of the $O(1/z^2)$ term.
Eliminating the functions $k_1$, $k_2$ and $r_1$, we obtain
\begin{equation}\label{eq:r_2}
r_2=(v_1+v_2+\frac{x}{2})^2-\frac{v_{1x}^2}{4v_1} -\frac{v_{2x}^2-4\alpha^2}{4v_2}-sv_2=\frac{x^2}{4} - H,
\end{equation}
where the definition of $H$ is given in \eqref{def:Hamiltonian-CPII}. This finishes the proof of our proposition.
\end{proof}

For later use, we derive two differential identities for the functions $v_2$ and $w_2$.
\begin{pro}\label{lem: integrals of w, v}
For $s<0$ and $z\to s$, with the behavior of $\Phi(z)$ given in \eqref{Phi at -s}, we have
\begin{equation}\label{eq:integral w-2}
\frac{d}{dx}\ln \left(\Phi^{(-s)}_0\right)_{11}(x)= w_2(x)
\end{equation}
and
\begin{equation}\label{eq:integral-v-2}
2\alpha \frac{d}{dx} \left(\Phi^{(-s)}_1\right)_{11}(x)=-v_2(x).
\end{equation}
\end{pro}
\begin{proof}
From \eqref{Pro:Lax pair-Phi-x} in the Lax pair and \eqref{Phi at -s}, we have
\begin{equation}\label{eq:deff-Phi-0}
\frac{d}{dx}\Phi^{(-s)}_0=\left(
                            \begin{array}{cc}
                              0 & i \\
                              -2i(v_1+v_2+\frac{x-s}{2})& 0 \\
                            \end{array}
                          \right)\Phi^{(-s)}_0,
\end{equation}
\begin{equation}\label{eq:deff-Phi-1}
\frac{d}{dx}\Phi^{(-s)}_1=\left(\Phi^{(-s)}_0\right)^{-1}\left(
                            \begin{array}{cc}
                              0 & 0 \\
                              -i& 0 \\
                            \end{array}
                          \right)\Phi^{(-s)}_0,
\end{equation}
and
\begin{equation}\label{eq:Phi-0-v}
\alpha\Phi^{(-s)}_0\sigma_3\left(\Phi^{(-s)}_0\right)^{-1}=\left(
                            \begin{array}{cc}
                              \frac{1}{2}v_2'& -iv_2 \\
                              \frac{v_2'^2-4\alpha^2}{4iv_2}& -\frac{1}{2}v_2' \\
                            \end{array}
                          \right).
\end{equation}
Then, the above three equations give us
\begin{equation}\label{eq:phi-0}
\frac {d^2}{dx^2}\left(\Phi^{(-s)}_0\right)_{11}=2(v_1+v_2+\frac{x-s}{2})\left(\Phi^{(-s)}_0\right)_{11}=(w_2^2+w_{2x})\left(\Phi^{(-s)}_0\right)_{11}
\end{equation}
and
\begin{equation}\label{eq:phi-1}
\frac {d}{dx}\left(\Phi^{(-s)}_1\right)_{11}=-\frac{1}{2\alpha}v_2.
\end{equation}
Thus, \eqref{eq:integral w-2} and \eqref{eq:integral-v-2} immediately follow from the above two equations.
\end{proof}

\subsection{Vanishing lemma}

\begin{lem}\label{vanishing lemma}
Suppose that the homogeneous RH problem for $\widehat{\Phi}(z)$ shares the same
jump conditions  and  boundary  behaviors  near $0$ and $-s$ as $\Phi(z)$, but satisfies the following  asymptotic behavior at infinity
$$
\widehat{\Phi}(z)=O(\frac 1z)z^{-\frac{1}{4}\sigma_3}\frac{I+i\sigma_1}{\sqrt{2}}
  e^{-(\frac 23 z^{3/2}+x z^{1/2})\sigma_3} \qquad \textrm{as } z \rightarrow \infty .
$$
Then, for $\alpha>-\frac 12$, $\omega\in \mathbb{C}\setminus (-\infty,0)$ and $s, x \in \mathbb{R}$,
the solution is trivial, that is $\widehat{\Phi}(z)\equiv 0$.
\end{lem}

\begin{proof}The proof is similar to \cite{ik1} and \cite[Lemma 1]{xz}. The key point is  that the entries in the jump matrices $J_i$ in \eqref{U -jumps} satisfy the following conjugate relations $(J_1)_{11} = \overline{(J_1)_{22}}$ and $(J_2)_{12} = \overline{(J_4)_{12}}$.
\end{proof}

\begin{cor} \label{Analytic of v-1}
For $\alpha>-\frac 12$, $\omega\in \mathbb{C}\setminus (-\infty,0)$ and $s, x \in \mathbb{R}$, there is unique solution to the RH problem for $\Phi$. And there exist real analytic solutions $v_1(x; s) $ and $ v_2(x;s)$
to the coupled $\textrm{P}_{2}$ equations \eqref{int-equation v} for real values of $x$. Moreover, the Hamiltonian $H$ in \eqref{eq:Hamiltotion -r-2} is also real analytic for real values of $x$.
\end{cor}

\begin{proof}
The solvability of the RH problem for $\Phi$ follows from  Lemma \ref{vanishing lemma}, namely the  vanishing Lemma; see the similar arguments in \cite{dkmv1,fmz,fz,z}. Then, the functions $r_i(x)$ in \eqref{Phi at infinity Phi-1} are all analytic for real values $x$. Therefore, the Hamiltonian $H$ in \eqref{eq:Hamiltotion -r-2} is also real analytic for real values of $x$. Similarly, from \eqref{Pro:Lax pair-Phi-z}, the functions $v_i(x)$ can be expressed as
\begin{eqnarray}
  v_1(x)&=&i\lim_{z\to 0}z(\Phi'(z)\Phi(z)^{-1})_{12},  \label{v1-Phi'} \\
  v_2(x)&=&i\lim_{z\to -s}(z+s)(\Phi'(z)\Phi(z)^{-1})_{12}, \label{v2-Phi'}
\end{eqnarray}
where $'$ indicates the derivative with respect to $z$. Then, the analyticity of  $v_1(x; s)$ and $ v_2(x;s)$ also follows from the solvability of the RH problem for $\Phi$.

Moreover, once the solution exists, it is easy to prove its uniqueness with the boundary conditions given in the RH problem for $\Phi$. Let $ \Upsilon:= \left(
         \begin{array}{cc}
           1 & 0 \\
           ir_2 & 1 \\
         \end{array}
       \right)$ and note that $\sigma_3\overline{\Upsilon\Phi(\bar{z})}\sigma_3$ also satisfies the RH problem for $\Upsilon\Phi$. From the uniqueness of the RH problem, we have
\begin{equation} \label{phi-phi-bar}
  \sigma_3\overline{\Upsilon\Phi(\bar{z})}\sigma_3= \Upsilon \Phi(z).
\end{equation}
Finally, the above formula ensures that $v_1(x; s), v_2(x;s)$  are real for real values of $x$.
\end{proof}

\subsection{B\"{a}cklund transformations}

From the model RH problem for $\Phi$, we have the following useful B\"{a}cklund transformations for the coupled $\textrm{P}_2$ system. Similar arguments work for the Painlev\'e equations; for example, see \cite{fikn} and \cite[Sec. 6]{dk}.
\begin{pro}\label{Pro:BT} Let $v_i(x;2\alpha+1)$ and $v_i(x;2\alpha)$ be the solutions to the coupled $\textrm{P}_2$ system given in Theorem \ref{theorem-v-1}, we have
\begin{equation}\label{BT-1}
v_1(x;2\alpha+1)+v_2(x;2\alpha+1)-v_1(x;2\alpha)-v_2(x;2\alpha)= -w_2'(x;2\alpha+\frac{1}{2}),
\end{equation}
\begin{equation}\label{BT-2}
w_2(x;2\alpha+\frac{3}{2})+w_2(x;2\alpha+\frac{1}{2})=-\frac{2\alpha+1}{v_2(x;2\alpha+1)} ,
\end{equation}
where 
$'$ indicates derivative with respect to $x$.
\end{pro}

\begin{proof}
From the RH problem for $\Phi$, we see that $\Phi(z;2\alpha+1)$ and $\sqrt{z+s}~\sigma_3\Phi(z;2\alpha)\sigma_3$ satisfy the same jump conditions.
According to the asymptotic behaviors of $\Phi$ at infinity and the singular points $z=0, -s$, we find that $\Phi(z;2\alpha)(\sqrt{z+s}~\sigma_3\Phi(z;2\alpha)\sigma_3)^{-1}$ is meromorphic with a simple pole at $z=-s$.
Moreover, using the behaviors of $\Phi$ at infinity in \eqref{Phi at infinity}, we have
\begin{align}
  \Phi(z;2\alpha+1) = & \left(
                           \begin{array}{cc}
                             1 & 0\\
                            -ir_2(x;2\alpha+1)& 1 \\
                           \end{array}
                         \right)
\left(
        \begin{array}{cc}
          -r_2(x;2\alpha+1)& i
           \\
          i(z+s)-i(r_1(x;2\alpha+1)-r_1(x;2\alpha)) & -r_2(x;2\alpha) \\
        \end{array}
      \right) \nonumber \\
&  \left(    \begin{array}{cc}
                             1 & 0\\
                            -ir_2(x;2\alpha)& 1 \\
                           \end{array}    \right) \frac{\sigma_3\Phi(z;2\alpha)\sigma_3}{\sqrt{z+s}}, \label{Phi-shift-1}
\end{align}
 where $r_i(x)$ are given in \eqref{Phi at infinity Phi-1}. Similarly, the local behavior of $\Phi$ at $z=-s$ gives us
 \begin{equation}\label{Phi-shift-2}
\lim_{z\to-s}\sqrt{z+s}\Phi(z;2\alpha+1)\sigma_3\Phi(z;2\alpha)^{-1}\sigma_3=\breve{\Phi}(-s;2\alpha+1)\left(
                                                                                        \begin{array}{cc}
                                                                                          0 & 0 \\
                                                                                          0 & 1\\
                                                                                        \end{array}
                                                                                      \right)\sigma_3 \breve{\Phi}(-s;2\alpha)^{-1}\sigma_3,
\end{equation}
where $\breve{\Phi}$ is the series part of the expansion of $\Phi$ near $z=-s$ in \eqref{Phi at -s}. From the equation \eqref{Pro:Lax pair-Phi-z}, we have
\begin{equation}\label{Phi-hat-1}
\breve{\Phi}(-s;2\alpha)=\sqrt{iv_2(x;2\alpha)/2\alpha}\left(
                                                                                        \begin{array}{cc}
                                                                                          1 & 1 \\
                                                                                          \frac{v_2'(x;2\alpha)-2\alpha}{2iv_2(x;2\alpha)} & \frac{v_2'(x;2\alpha)+2\alpha}{2iv_2(x;2\alpha)}\\
                                                                                        \end{array}
                                                                                      \right)d_1^{\sigma_3} \quad  \textrm{for $\alpha\neq0$}
\end{equation}
and
\begin{equation}\label{Phi-hat-2}
\breve{\Phi}(-s;2\alpha)=\sqrt{-iv_2(x;2\alpha)}\left(
                                                                                        \begin{array}{cc}
                                                                                          1 & 0 \\
                                                                                          \frac{v_2'(x;2\alpha)}{2iv_2(x;2\alpha)} & \frac{i}{v_2(x;2\alpha)}\\
                                                                                        \end{array}
                                                                                      \right)\left(
                                                                                               \begin{array}{cc}
                                                                                                 1 & d_2 \\
                                                                                                 0 & 1 \\
                                                                                               \end{array}
                                                                                             \right) \quad \textrm{for $\alpha=0$}
\end{equation}
with certain non-zero constants $d_1$ and $d_2$.
Then, a combination of \eqref{Phi-shift-1}-\eqref{Phi-hat-2} yields the following equations
\begin{equation}\label{Relation-shift}
\begin{cases}
  r_2(x;2\alpha+1)-r_2(x;2\alpha)=\frac{-v'_2(x;2\alpha)+2\alpha}{2v_2(x;2\alpha)}, \vspace{0.2cm} \\
         \frac{v'_2(x;2\alpha+1)+2\alpha+1}{2v_2(x;2\alpha+1)}+\frac{v'_2(x;2\alpha)-2\alpha}{2v_2(x;2\alpha)}=0.
\end{cases}
\end{equation}
Using \eqref{eq:Hamiltotion -r-2}, the first equation of the above formula gives us the B\"{a}cklund transformation for the Hamiltonian $H$:
\begin{equation}
  H(x;2\alpha+1)-H(x;2\alpha)=\frac{v'_2(x;2\alpha)-2\alpha}{2v_2(x;2\alpha)}.
\end{equation}
Finally, \eqref{BT-1} and \eqref{BT-2} follow from \eqref{eq:Hamiltonian system}, \eqref{equation r and v} and \eqref{Relation-shift}.
\end{proof}

{\rem{In the  case $s=0$,  we have $v_1=0$ and $v_2(x;2\alpha,\omega)=u(x;2\alpha,0)$ satisfies
   the $\textrm{P}_{34}$ equation \eqref{int-painleve 34} with parameter $2\alpha$ and $w_2(x;2\alpha+\frac {1}{2})=-2^{1/3} y(-2^{1/3}x;2\alpha+\frac {1}{2})$ satisfies the $\textrm{P}_{2}$ equation \eqref{Painleve II} with parameter $2\alpha+\frac{1}{2}$. From \eqref{BT-1} and \eqref{BT-2}, we recover the following B\"{a}cklund transformations for $\textrm{P}_{2}$ and $\textrm{P}_{34}$ equations
\begin{equation}\label{BT-P34}
u(x;2\alpha+1,0)-u(x;2\alpha,0)=-2^{2/3}y'(-2^{1/3}x;2\alpha+\frac {1}{2}),
\end{equation}
\begin{equation}\label{BT-p2.}
y(x;2\alpha+\frac{3}{2})+y(x;2\alpha+\frac{1}{2})=\frac{2\alpha+1}{y^2(x;2\alpha+\frac{1}{2})+y'(x;2\alpha+\frac{1}{2})+\frac{x}{2}};
\end{equation}
see \cite[(32.7.1)-(32.7.2)]{O}, \cite[(3.23)]{FW-2001} and \cite[(3.11)]{FW-2015}}.}


\section{Asymptotics of $v_i(x)$ as $x\rightarrow+\infty$ }\label{sec:asy of v1-infty}

In the present section and Appendix \ref{sec:asy of v1--infty}, we perform the Deift-Zhou nonlinear steepest descend method \cite{dkmv1,dkmv2,dz} for the model RH problem of $\Phi$ as the parameter $x \to \pm \infty$. Then, based on the connection of $v_i(x)$ to $\Phi$ given in \eqref{v1-Phi'} and \eqref{v2-Phi'}, we obtained their asymptotics as $x\to \pm\infty$. As the steepest descent analysis in each section is independent, we are going
to use the same notation for the functions. We trust that this will not lead to any confusion.

\subsection{Nonlinear steepest descent analysis of $\Phi$ as $x \to +\infty$}

We rescale the variable and introduce the first transformation as follows:
\begin{equation}\label{A-scaling}
  A(z)=x^{\sigma_3/4}\left(
                       \begin{array}{cc}
                         1 & 0 \\
                         ir_2 & 1 \\
                       \end{array}
                     \right)
  \Phi (xz).
\end{equation}
To normalize the large-$z$ behavior of $A(z)$, we define the $g$-function
\begin{equation}\label{g-function}
g_1(z):=\frac 23 (z+1)^{3/2},\quad \arg(z+1)\in(-\pi,\pi).
\end{equation}
It is easy to see that $ g_1(z)-\left( \frac{2z^{3/2}}{3}+z^{1/2} \right)=\frac{1}{4z^{1/2}}  +O( z^{-3/2})$ $\textrm{as } z \to \infty.
$
The second transformation is devoted to  a normalization at infinity and a shift of jump contours
\begin{equation}\label{A-B-scaling}
B(z)=\left(\begin{array}{cc}
1 & 0\\
-\frac i 4 x^{\frac 32} & 1
\end{array}\right)\left\{\begin{array}{ll}
A(z)e^{x^{\frac 32}g_1(z)\sigma_3}, & z\in I \cup III\cup IV,\\
A(z)e^{x^{\frac 32}g_1(z)\sigma_3}\left(\begin{array}{cc}
1 & 0\\
\pm e^{2x^{\frac 32}g_1(z)\sigma_3} e^{\pm 2\pi i\alpha}  &1
\end{array}\right), & z\in II\cup V,\end{array}\right.
\end{equation}
where the regions are illustrated in Fig. \ref{Figure-B-2}.
\begin{figure}[h]
 \begin{center}
   \includegraphics[width=6cm]{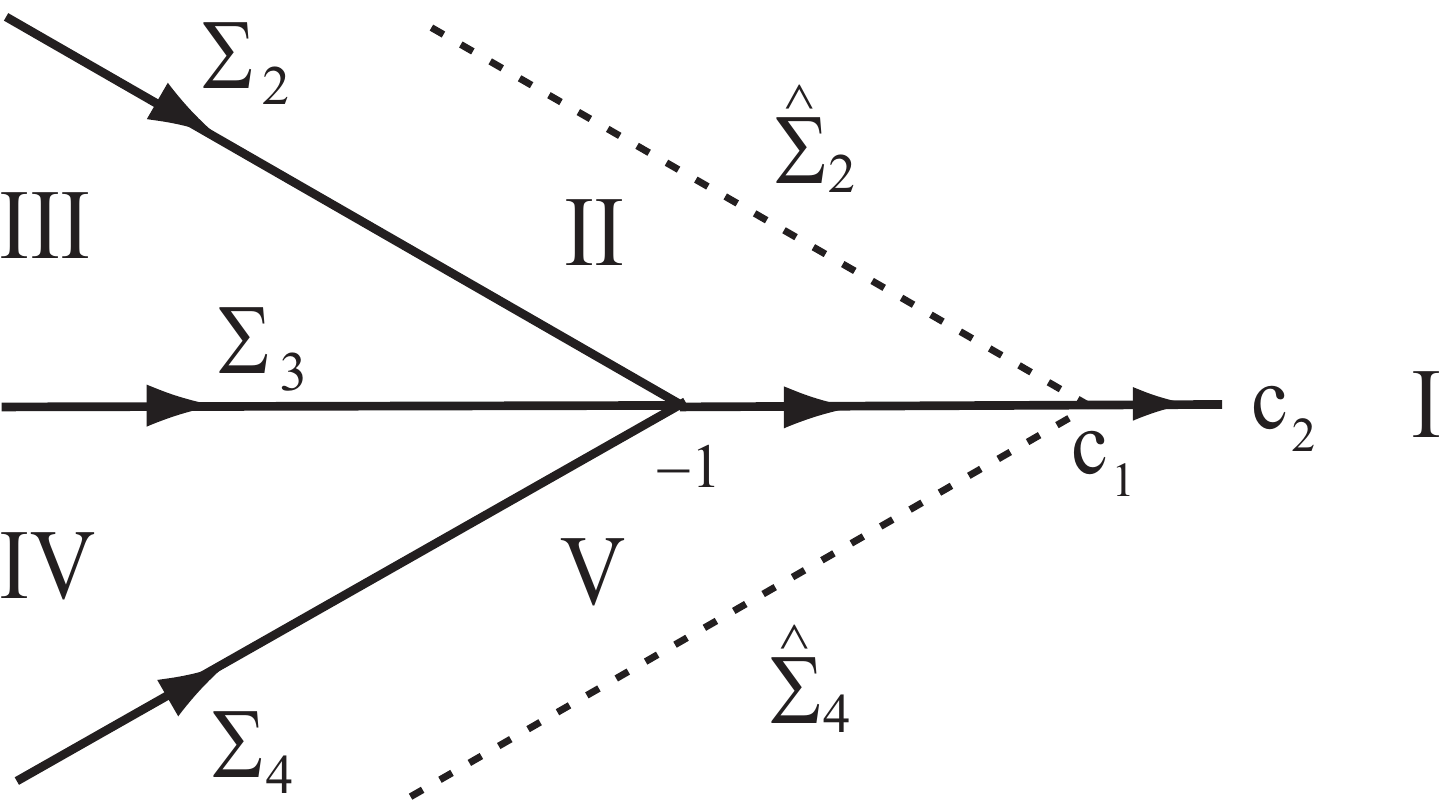} \end{center}
 \caption{\small{Regions and contours  for   $B$ ($c_1=\min(0,-\frac {s}{x})$ and $c_2=\max(0,-\frac {s}{x})$).}}
 \label{Figure-B-2}
 \end{figure}
Then, $B(z)$ satisfies the following RH problem.
\begin{rhp} The function $B(z)$ satisfies the following properties:
\begin{itemize}
\item[(a)] $B(z)$ is analytic in $\mathbb{C}\backslash \{ \Sigma_2\cup \Sigma_4 \cup (-\infty,\max(0,-\frac {s}{x})] \}$;

\item[(b)] $B(z)$ satisfies the jump condition
\begin{equation}\label{B-jump}
B_{+}(z)=B_{-}(z)J_B(z),
\end{equation}
where
\begin{equation*}
J_B(z)=\left\{\begin{array}{lll}
\left(\begin{array}{cc} 1 & 0\\ e^{2x^{\frac 32}g_1(z)}e^{2\pi i\alpha} & 1 \end{array}\right), & z\in \Sigma_2,\\
\left(\begin{array}{cc} 1 & 0\\e^{2x^{\frac 32}g_1(z)} e^{-2\pi i\alpha} & 1 \end{array}\right), & z\in \Sigma_4,\\
\left(\begin{array}{cc} 0 & 1\\ -1 & 0 \end{array}\right), & z\in (-\infty,-1),\\
\left(\begin{array}{cc} e^{2\pi i\alpha} & e^{-2x^{\frac 32}g_1(z)}\\ 0 & e^{-2\pi i\alpha} \end{array}\right), & z\in (-1,\min(0,-\frac {s}{x})),\\
\end{array}\right.
\end{equation*}
and
\begin{equation*}
J_B(z)=\left\{\begin{array}{ll}
e^{2\pi i\alpha\sigma_3}  & z\in (0,-\frac {s}{x})\quad \mbox{if}\quad  s<0,\\
\left(\begin{array}{cc} 1 & \omega e^{-2x^{\frac 32}g_1(z)}\\ 0 & 1 \end{array}\right), & z\in (-\frac {s}{x},0)\quad \mbox{if}\quad  s>0;
\end{array}\right.
\end{equation*}

\item[(c)]  The asymptotic behavior of $B(z)$ at infinity:
\begin{equation}\label{Apos-infinity}
B(z)=
    \left (I+O\left (\frac 1 z\right )\right)
 z^{-\frac{1}{4}\sigma_3}\frac{I+i\sigma_1}{\sqrt{2}} \qquad \textrm{as } z \to \infty.
\end{equation}
\end{itemize}
\end{rhp}

Note that
\begin{equation}
  \Re g_1(z)  < 0, \quad \textrm{for } z \in \Sigma_2 \cup \Sigma_4, \quad \textrm{and} \quad \Re g_1(z)  >  0 \quad \textrm{for } z \in (-1, +\infty),
\end{equation}
then, the off-diagonal entries of the jump matrices  are exponentially small as $x\to +\infty$. Neglecting the exponential small terms, we arrive at the following outer parametrix.
\begin{rhp} \label{rhp:outer1} The function $B^{(\infty)}(z)$ satisfies the following properties:
\begin{itemize}
  \item[(a)]   $B^{(\infty)}(z)$ is analytic in
  $\mathbb{C}\backslash (-\infty,-\frac{s}{x}]$;
\item[(b)]   $B^{(\infty)}(z)$  satisfies the jump condition
 \begin{equation}
B^{(\infty)}_+(z)=B^{(\infty)}_-(z)\left\{\begin{array}{lll}
\left(\begin{array}{cc} 0 & 1\\ -1 & 0 \end{array}\right), & z\in (-\infty,-1), \vspace{0.2cm} \\
e^{2\pi i\alpha \sigma_3}, & z\in (-1,-\frac {s}{x});
\end{array}\right.
\end{equation}
\item[(c)] At infinity, $B^{(\infty)}(z)$ satisfies the same asymptotics as $B(z)$ in \eqref{Apos-infinity}.

\end{itemize}
\end{rhp}

To construct a solution to the above RH problem, let us first introduce a scalar function $h(z)$ as follows:
\begin{equation}\label{h-function}
h(z):=\left(\frac {\sqrt{z}-1}{\sqrt{z}+1}\right)^{\alpha},\quad z\in \mathbb{C}\setminus (-\infty,1],
\end{equation}
where the power function $z^c, c \notin \mathbb{Z},$ takes the principle branch with the cut along  $(-\infty,0)$.
Note that, $h(z)$ satisfies the following jump conditions
\begin{equation}\label{h-jump} \left\{\begin{array}{c}
                                        h_+(x)=h_-(x)e^{2\pi i\alpha}, \quad x\in(0,1) \\
                                         h_+(x)h_-(x)=1, \quad x\in(-\infty,0).
                                      \end{array}
\right.
\end{equation}
Then, a solution to the RH problem for $B^{(\infty)}(z)$ is given explicitly as
\begin{equation}\label{B-infiniy}
B^{(\infty)}(z)=\left(\begin{array}{cc} 1 & 0\\ 2\alpha i\sqrt{1-\frac {s}{x}}& 1\end{array}\right)(z+1)^{-\sigma_3/4}\frac{I+i\sigma_1}{\sqrt{2}}
h_1(z)^{\sigma_3},
\end{equation}
where
\begin{equation}\label{h-2}
h_1(z)=h\left(\frac {z+1}{1-\frac {s}{x}}\right)=\left(\frac {\sqrt{z+1}-\sqrt{1-\frac {s}{x}}}{\sqrt{z+1}+\sqrt{1-\frac {s}{x}}}\right)^{\alpha}, \quad z\in \mathbb{C}\setminus \left(-\infty,-\frac {s}{x}\right].
\end{equation}


The jump matrices of $B(z)B^{(\infty)}(z)^{-1}$ are not uniformly close to the unit matrix near the end-points
$-1$, 0 and $-\frac {s}{x}$. Then, local parametrices have to be constructed in neighborhoods of the end-points.

Let $U_{-1}:= \{z: \ |z + 1| < \delta\}$ for certain small $0 < \delta < 1/2$. The local parametrix near $-1$ should share the same jumps \eqref{B-jump} with $B$ in the neighborhood $U_{-1}$, and match with $B^{(\infty)}(z)$ on $|z + 1| = \delta.$ It is readily verified that such a parametrix can be represented as follows
\begin{equation}\label{parametrix at -1 }
B^{(-1)}(z)=E_0(z)Z_{A}(x(z+1))\left\{\begin{array}{ll}
                                                     e^{(\frac {2}{3}x^{\frac 32}(z+1)^{\frac 32}+\alpha\pi i)\sigma_3}, \quad &\arg z\in (0, \pi),\\
                                                     e^{(\frac {2}{3}x^{\frac 32}(z+1)^{\frac 32}-\alpha\pi i)\sigma_3},\quad &\arg z\in (-\pi, 0),
                                                   \end{array}
\right.
\end{equation}
where the pre-factor $E_0(z)$ is an analytic function in $U_{-1}$. Here, $Z_{A}(z)$ is the following solution to the well-known Airy model RH problem:
\begin{figure}[h]
 \begin{center}
   \includegraphics[width=5cm]{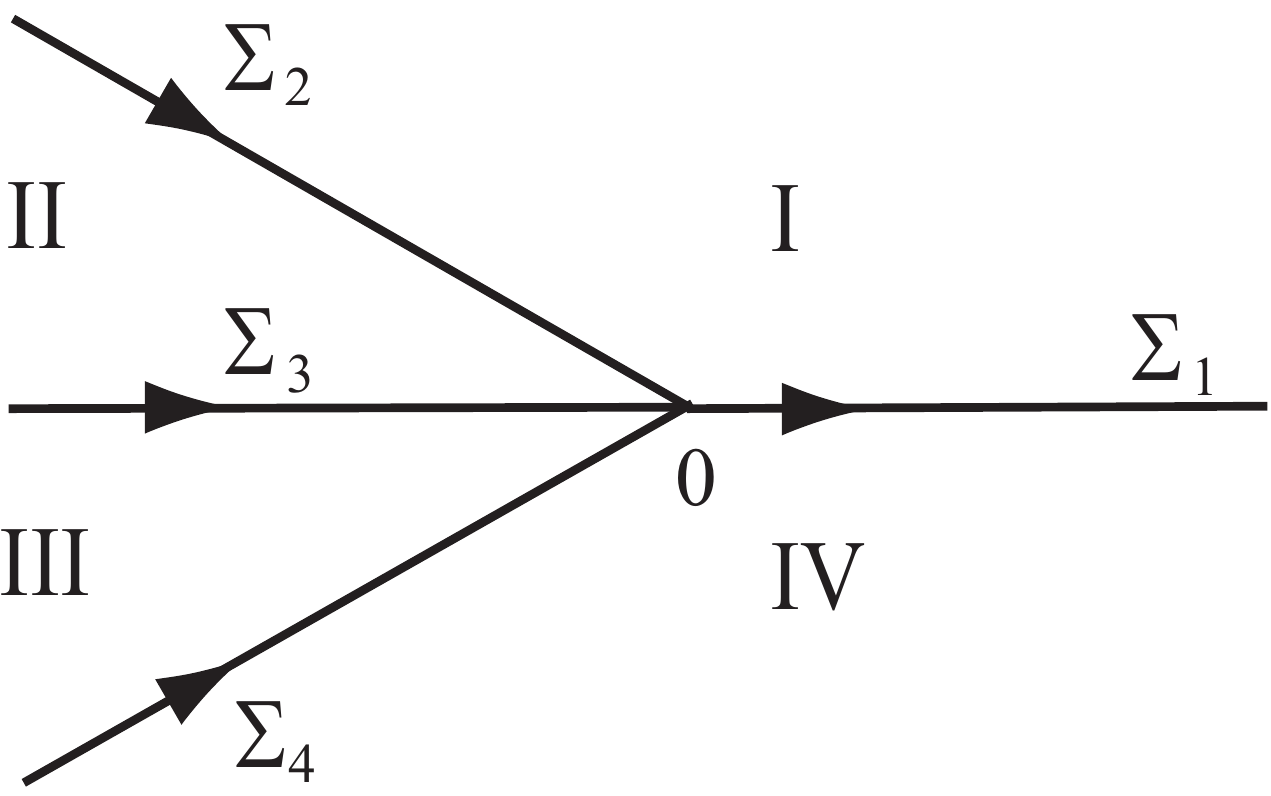} \end{center}
 \caption{\small{Regions and contours  for   $Z_A$.}}
 \label{Figure-Airy}
 \end{figure}
\begin{equation}\label{Airy-model-solution}
  Z_A(\lambda)=M_A\left\{
                 \begin{array}{ll}
                  \left(
                     \begin{array}{cc}
                        \Ai(\lambda) & \Ai(\omega^2 \lambda) \\
                   \Ai'(\lambda)&\omega^2 \Ai'(\omega^2 \lambda) \\
                      \end{array}
                 \right)e^{-\frac{\pi i}{6}\sigma_3}, &\lambda\in I \\[.4cm]
                       \left(
                     \begin{array}{cc}
                        \Ai(\lambda) & \Ai(\omega^2 \lambda) \\
                      \Ai'(\lambda)&\omega^2 \Ai'(\omega^2\lambda) \\
                     \end{array}
                 \right)e^{-\frac{\pi i}{6}\sigma_3}\left(
                                                       \begin{array}{cc}
                                                    1 & 0 \\
                                                        -1 & 1 \\
                                                      \end{array}
                                                     \right)
                 , & \lambda\in II \\[.4cm]
                            \left(
                    \begin{array}{cc}
                         \Ai(\lambda) & -\omega^2 \Ai(\omega\lambda) \\
                        \Ai'(\lambda)&- \Ai'(\omega \lambda) \\
                    \end{array}
                 \right)e^{-\frac{\pi i}{6}\sigma_3}\left(
                                                      \begin{array}{cc}
                                                       1 & 0 \\
                                                      1 & 1 \\
                                                      \end{array}
                                                   \right)
                 , &\lambda\in III\\[.4cm]
                       \left(
                     \begin{array}{cc}
                    \Ai(\lambda) & -\omega^2 \Ai(\omega \lambda) \\
                        \Ai'(\lambda)&- \Ai'(\omega \lambda) \\
                      \end{array}
                  \right)e^{-\frac{\pi i}{6}\sigma_3}, & \lambda\in IV,
                 \end{array} \right.
   \end{equation}
where  $\omega=e^\frac{2\pi i}{3}$, the constant matrix $M_A=\sqrt{2\pi} e^{\frac 1 6\pi i}  \left(
                                                                               \begin{array}{cc}
                                                                                 1 & 0 \\
                                                                                 0 & -i \\
                                                                               \end{array}
                                                                             \right)$ and the regions are indicated in Fig. \ref{Figure-Airy}; cf. \cite[(7.9)]{dkmv2}.

As $x \to +\infty$, the other two endpoints $0$ and $-\frac {s}{x}$ are every close to each other. Then, we consider them together and look for the following local parametrix in $U_0:= \{z: \ |z| < \delta\}$.
\begin{rhp} The function $B^{(0)}(z)$ satisfies the following properties:
\begin{itemize}
  \item[(a)]   $B^{(0)}(z)$ is analytic in
  $U_0\backslash (-\infty,\max(0,-\frac {s}{x})]$;
\item[(b)] $B^{(\infty)}(z)$  satisfies  the jump condition
 \begin{equation}
B^{(0)}_+(z)=B^{(0)}_-(z)\left\{\begin{array}{ll}
e^{2\pi i\alpha\sigma_3}, & z\in (0, -\frac {s}{x}),\\
\left(\begin{array}{cc} e^{2\pi i\alpha} & e^{-2x^{\frac 32}g_1(z)}\\ 0 & e^{-2\pi i\alpha} \end{array}\right), & z\in (-1,0);
\end{array}\right. \quad  \textrm{for $s<0$},
\end{equation}
and
\begin{equation}
B^{(0)}_+(z)=B^{(0)}_-(z)\left\{\begin{array}{ll}
\left(\begin{array}{cc} e^{2\pi i\alpha} & e^{-2x^{\frac 32}g_1(z)}\\ 0 & e^{-2\pi i\alpha} \end{array}\right), & z\in (-1,-\frac {s}{x}),\\
\left(\begin{array}{cc} 1 & \omega e^{-2x^{\frac 32}g_1(z)}\\ 0 & 1\end{array}\right), & z\in (-\frac {s}{x},0),
\end{array}\right. \quad \textrm{for $s>0$};
\end{equation}

\item[(c)] $B^{(0)}(z)$ fulfils the following matching condition on $\partial U_0$:
$$
B^{(0)}(z) B^{(\infty)}(z)^{-1} = I + o(1), \qquad \textrm{as } x \to +\infty .
$$
\end{itemize}
\end{rhp}

A solution to the above RH problem is given explicitly as follows:
\begin{equation} \label{b0-x-positive}
B^{(0)}(z)=\left(\begin{array}{cc} 1 & 0\\ 2\alpha i\sqrt{1-\frac {s}{x}}& 1\end{array}\right)(z+1)^{-\sigma_3/4}\frac {I+i\sigma_1}{\sqrt{2}}
\left(
  \begin{array}{cc}
    1 & j(z) \\
    0 &1 \\
  \end{array}
\right)
h_1(z)^{\sigma_3},
\end{equation}
where
\begin{equation*}\label{j}
j(z)=\left\{\begin{array}{ll}
\frac {\omega}{2\pi i}\int_{-\frac {s}{x}}^0 \frac {\exp(-\frac {4}{3}x^{\frac 32}(\zeta+1)^{\frac {3}{2}})|h_1(\zeta)|^2}{\zeta-z}d\zeta+\frac {1}{2\pi i}\int_{-\frac {1}{2}}^{-\frac {s}{x}} \frac {\exp(-\frac {4}{3}x^{\frac 32}(\zeta+1)^{\frac {3}{2}})|h_1(\zeta)|^2}{\zeta-z}d\zeta, & s>0\\
\frac {1}{2\pi i}\int_{-\frac {1}{2}}^0 \frac {\exp(-\frac {4}{3}x^{\frac 32}(\zeta+1)^{\frac {3}{2}})|h_1(\zeta)|^2}{\zeta-z}d\zeta, &s<0,
\end{array}\right.
\end{equation*}
for $|z|<\delta$ and $h_1(z)$ is defined in \eqref{h-2}.

Now the final transformation is given by
\begin{equation} \label{trans:C-final}
C(z)=\left\{\begin{array}{lll}
B(z)B^{(\infty)}(z)^{-1}, & \mathbb{C}\backslash (U_{-1}\cup U_0),\\
B(z)(B^{(0)}(z))^{-1}, & z\in U_{0},\\
B(z)(B^{(-1)}(z))^{-1}, & z\in U_{-1}.\end{array}\right.
\end{equation}
Then, $C(z)$ satisfies the following RH problem:
\begin{figure}[h]
 \begin{center}
   \includegraphics[width=5cm]{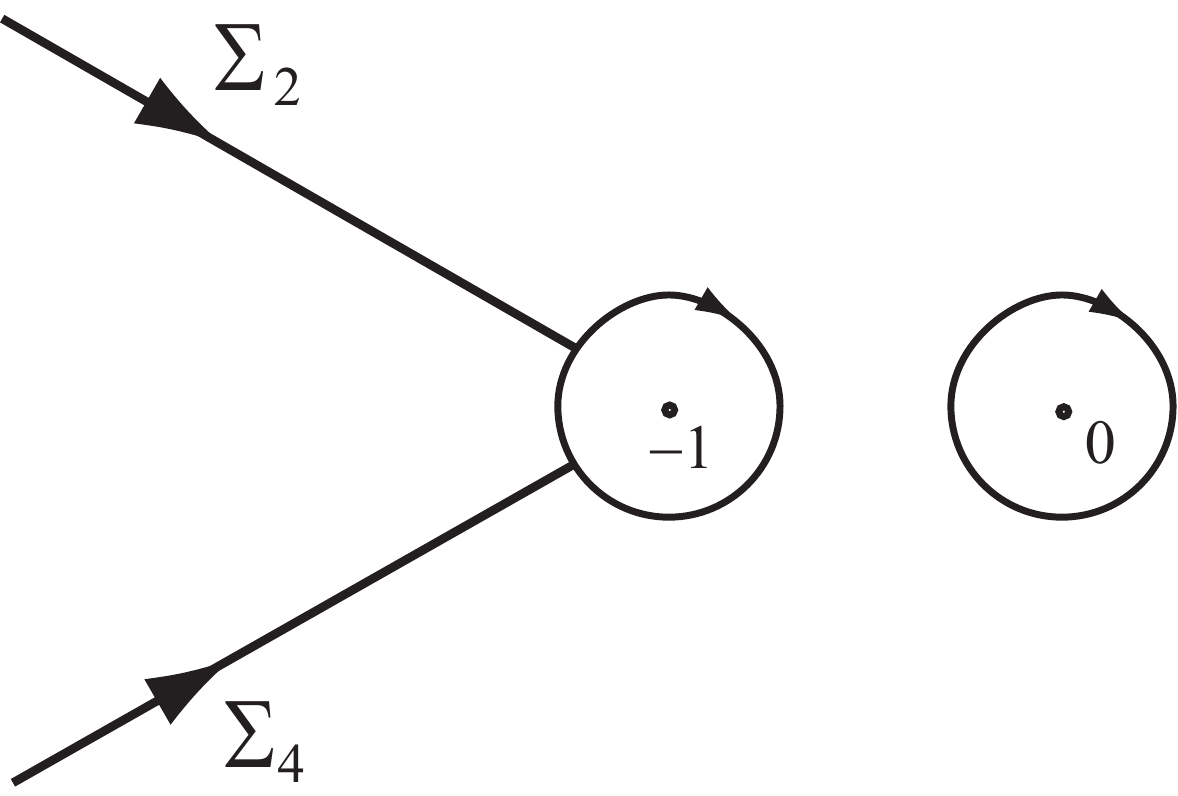} \end{center}
 \caption{\small{Contour $\Sigma_C$.}}
 \label{Figure-c-infty}
 \end{figure}
\begin{rhp} \label{rhp:c-final} The function $C(z)$ satisfies the following properties:
  \item[(a)] $C(z)$ is analytic in $\mathbb{C} \setminus \Sigma_C$; see Fig.\ref{Figure-c-infty};

  \item[(b)] $C(z)$ satisfies the jump condition  $ C_+(z)=C_-(z)J_C(z)$,
  $$\begin{array}{ll}
                             J_C(z)= B^{(0)}(z)(B^{(\infty)}(z))^{-1}, &  z \in \partial U_0,\\
                             J_C(z)= B^{(-1)}(z)(B^{(\infty)}(z))^{-1}, &  z \in \partial U_{-1}, \\
                             J_C(z)= B^{(\infty)}(z) J_B(z) (B^{(\infty)}(z))^{-1}, &  z\in \Sigma_{C}\setminus \partial (U_0 \cup U_{-1});
                             \end{array}$$

  \item[(c)] As $z \to \infty$, $C(z)=I+O(1/z).$
\end{rhp}

It follows from the properties of the local parametrices that
\begin{equation}\label{R-jump-approx}
J_C(z)=\left\{ \begin{array}{ll}
                     I+O\left (x^{-3/2}\right ),&  z\in\partial
                    U_{-1},\\[.1cm]
I+O(e^{-cx}), ~& z \in  \Sigma_R\setminus \partial U_{-1},
                 \end{array}\right .
\end{equation}
where $c$ is a positive constant, and the error term is uniform for $z$ on the corresponding contours. Then, 
it follows from the above formula
\begin{equation} \label{C-estimate-positive x}
C(z)=I+O(x^{-3/2}),
\end{equation}
uniformly  for $z$ in the complex plane;  see \cite{deift,dkmv2}.

\subsection{Proof of Theorem \ref{theorem-v-1}} \label{sec:vi+}

The analyticity of $v_i(x)$ for real values of $x$ is proved in Corollary \ref{Analytic of v-1}. Next, we compute their asymptotics.

Based on the transformations $\Phi \mapsto A \mapsto B$ in \eqref{A-scaling} and \eqref{A-B-scaling}, we have from \eqref{v1-Phi'} and \eqref{v2-Phi'}
\begin{equation}\label{v1-A}
v_1(x)=i\lim_{z\to 0}z(\Phi'(z)\Phi(z)^{-1})_{12}=\frac {i}{\sqrt{x}}\lim_{z\to 0}z(B'(z)B(z)^{-1})_{12},
\end{equation}
\begin{equation}\label{v2-A}
v_2(x)=i\lim_{z\to -s}(z+s)(\Phi'(z)\Phi(z)^{-1})_{12}=\frac {i}{\sqrt{x}}\lim_{z\to -s/x}(z+\frac {s}{x})(B'(z)B(z)^{-1})_{12}.
\end{equation}
By the transformation \eqref{trans:C-final} and the approximation \eqref{C-estimate-positive x}, we have for $|z|<\delta$
\begin{equation}\label{B-C-0}
B(z)=(I+O(x^{-3/2})) B^{(0)}(z) .
\end{equation}
Recalling the expression of $B^{(0)}(z)$ in \eqref{b0-x-positive}, we find the asymptotics of $v_1$ from \eqref{v1-A} and \eqref{B-C-0}
\begin{equation}\label{v-1-+infty}
v_1(x)=\left\{\begin{array}{ll}
\frac {\omega}{4\pi \sqrt{x}}e^{-\frac {4}{3}x^{\frac {3}{2}}}|h_1(0)|^2(1+O(x^{-\frac {3}{2}})), & s>0,\\
\frac {1}{4\pi \sqrt{x}}e^{-\frac {4}{3}x^{\frac {3}{2}}}|h_1(0)|^2(1+O(x^{-\frac {3}{2}})), &s<0,
\end{array}\right. \qquad \textrm{as } x \to +\infty,
\end{equation}
where
$$|h_1(0)|^2=\left|\frac {1-\sqrt{1-\frac {s}{x}}}{1+\sqrt{1-\frac {s}{x}}}\right|^{2\alpha}=\frac {1}{2^{4\alpha}}\left|\frac {s}{x}\right|^{2\alpha}\left(1+\alpha\frac {s}{x}+O\left(\frac {1}{x^2}\right)\right).$$
Similarly, we get the asymptotics of $v_2(x)$
\begin{equation}\label{v-2-+infty}
v_2(x)=\frac{\alpha}{\sqrt{x-s}}-\frac{\alpha^2}{(x-s)^2}+O(x^{-3}), \quad \textrm{as } x\to+\infty, \ \textrm{ if } \alpha \neq 0.
\end{equation}
If $\alpha=0$ and $s<0$, the function $\Phi(z)$ is analytic at the point $z = -s$. By \eqref{v2-A}, we have $v_2(x)=0$.
If $\alpha=0$ and $s>0$, we have the following
asymptotics
\begin{equation} \label{v2-asy-+infty2}
v_2(x)=\frac {1-\omega}{4\pi \sqrt{x}}e^{-\frac {4}{3}(x-s)^{\frac {3}{2}}}(1+O(x^{-3/2})).
\end{equation}
Recall the relations among $v_i$, $w_i$ and $H$ in \eqref{eq:Hamiltonian system} and \eqref{def:Hamiltonian-CPII}. Then, the asymptotics of $w_i$ in \eqref{eq:asy-w-1} and \eqref{eq:asy-w-2}, as well as the Hamiltonian $H$ in \eqref{eq:asy-H},  are derived by direct computations. This completes the proof of Theorem \ref{theorem-v-1}.

\medskip

To obtain the large gap asymptotics in Section \ref{sec:large gap asy}, let us compute two more quantities, namely $\left(\Phi^{(-s)}_0\right)_{11}(x;s)$ and $\left(\Phi^{(-s)}_1\right)_{11}(x;s)$. When $s<0$, recalling the expansion of $\Phi(z)$ in \eqref{Phi at -s}, we have from \eqref{A-scaling}, \eqref{A-B-scaling}, \eqref{b0-x-positive} and \eqref{B-C-0}
\begin{equation}\label{eq:asy-phi-0}
\left(\Phi^{(-s)}_0\right)_{11}(x;s)=2^{-2\alpha-\frac{1}{2}}\exp(-\frac{2}{3}(x-s)^{\frac{3}{2}})(x-s)^{-(\alpha+\frac{1}{4})}(1+O(x^{-\frac{3}{2}})) \quad \textrm{as } x\to \infty
\end{equation}
and
\begin{equation}\label{eq:asy-phi-1}
\left(\Phi^{(-s)}_1\right)_{11}(x;s)=-\sqrt{x-s}+O(1/x) \qquad \textrm{as } x\to \infty.
\end{equation}


\section{Proof of Theorem \ref{thm:TW-P34}-\ref{thm:TW-P2} } \label{sec: tw}

%

\subsection{Tracy-Widom formula for $\textrm{P}_{34}$ kernel: proof of Theorem \ref{thm:TW-P34} }

Before we use the differential identity in \eqref{ logrithmic derivaive t related to Y}, let us first compute the coefficient $Y_{-1}$. Tracing back the transformation $Y\to \widetilde{\Phi}\to \Phi$ in \eqref{U-1}, \eqref{U} and \eqref{U to Phi}, we obtain
\begin{equation}\label{Phi and Y}
Y(z)=\left(
       \begin{array}{cc}
         1 & 0 \\
         -\eta & 1\\
       \end{array}
     \right)\Phi(z-s)
\Psi(z)^{-1}.
\end{equation}
From the expansions of $\Psi$ and $\Phi$ at infinity in \eqref{psi-alpha infinity} and \eqref{Phi at infinity}, the differential identity in \eqref{ logrithmic derivaive t related to Y} yields the following more explicit form
\begin{equation}\label{differential identity-2}
 \frac{d}{dt}\ln \det[I-K^{P34}_{\alpha,\omega,s}]=-r_2(s+t;s)+m_2(t) - \frac{t^2}{4} +\frac {(s+t)^2}4,
\end{equation}
where $m_2$ and $r_2$ are coefficients in the asymptotics of $\Psi$ and $\Phi$ near infinity in  \eqref{psi-alpha infinity} and \eqref{Phi at infinity}, respectively. Note that $x^2/4- r_2(x)$ and $x^2/4-m_2(x)$ are Hamiltonians for the coupled $\textrm{P}_2$ and the $\textrm{P}_2$ equations, respectively; see \eqref{eq:Hamiltotion -r-2}. The above formula and \eqref{equation r and v} give us
\begin{equation}\label{sencond Dt}
\frac{d^2}{dt^2}\ln \det[I-K^{P34}_{\alpha,\omega,s}]= u(t)-v_1(s+t)-v_2(s+t),
\end{equation}
where $v_i(x)=v_i(x,s;2\alpha,\omega)$ are  the solutions to the coupled $\textrm{P}_{2}$ equations \eqref{int-equation v}  with the properties stated in Theorem \ref{theorem-v-1}
and $u(x)=u(x;2\alpha,\omega)$ is the $\textrm{P}_{34}$ transcendent with the boundary condition \eqref{thm: painleve  positive infinity}.

For fixed $s\in \mathbb{R}$, the limit below
\begin{equation}\label{FD P34-t large}
\lim_{t\to+\infty}\ln \det[I-K^{P34}_{\alpha,\omega,s}]=0
\end{equation}
follows from the estimation of functions $\psi_i$ in the $\textrm{P}_{34}$ kernel in \eqref{psi-kernel}
\begin{equation}\label{asy of Psi}
|\psi_i(z,t)|=O( e^{-\frac{2}{3} (z+t)^{2/3}}|z|^{\alpha}), \qquad \textrm{as } t \to +\infty,
\end{equation}
uniformly for $z>c_0$, $c_0 \in \mathbb{R}$. The above estimation can be found in the derivation of the large-$t$ asymptotics  of the $\textrm{P}_{34}$ transcendent in Its, Kuijlaars  and \"{O}stensson \cite{ik2}; see also Section \ref{sec:asy of v1-infty} with the parameter $s=0$ therein.

Denote $f(t,s):=\ln \det[I-K^{P34}_{\alpha,\omega,s}]$. From \eqref{sencond Dt} and \eqref{FD P34-t large}, an integration by parts gives us
\begin{align}
  \ln \det[I-K^{P34}_{\alpha,\omega,s}] & = - \int_t^{+\infty} \frac{d}{dx} f(x,s) dx \nonumber \\
  & = -\int_{t}^{+\infty}(v_1(x+s)+v_2(x+s)-u(x))(x-t)dx, \label{tw-int-parts} 
\end{align}
where the convergence of the above integral is guaranteed by the properties of $u$ in Proposition \ref{pro- p34} and $v_i$ in Theorem \ref{theorem-v-1}. This completes the proof of Theorem \ref{thm:TW-P34}.

\subsection{Tracy-Widom formula for $\textrm{P}_{2}$ kernel: proof of Theorem \ref{thm:TW-P2} }

Let us first prove Lemma \ref{lem:kp2-kp34}, which indicates the connection between the Fredholm determinants of the $\textrm{P}_{2}$ and the $\textrm{P}_{34}$ kernels.

\medskip

\noindent\emph{Proof of Lemma \ref{lem:kp2-kp34}.}
By definition, we have the gap probability of the random matrices \eqref{quartic-UE} near zero
\begin{align}
  &\mbox{Prob}[\mbox{$M$ has no eigenvalues in }(-s(n),s(n))]=  \nonumber \\
  & \hspace{4cm} \frac {1}{Z_n}\int_{I^n}\prod_{k=1}^n |x_k|^{2\alpha}e^{-nV(x_k)}\prod_{1\leq i<j\leq n} |x_i-x_j|^2 \prod_{k=1}^n dx_k, \label{gap pro zero-int}
\end{align}
where $I:= (-\infty,-s(n))\cup(s(n),+\infty)$, $s(n)=\frac {s}{2^{-2/3}n^{1/3}}$, $V(x)=\frac {x^4}{4}+\frac {g}{2}x^2$
and  the constant $Z_n=\int_{\mathbb{R}^n}\prod_{k=1}^n |x_k|^{2\alpha}e^{-nV(x_k)}\prod_{1\leq i<j\leq n} |x_i-x_j|^2 \prod_{k=1}^n dx_k$.
The multiple integrals can be written in a form of  the following Hankel determinants
 \begin{equation}\label{gap pro zero-Hankel}
\mbox{Prob}[\mbox{$M$ has no eigenvalues in }(-s(n),s(n))
]=
\frac { \det\left[\int_{I}|x|^{2\alpha}e^{-nV(x)}x^{i+j}dx\right]_{i,j=0}^{n-1}}{ \det\left[\int_{\mathbb{R}}|x|^{2\alpha}e^{-nV(x)}x^{i+j}dx\right]_{i,j=0}^{n-1}}.
 \end{equation}
Using the fact that the entries in the Hankel determinant vanish when $i$ and $j$ have different parity, we rearrange the rows and columns in the Hankel determinant and obtain
\begin{equation*}\label{Hankel-split}
 \begin{split}
\det D:&=\det\left[\int_{I}|x|^{2\alpha}e^{-nV(x)}x^{i+j}dx\right]_{i,j=0}^{n-1}\\
&=\det\left[
                       \begin{array}{cc}
                        \left[\int_{I}|x|^{2\alpha}e^{-nV(x)}x^{2(i+j)}dx\right]_{i,j=0}^{[(n+1)/2]-1} & 0 \\
                         0 & \left[\int_{I}|x|^{2\alpha}e^{-nV(x)}x^{2(i+j+1)}dx\right]_{i,j=0}^{[n/2]-1} \\
                       \end{array}
                     \right]\\
&=\det \left[\int_{I}|x|^{2\alpha}e^{-nV(x)}x^{2(i+j)}dx\right]_{i,j=0}^{[(n+1)/2]-1} \det\left[\int_{I}|x|^{2\alpha}e^{-nV(x)}x^{2(i+j+1)}dx\right]_{i,j=0}^{[n/2]-1};
\end{split}
\end{equation*}
see similar arguments in Forrester \cite{F}. From the above formula, we have
\begin{equation}\label{gap pro zero-Hankel-split}
\mbox{Prob}[\mbox{$M$ has no eigenvalues in }(-s(n),s(n))]=\mathbb{P}_{\frac {\alpha}2+\frac {1}{4}}\mathbb{P}_{\frac{\alpha}{2}-\frac 14},
\end{equation}
where
\begin{equation}\label{pro-i}
\mathbb{P}_{\frac {\alpha}{2}\pm\frac 14}=\frac { \det \left[\int_{s(n)^2}^{+\infty}|x|^{\alpha\pm\frac {1}{2}}e^{-nV(\sqrt{x})}x^{i+j}dx\right]_{i,j=0}^{l_{\pm}-1}}{ \det\left[\int_{0}^{+\infty}|x|^{\alpha\pm\frac {1}{2}}e^{-nV(\sqrt{x})}x^{i+j}dx\right]_{i,j=0}^{l_{\pm}-1}},
 \end{equation}
with $l_+=[n/2]$ and $l_-=[(n+1)/2]$.
Then $\mathbb{P}_{\frac {\alpha}{2}\pm\frac 14}$ describe the gap probability of the following unitary ensembles of positive definite Hermitian matrices of size $n/2$
\begin{equation}\label{pGUE-positive}\frac{1}{Z_{n/2}}|\det(M)|^{\alpha\pm\frac 12}e^{-\frac {n}{2} \mathrm{ Tr}( \frac {M^2}{2}+g M)}dM,\end{equation}
where $n$ is even.  For $g_{cr}=-2$, the density of the limiting eigenvalue distribution is given by $\frac{1}{2\pi} \sqrt{x(4-x)}$, $x \in [0,4]$. At the origin,
the soft edge and the hard edge  coalesce.  In the critical regime $g=-2+\frac {2^{1/3}t}{n^{2/3}}$, the limiting eigenvalue correlation kernel is the $\textrm{P}_{34}$ kernel  $ K^{P34}_{\frac {\alpha}{2}\pm\frac 14, 0}(x,y; -2^{-1/3}t)$ in  \eqref{psi-kernel}; see \cite{ck-2}. Then, the gap probability of these unitary ensembles is given by
\begin{equation}\label{gap pro-pGUE-pos-p34}
\lim_{n\rightarrow\infty}\mathbb{P}_{\frac {\alpha}{2}\pm\frac 14}
=\det[I-K^{P34}_{\frac {\alpha}{2}\pm\frac 14,0,s'}],
 \end{equation}
 where $K^{P34}_{\frac {\alpha}{2}\pm\frac 14,0,s'}$ is the trace-class operator acting on $L^2(s', + \infty)$ with $s'=-2^{2/3}s^2$. Meanwhile, the large-$n$ limit of the gap probability for the unitary ensemble with quartic potential is expressed as the Fredholm determinant of the  $\textrm{P}_{2}$ kernel in \eqref{gap pro zero}. Then, the uniqueness of the limit of \eqref{gap pro zero-Hankel-split} as $n \to \infty$ implies
 \eqref{Fredholm det-P2-P34}. \hfill $\Box$

\medskip

Now, we are ready to derive the Tracy-Widom formula for the $\textrm{P}_{2}$-kernel determinant.

\medskip

\noindent\emph{Proof of Theorem \ref{thm:TW-P2}.} By Theorem \ref{thm:TW-P34} and Lemma \ref{lem:kp2-kp34}, we have  for real values  $t$ and $s\geq0$
\begin{equation*}
 \begin{split}\det[I-K^{P2}_{\alpha,s}]&= \exp\left(-\int_{t'}^{+\infty}(\sum_{i=1}^2v_i(\tau+s';s';\alpha-\frac {1}{2},0)-u(\tau;\alpha-\frac {1}{2},0))(\tau-t')d\tau\right)\\
 &\quad \exp\left(-\int_{t'}^{+\infty}(\sum_{i=1}^2v_i(\tau+s';s';\alpha+\frac {1}{2},0)-u(\tau;\alpha+\frac {1}{2},0))(\tau-t')d\tau\right),
 \end{split}
\end{equation*}
where $s'=-2^{2/3}s^2$, $t'=-2^{-1/3}t$,
 $v_i(x)=v_i(x;s;\alpha,\omega)$ are  the solutions to the coupled $\textrm{P}_{2}$ equations \eqref{int-equation v}  with the properties stated in Theorem \ref{theorem-v-1} and
$u(t)=u(t;\alpha,\omega)$ is the $\textrm{P}_{34}$ transcendent determined by the boundary condition \eqref{thm: painleve  positive infinity}.
By the Hamiltonian systems of equations \eqref{eq:Hamiltonian system} and the B\"{a}cklund transformations \eqref{BT-1},
we have
\begin{equation}\label{q-vi}
\sum_{i=1}^2(v_i(x;\alpha+\frac {1}{2})+v_i(x;\alpha-\frac{1}{2}))= w_2^2(x;\alpha)-(x-s).
\end{equation}
Similarly, using \eqref{BT-P34} and the relation between the $\textrm{P}_2$ and $\textrm{P}_{34}$ transcendents in \eqref{p34-p2-relation}, we get
\begin{equation} \label{u1u2-y-HM}
u(x;\alpha+\frac {1}{2},0)+u(x;\alpha-\frac{1}{2},0)=2^{2/3}y^2(-2^{1/3}x;\alpha)-x,
\end{equation}
where $y(x;\alpha)$ is the Hastings-McLeod solution to the  $\textrm{P}_{2}$ equation described in \eqref{HM solution}.
Combining the above three formulas, we arrive at \eqref{int-TW-pII}. This completes the proof of Theorem \ref{thm:TW-P2}. \hfill $\Box$


\section{Large gap asymptotics}\label{sec:large gap asy}

Finally, we apply the Deift-Zhou nonlinear steepest descent method to the RH problem for $\Phi$ and obtain the asymptotics of $\Phi(z; s+t,s)$ as $s\to-\infty$. 
Then, we evaluate the large gap asymptotics for the Fredholm determinants as $s\to-\infty$.

\subsection{Nonlinear steepest descent analysis of $\Phi$ as $s\rightarrow-\infty$} \label{Sec:steepest-large-gap}

As we need the asymptotics of $\Phi(z; s+t,s)$ for $t$ fixed and  $s\to-\infty$, let us focus on the case $s<0$.
In the first transformation, we rescale the variable and define:
\begin{equation}\label{S}
S(\xi)=\left(
         \begin{array}{cc}
           1 & 0 \\
           ir_2 & 1 \\
         \end{array}
       \right)
\Phi(|s|\xi)e^{|s|^{\frac 32}\lambda(\xi)\sigma_3} ,\end{equation}
where
$$\lambda(\xi):=\frac{2}{3}\xi^{\frac {3}{2}}+\frac {s+t}{|s|}\xi^{\frac {1}{2}}=\xi^{\frac12}\left(\frac23\xi-1-\frac{t} {s} \right).$$
\begin{figure}[h]
 \begin{center}
   \includegraphics[width=5 cm]{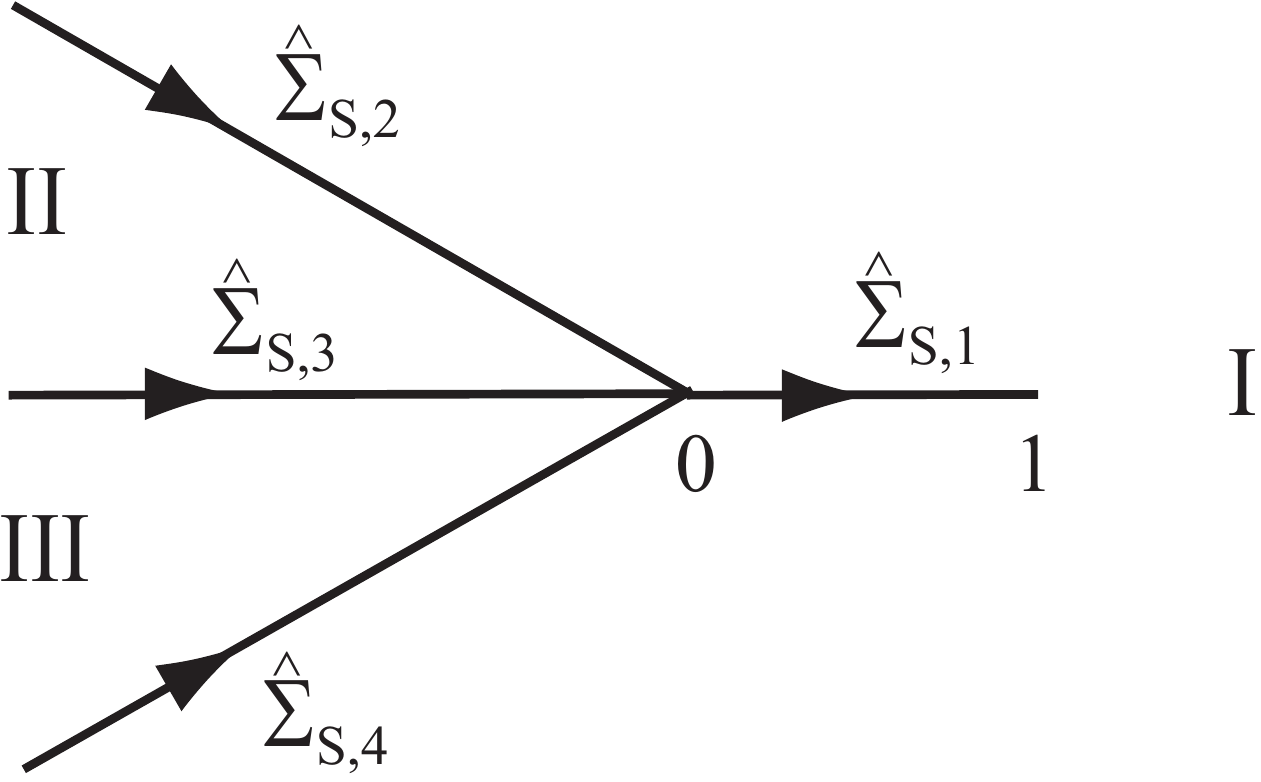} \end{center}
 \caption{\small{Contours  for the RH problem S.}}
 \label{figure-S}
 \end{figure}
Then $S$ satisfies the following RH problem.
\begin{rhp} The function $S(\xi)$ satisfies the following properties:
\begin{itemize}
  \item[(a)]  $S(\xi)$  is analytic in
  $\mathbb{C}\backslash \hat{\Sigma}_{S,i}$, where $\hat{\Sigma}_{S,i}$ are indicated in Fig. \ref{figure-S};

  \item[(b)]  $S(\xi)$ satisfies the jump condition
  \begin{equation}\label{S-jumps}
  S_+(\xi)=S_-(\xi) J_{S,i}(\xi) \qquad  \textrm{for } \xi\in\hat{\Sigma}_{S,i}
  \end{equation}
with
$$J_{S,1}(\xi)=\left(\begin{array}{cc}
                                                              e^{2\pi i \alpha} & 0 \\
                                                            0 &  e^{-2\pi i \alpha}
                                                           \end{array}\right),  \quad  J_{S,2}(\xi)=\left(\begin{array}{cc}
                                                             1 & 0 \\
                                                            e^{2|s|^{\frac 32}\lambda(\xi)} e^{2\pi i \alpha} & 1
                                                           \end{array}\right), $$
                                                           $$      J_{S,3}(\xi)=\left(\begin{array}{cc}
                                                             0 & 1 \\
                                                            -1 & 0
                                                           \end{array}\right),\quad J_{S,4}(\xi)=\left(\begin{array}{cc}
                                                             1 & 0 \\
                                                             e^{2|s|^{\frac 32}\lambda(\xi)}e^{-2\pi i \alpha} & 1
                                                           \end{array}\right);$$

  \item[(c)]   The asymptotic behavior of $S(\xi)$  at infinity
  \begin{equation}\label{S at infinity} S(\xi)=
    \left (I+O\left (\frac 1 \xi\right )\right)
  (|s|\xi)^{-\frac{1}{4}\sigma_3}\frac{I+i\sigma_1}{\sqrt{2}}  \qquad \textrm{as } \xi \rightarrow \infty ;
  \end{equation}

\item[(d)] The asymptotic behavior of $S(\xi)$  at $\xi=0,1$ is the same as that of $\Phi(z)$ at $z=0,|s|$, given in \eqref{Phi at 0} and \eqref{Phi at -s}.
\end{itemize}
\end{rhp}

Note that
\begin{equation}
  \Re \lambda(\xi)  < 0, \qquad \textrm{for } \xi \in \hat{\Sigma}_{S,2} \cup \hat{\Sigma}_{S,4},
\end{equation}
the off-diagonal entries of the jump matrices  are exponentially small as $|s| \to +\infty$. Neglecting the exponential small terms, we arrive at the following outer parametrix.

\begin{rhp}  The   function $S^{(\infty)}(\xi)$ satisfies the following properties:
\begin{itemize}
  \item[(a)]   $S^{(\infty)}(\xi)$ is analytic in
  $\mathbb{C}\backslash (-\infty,1]$;

  \item[(b)]  $S^{(\infty)}(\xi)$ satisfies the jump condition
  \begin{align}\label{Parametrix outside-jump}
  S^{(\infty)}_+(\xi)=S^{(\infty)}_-(\xi) \begin{cases}
    \left(\begin{array}{cc}
                                                             0 &  1\\
                                                            -1& 0
                                                           \end{array}\right), &   \xi\in (-\infty,0), \vspace{0.2cm}  \\
      e^{2\pi \alpha i\sigma_3}, & \xi\in (0,1);
  \end{cases}
\end{align}

  \item[(c)] At infinity, $S^{(\infty)}(\xi)$ satisfies the same asymptotics as $S(\xi)$ in \eqref{S at infinity}.

\end{itemize}
\end{rhp}

Similar to \textbf{RH problem \ref{rhp:outer1}}, the solution is given explicitly as
\begin{equation}\label{parametrix outside solution}
S^{(\infty)}(\xi)=\left(
                                                                         \begin{array}{cc}
                                                                           1 & 0 \\
                                                                           2\alpha i |s|^{1/2} & 1 \\
                                                                         \end{array}
                                                                       \right)
(|s|\xi)^{-\frac{1}{4}\sigma_3}\frac{I+i\sigma_1}{\sqrt{2}}
   h(\xi)^{\sigma_3}  ,
\end{equation}
where $ h(\xi)=\left(\frac {\sqrt{\xi}-1}{\sqrt{\xi}+1}\right)^{\alpha} $ is defined in \eqref{h-function} and satisfies the following expansion
\begin{equation}\label{h-expansion-infty}
h(\xi)=1-2\alpha\frac {1}{ \xi^{\frac 12}}+2\alpha^2\frac {1}{\xi}-\frac 23(\alpha+2\alpha^3) \frac {1}{\xi^{\frac 32}}+\cdots, \quad \textrm{ as $\xi \to \infty$.}
\end{equation}
For later use, we also compute the following refined expansion
\begin{equation}\label{Parametrix-outside-expand} S^{(\infty)}(\xi)=
    \left (I+\frac {1}{\xi}\left(
                             \begin{array}{cc}
                               2\alpha^2 & 2\alpha i|s|^{-1/2} \\
                             -\frac 23(\alpha-4\alpha^3)i |s|^{1/2} & -2\alpha^2\\
                             \end{array}
                           \right)
    +O\left (\frac 1 {\xi^2}\right )\right)
  (|s|\xi)^{-\frac{1}{4}\sigma_3}\frac{I+i\sigma_1}{\sqrt{2}}.
  \end{equation}

To construct the local parametrix near $\xi = 0$, we take the conformal mapping
\begin{equation}\label{conformal mapping}
\lambda^2(\xi)=\xi(1+\frac t {s}-\frac23\xi)^2,\quad |\xi|\leq\frac12.
\end{equation}
The local parametrix is constructed in terms of the Bessel functions as follows
\begin{equation}\label{parametrix local}
S^{(0)}(\xi)=E(\xi)Z_0(|s|^3\lambda(\xi)^2)\left\{\begin{array}{ll}
                                                     e^{(|s|^{\frac 32}\lambda(\xi)+\alpha\pi i)\sigma_3}, \quad &\arg \xi\in (0, \pi),\\
                                                     e^{(|s|^{\frac 32}\lambda(\xi)-\alpha\pi i)\sigma_3},\quad &\arg \xi\in (-\pi, 0),
                                                   \end{array}
\right.
    \end{equation}
where $E(\xi)$ is an analytic function in the disk $|\xi|\leq 1/2$, and the function $Z_0(z)$ is explicitly given in terms of the modified Bessel functions as follows:
 \begin{equation}\label{phi related to bessel}
Z_0(z)=\pi^{\frac 12 \sigma_3} \left\{
 \begin{array}{ll}
   \left(
                               \begin{array}{cc}
                                 I_{0}(\sqrt{z}) &\frac{i}{\pi} K_{0}(\sqrt{z})  \\
                                 \pi i\sqrt{z}    I'_{0}(\sqrt{z})& -\sqrt{z}    K'_{0}(\sqrt{z}) \\
                                 \end{array}
                             \right),  &  \quad \mbox{for} ~  z \in \mathrm{I}, \\[0.5cm]
  \left(
                               \begin{array}{cc}
                                 I_{0}(\sqrt{z}) &\frac{i}{\pi} K_{0}(\sqrt{z})  \\
                                 \pi i\sqrt{z}    I'_{0}(\sqrt{z})& -\sqrt{z}    K'_{0}(\sqrt{z}) \\
                                 \end{array}
                             \right)  \left(
                               \begin{array}{cc}
                                 1 & 0 \\
                                 -1 & 1 \\
                                 \end{array}
                             \right),    & \mbox{for}~  z\in \mathrm{II},  \\[0.5cm]
   \left(
                               \begin{array}{cc}
                                 I_{0}(\sqrt{z}) &\frac{i}{\pi} K_{0}(\sqrt{z})  \\
                                 \pi i\sqrt{z}    I'_{0}(\sqrt{z})& -\sqrt{z}    K'_{0}(\sqrt{z}) \\
                                 \end{array}
                             \right)\left(
                               \begin{array}{cc}
                                 1 & 0 \\
                                1 & 1 \\
                                 \end{array}
                             \right), & \mbox{for}~z\in \mathrm{III},
 \end{array}
 \right .
\end{equation}
for $\arg z\in (-\pi, \pi)$. It is well-known that, the above function $Z_0(z)$ satisfies the following model RH problem; see \cite{kmvv}.

\begin{figure}[ht]
 \begin{center}
   \includegraphics[width=5 cm]{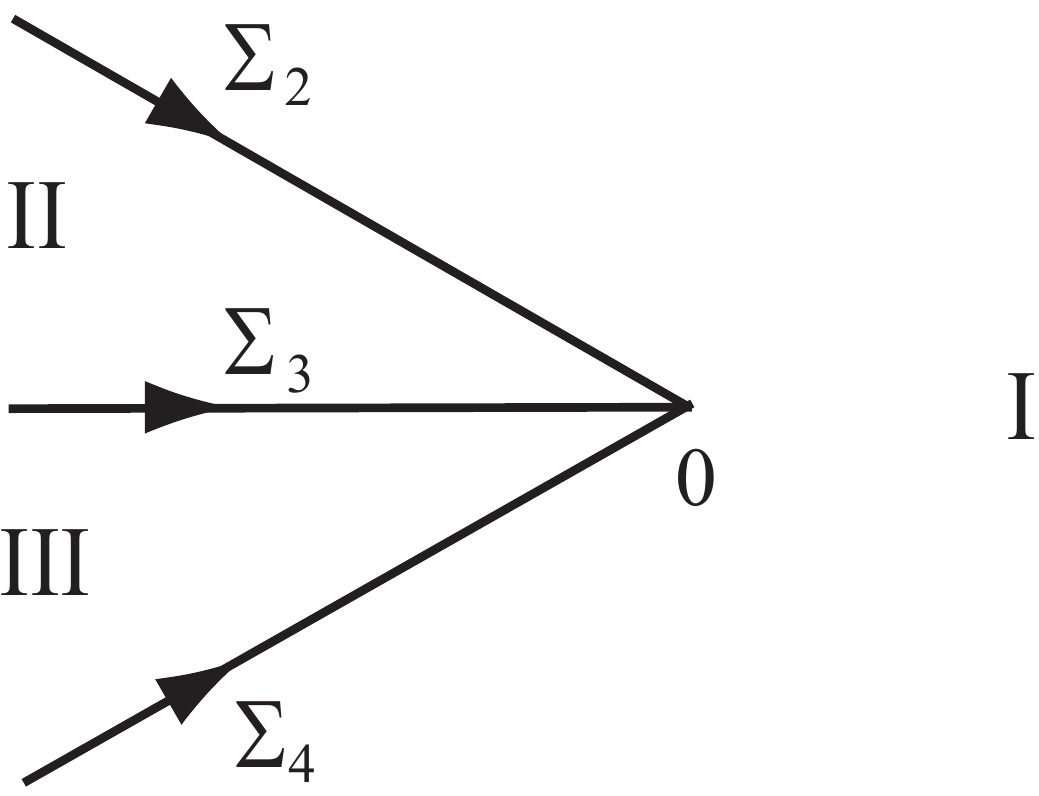} \end{center}
 \caption{\small{Contours  for the Bessel model RH problem of $Z_0$}}
 \label{figure: Bessel model RHP}
 \end{figure}
\begin{rhp} $Z_0(z)$ satisfies the following properties:
\begin{itemize}
  \item[(a)]   $Z_0(z)$ is analytic in
  $\mathbb{C}\backslash \Sigma_{i}$, where the contours $\Sigma_{i}$ are illustrated in Fig. \ref{figure: Bessel model RHP};

  \item[(b)]   $Z_0(z)$  satisfies the jump condition
  \begin{equation}
  Z_{0+}(z)=Z_{0-}(z) J_{i}(z), \quad  z\in \Sigma_i, \quad i=2,3,4,
  \end{equation}
where
$$J_{2}=\left(\begin{array}{cc}
                                                             1 & 0 \\
                                                         1 & 1
                                                           \end{array}\right), \quad J_{3}=\left(\begin{array}{cc}
                                                             1 & 0 \\
                                                            1 & 1
                                                           \end{array}\right), \quad J_{4}=\left(\begin{array}{cc}
                                                             0 & 1 \\
                                                            -1 & 0
                                                           \end{array}\right);$$

  \item[(c)]   The asymptotic behavior of $Z_0(z)$  at infinity
  \begin{equation}\label{phi at infinity-refine}
Z_0(z)=z^{-\frac{1}{4}\sigma_3}\frac{I+i\sigma_1}{\sqrt{2}}
   \left (I+\frac {1
   }{8\sqrt{z}}\left(
                     \begin{array}{cc}
                       -1 & -2i \\
                       -2i & 1 \\
                     \end{array}
                   \right)+O\left (\frac 1{z}\right )\right)e^{\sqrt{z}\sigma_3} \qquad \textrm{as } z \rightarrow \infty .
 \end{equation}
\end{itemize}
\end{rhp}

To match the local parametrix $S^{(0)}(\xi)$ with the outer parametrix $S^{(\infty)}(\xi)$ in \eqref{parametrix outside solution} on $|\xi|=\frac12$, we choose the analytic pre-factor $E(\xi)$ in \eqref{parametrix local} as
\begin{equation}\label{E-1 matrix}
E(\xi)= S^{(\infty)}(\xi)\left\{\begin{array}{ll}
                                                     e^{ -\pi i\alpha \sigma_3}\frac{I-i\sigma_1}{\sqrt{2}} \left ( |s|^3\lambda(\xi)^2\right )^{\frac{1}{4}\sigma_3}, \quad& \arg \xi\in (0, \pi),\\[.2cm]
                                                     e^{ \pi i\alpha\sigma_3}\frac{I-i\sigma_1}{\sqrt{2}}\left ( |s|^3\lambda(\xi)^2\right )^{\frac{1}{4}\sigma_3}, \quad&   \arg \xi\in (-\pi, 0).
                                                   \end{array}
\right.
\end{equation}
Then,  $E(\xi)$ is analytic in  the disk $|\xi|\leq 1/2$. Moreover the following matching condition is fulfilled
\begin{equation}\label{matching condition}
 S^{(0)}(\xi)=(I+O(1/s) )S^{(\infty)}(\xi),\quad |\xi|=\frac12.
\end{equation}

For later use, we compute $E(0)$ and $E'(0)$. Let
$$\hat{E}(\xi)=\left(
                                                                         \begin{array}{cc}
                                                                           1 & 0 \\
                                                                           -2\alpha i |s|^{1/2} & 1 \\
                                                                         \end{array}
                                                                       \right)E(\xi),$$
then we get from \eqref{E-1 matrix}
\begin{equation}
\hat{E}(\xi)=\frac 12(\hat{h}(\xi)+\hat{h}(\xi)^{-1})(s^2\lambda(\xi)^2/\xi)^{\frac14\sigma_3}+\frac 12 (|s|\xi)^{-\frac 14\sigma_3}(\hat{h}(\xi)-\hat{h}(\xi)^{-1})\sigma_2(|s|^3\lambda(\xi)^2)^{\frac 14\sigma_3},
\end{equation}
where
\begin{equation}\label{h-hat}
\hat{h}(\xi):=\left(\frac {1-\sqrt{\xi}}{1+\sqrt{\xi}}\right)^{\alpha}=1-2\alpha \xi^{\frac 12}+2\alpha^2\xi-\frac 23(\alpha+2\alpha^3) \xi^{\frac 32}+O(\xi)^2
\end{equation}
as $\xi \to 0.$ Thus, we find
\begin{equation}\label{E at zero-1}
 \hat{E}(0)=\left(\begin{array}{cc}
    |s|^{\frac {1}{2}}(1+\frac {t}{s})^{\frac {1}{2}} & 2\alpha i|s|^{-1}(1+\frac {t}{s})^{-\frac{1}{2}} \\
    0& |s|^{-\frac {1}{2}}(1+\frac {t}{s})^{-\frac {1}{2}} \\
  \end{array}
\right), \end{equation}
and
\begin{equation}\label{E at zero-2}
\hat{E}_{21}'(0)=
    - 2\alpha i|s|(1+\frac {t}{s})^{1/2} .
\end{equation}

In the final transformation, we define
\begin{equation}\label{R}
R(\xi)= \begin{cases}
  S(\xi)(S^{(\infty)}(\xi))^{-1}, &  |\xi|>1/2 \\
 S(\xi)(S^{(0)}(\xi))^{-1}, & |\xi|< 1/2.
\end{cases}
\end{equation}
Then, $R$ satisfies the following RH problem.

\begin{rhp} The function $R(\xi)$ defined in \eqref{R} satisfies the following properties:
\begin{itemize}
  \item[(a)]   $R(\xi)$ is analytic in
  $\mathbb{C}\backslash \Sigma_{R}$;

  \item[(b)]  $R(\xi)$ satisfies the jump condition  $ R_+(\xi)=R_-(\xi)J_R(\xi)$,
  $$\begin{array}{ll}
                             J_R(\xi)=S^{(0)}(\xi)(S^{(\infty)}(\xi))^{-1}, &  |\xi|=\frac 12,\\
                             J_R(\xi)=S^{(\infty)}(\xi) J_S (S^{(\infty)}(\xi))^{-1}, &  \xi\in \Sigma_{R}\setminus \{|\xi|=\frac 12\};
                             \end{array}$$

  \item[(c)] As $\xi \to \infty$, $R(\xi)=I+O(1/\xi).$
\end{itemize}
\end{rhp}
The jump  $J_R-I$ is exponentially small for $\xi \in \Sigma_{R}\setminus \{|\xi|=\frac 12\}$. On the circle $|\xi|=\frac 12$, from  asymptotics of the Bessel parametrix in \eqref{phi at infinity-refine}, we have
\begin{equation}\label{R-jump-estimate}
J_R=|s|^{-\frac 14\sigma_3}\biggl(I+|s|^{-3/2}\left(
                                                                         \begin{array}{cc}
                                                                           1 & 0 \\
                                                                           2\alpha i  & 1 \\
                                                                         \end{array}
                                                                       \right)
J_1(\xi)\left(
                                                                         \begin{array}{cc}
                                                                           1 & 0 \\
                                                                           -2\alpha i & 1 \\
                                                                         \end{array}
                                                                       \right)+O((-s)^{-5/2})\biggr)|s|^{\frac 14\sigma_3}
,
 \end{equation}
where $J_1(\xi)$ is given by
\begin{equation}J_{1}(\xi)=\left(
  \begin{array}{cc}
    \frac {-\hat{h}(\xi)^2+\hat{h}(\xi)^{-2}}{8\xi^{1/2} (1+\frac {t}{s}-\frac {2}{3}\xi)} & \frac{-i((\hat{h}(\xi)+\hat{h}(\xi)^{-1})^2-3)}{8\xi (1+\frac {t}{s}-\frac {2}{3}\xi)}  \\
    \frac{-i((\hat{h}(\xi)+\hat{h}(\xi)^{-1})^2-1)}{8 (1+\frac {t}{s}-\frac {2}{3}\xi)} &\frac {\hat{h}(\xi)^2-\hat{h}(\xi)^{-2}}{8\xi^{1/2} (1+\frac {t}{s}-\frac {2}{3}\xi)} \\
  \end{array}
\right)
\end{equation}
with $\hat{h}(\xi)$ defined in \eqref{h-hat}. Note that, the functions $\hat{h}(\xi)+\hat{h}(\xi)^{-1}$, $(\hat{h}(\xi)-\hat{h}(\xi)^{-1})/\xi^{1/2}$
are analytic near the origin and satisfy
\begin{equation}\label{h-estimate-2}\begin{array}{l}
    \hat{h}(\xi)+\hat{h}(\xi)^{-1}=2+4\alpha^2\xi+O(\xi^2), \\
    (\hat{h}(\xi)-\hat{h}(\xi)^{-1})/\xi^{1/2}=-4\alpha-\frac 43 (\alpha+2\alpha^3)\xi+O(\xi^2).
  \end{array}
\end{equation}
From the expansion of the jump $J_R$ in \eqref{R-jump-estimate}, we get the following expansion
\begin{equation}\label{R-expand}R(\xi)=|s|^{-\frac 14\sigma_3}\left(I
+\frac 1 {|s|^{3/2}}\left(
\begin{array}{cc}
1 & 0 \\
 2\alpha i  & 1 \\
 \end{array}
 \right) R_1(\xi)\left(
  \begin{array}{cc}
  1 & 0 \\
  -2\alpha i  & 1 \\
  \end{array}
  \right)+O(s^{-3})\right)|s|^{\frac 14\sigma_3},
\end{equation}
where $R_1(\xi)$ satisfies a RH problem as follows.
\begin{rhp} The function $R_1(\xi)$ satisfies the following properties:
  \begin{itemize}
  \item[(a)]   $R_1(\xi)$ is analytic in   $\mathbb{C}\backslash \{|\xi|=1/2\}$;

  \item[(b)]  $R_{1+}(\xi)-R_{1-}(\xi)=J_1(\xi)$, for  $|\xi|=1/2;$

  \item[(c)] As $\xi \to \infty$, $R_1(\xi)=O(1/\xi)$.
\end{itemize}
\end{rhp}

Performing a residue computation, we get the solution to the above RH problem
\begin{equation}\label{R-1}
R_1(\xi)=-\left(
  \begin{array}{cc}
    \frac {-\hat{h}(\xi)^2+\hat{h}(\xi)^{-2}}{8\xi^{1/2} (1+\frac {t}{s}-\frac {2}{3}\xi)} & \frac{-i((\hat{h}(\xi)+\hat{h}(\xi)^{-1})^2-4+\frac 23\xi(1+\frac ts)^{-1})}{8\xi  (1+\frac {t}{s}-\frac {2}{3}\xi)}  \\
    \frac{-i((\hat{h}(\xi)+\hat{h}(\xi)^{-1})^2-1)}{8 (1+\frac {t}{s}-\frac {2}{3}\xi)} &\frac {\hat{h}(\xi)^2-\hat{h}(\xi)^{-2}}{8\xi^{1/2} (1+\frac {t}{s}-\frac {2}{3}\xi)} \\
  \end{array}
\right) \quad \textrm{for  $|\xi|< \frac 12$}
\end{equation}
and
\begin{equation}\label{R-1-outside} R_1(\xi)=\left(
                                                                       \begin{array}{cc}
                                                                         0 & \frac {1}{8i\xi(1+\frac ts)} \\
                                                                         0 & 0 \\
                                                                       \end{array}
                                                                     \right) \quad \textrm{ for  $|\xi|>\frac 12$. }
\end{equation}
 Particularly,  we get by substituting \eqref{h-estimate-2} into \eqref{R-1}
\begin{equation}\label{R-1 at zero}
 (R_1'(0))_{21}=
             i(2\alpha^2(1+\frac ts)^{-1}+\frac 14(1+\frac ts)^{-2}) .
\end{equation}
Tracing back the transformation $ \Phi\to S\to R$ gives
\begin{equation}\label{Phi and R-P}
\Phi(|s|\xi)=\left(
                                                                              \begin{array}{cc}
                                                                                1 & 0 \\
                                                                               -i r_2 & 1 \\
                                                                              \end{array}
                                                                            \right)
R(\xi) S^{(\infty)}(\xi)e^{-|s|^{3/2}\lambda(\xi)\sigma_3}.\end{equation}
By \eqref{R-expand} and \eqref{R-1-outside}, we get
\begin{equation}\label{R-expand-infty}
R(\xi)=I+\frac 1{8(1+\frac ts)\xi}\left(
\begin{array}{cc}
-2\alpha  |s|^{-3/2} & -i|s|^{-2} \\
 -4\alpha^2 i|s|^{-1}  & 2\alpha |s|^{-3/2} \\
 \end{array}
 \right) +O(|s|^{-5/2}),
\end{equation}
for $|\xi|>1/2$.

From the steepest descent analysis done above, we are able to derive the following asymptotics for the functions $v_i$, $w_i$ and the Hamiltonian $H$. These results will be used in the derivation of the large gap asymptotics for the Fredholm determinants in the next subsection.
\begin{pro}\label{Pro: asy phi negative infty}
For fixed $t$, we have the following asymptotics, as $s\to -\infty$
\begin{equation}\label{eq:asy-v-1-negative infty}
v_1(s+t;s)=-\frac{s+t}{2}+O(s^{-1/2}),
\end{equation}
\begin{equation}\label{eq:asy-v-2-negative infty}
v_2(s+t;s)=\alpha \frac{1}{\sqrt{|s|}}+O(1/s),
\end{equation}
\begin{equation}\label{eq:asy-w-1-negative infty}
w_1(s+t;s)=\frac{1}{2(s+t)}+O(s^{-3/2}),
\end{equation}
\begin{equation}\label{eq:asy-w-2-negative infty}
w_2(s+t;s)=-\sqrt{|s|}+O(s^{-1/2}),
\end{equation}
\begin{equation}\label{eq:asy-H-negative infty}
H(s+t;s)=\frac{(s+t)^2}{4}-2\alpha\sqrt{|s|}-\frac{1}{8}\frac{1}{s+t}+O(s^{-3/2}),
\end{equation}
\end{pro}
\begin{proof}
Recall that the outer parametrix $S^{(\infty)}(\xi)$ is given explicitly in \eqref{parametrix outside solution}. From the relation \eqref{S}, the asymptotics of $\Phi(z)$ at $z=-s$ are obtained from $S^{(\infty)}(1)$. More precisely, we have
\begin{equation}\label{eq:asy-phi-0-negative infty}
\ln\left(\Phi^{(-s)}_0\right)_{11}(s+t;s)=\frac{1}{3}|s|^{3/2}-t\sqrt{|s|}-(\alpha+\frac{1}{4})\ln|s|-(2\alpha+\frac{1}{2})\ln 2+O(1/s),
\end{equation}
and
\begin{equation}\label{eq:asy-phi-1-negative infty}
\left(\Phi^{(-s)}_1\right)_{11}(s+t;s)=-\frac{1}{2}\sqrt{|s|}-\frac{1}{2}\frac{t}{\sqrt{|s|}}-\frac{\alpha}{2}\frac{1}{|s|}+O(s^{-3/2}).
\end{equation}
Then, the expansions for $v_2$ in \eqref{eq:asy-v-2-negative infty} and $w_2$ in \eqref{eq:asy-w-2-negative infty} follow from the above formulas and the differential identities in \eqref{eq:integral w-2} and \eqref{eq:integral-v-2}. Based on the relations among $v_i$, $w_i$ and $H$ in \eqref{eq:Hamiltonian system} and \eqref{def:Hamiltonian-CPII}, the other asymptotic expansions follow directly.
\end{proof}

\subsection{Large gap asymptotics: proof of Theorem \ref{theorem-large gap asy-P34} }

In Theorem \ref{thm:TW-P34}, we have successfully expressed the Fredholm determinant of the $\textrm{P}_{34}$ kernel as an integral of solutions to the coupled $\textrm{P}_{2}$ equations \eqref{int-equation v}. This important representation can be further rewritten in terms of the tau function for the coupled $\textrm{P}_{2}$ system. Quite recently, in \cite{Bot:Its:Prok2017,Its:Lis:Pro,Its:Lis:Tyk2015,Its:Prok2016}, the asymptotics of the tau functions for the classical Painlev\'e equations have been successfully evaluated including the constant terms. In this section, we will derive the large gap asymptotics by evaluating the asymptotic of the tau function for the coupled $\textrm{P}_{2}$ system.

From \eqref{equation r and v} and \eqref{eq:Hamiltotion -r-2}, we have
\begin{equation} \label{Hamilton-prime}
  H'(x) = -v_1(x) - v_2(x).
\end{equation}
Then, the Tracy-Widom formula \eqref{int-TW-p34-1} obtained in Theorem \ref{thm:TW-P34} can be written as
\begin{equation*}
  \ln\det[I-K^{P34}_{\alpha, \omega, s}] = \int_t^{+\infty} \biggl( H'(x+s) + \frac{\alpha}{\sqrt{x}} - \frac{\alpha^2}{x^2} \biggr) (x-t) dx + \int_t^{+\infty} \biggl( u(x) - \frac{\alpha}{\sqrt{x}} + \frac{\alpha^2}{x^2} \biggr) (x-t) dx.
\end{equation*}
Note that both integrals above are convergent due to the asymptotics of $u(x)$ in \eqref{thm: painleve  positive infinity} and $H(x)$ in \eqref{eq:asy-H}. Moreover, an integration by parts of the first integral gives us
\begin{equation} \label{p34-fred-new-expression}
  \ln\det[I-K^{P34}_{\alpha, \omega, s}] = - \int_{s+t}^{+\infty} \biggl( H(\tau) + 2\alpha \sqrt{|\tau -s |} + \frac{\alpha^2}{\tau -s } \biggr)  d\tau +\int_t^{+\infty}(\tau-t)\left(u(\tau)-\frac {\alpha}{|\tau|^{1/2}}+\frac {\alpha^2}{\tau^2}\right )d\tau.
\end{equation}
Note that, according to  the theory of isomonodromic tau-functions in the sense of Jimbo-Miwa-Ueno \cite{Jimbo:Miwa:Ueno1981}, the tau function can be defined as $d_x \ln \tau=H(x)dx$, where $H$ is the Hamiltonian. Now, the only task for us is to compute the asymptotics of the first integral as $s \to -\infty$.

From the Hamiltonian system \eqref{eq:Hamiltonian system}, we have
\begin{equation}\label{eq:total differential}
H=  \frac{1}{3}(v_1w_1+v_2w_2+2xH)_{x}+2\alpha w_2+\frac {2}{3}sv_2- (v_1w_{1x}+v_2w_{2x}+H).
\end{equation}
Next, let us integrate both sides of the above formula from $s+t$ to $+\infty$. To ensure the convergence, we need to add a few terms according to the asymptotics of the functions $v_i$, $w_i$ and $H$ in Theorem \ref{theorem-v-1}. More precisely, we get
\begin{align}\nonumber
&-\int_{s+t} ^{\infty}\biggl(H(\tau)+2\alpha\sqrt{|\tau-s|}+\frac{\alpha^2}{\tau-s} \biggr)d\tau = -2\alpha I_1(s+t;\alpha,\omega) -\frac {2s}{3}I_2(s+t;\alpha,\omega)+I_3(s+t;\alpha,\omega) \\
&\hspace{2cm} +\frac{1}{3}\left(v_1w_1+v_2w_2+\alpha+2(s+t)H+4\alpha s \, \mathrm{sgn}(t) \sqrt{|t|}+4\alpha \, \mathrm{sgn}(t) |t|^{3/2}+2\alpha^2\right) , \label{eq:total integral H-1}
\end{align}
where
\begin{align}
I_1(s+t;\alpha,\omega)&=\int_{s+t}^{\infty}\biggl(w_2(\tau)+\sqrt{|\tau-s|}+\frac{\alpha+\frac{1}{4}}{\tau-s} \biggr)d\tau \nonumber \\
&=-\frac{2}{3}\, \mathrm{sgn}(t) |t|^{3/2}-(\alpha+\frac{1}{4})\ln |t|-(2\alpha+\frac{1}{2})\ln 2-\ln\left(\Phi^{(-s)}_0\right)_{11}(s+t;s), \label{eq:total integral w-2}
\end{align}
\begin{equation}\label{eq:total integral-v-2}
I_2(s+t;\alpha,\omega)=\int_{s+t}^{\infty} \biggl( v_2(\tau)-\frac {\alpha}{\sqrt{|\tau-s|}} \biggr)d\tau=2\alpha \left(\left(\Phi^{(-s)}_1\right)_{11}(s+t;s)+ \mathrm{sgn}(t) \sqrt{|t|}\right),
\end{equation}
and
\begin{equation}\label{def:I-3}
I_3(s+t;\alpha,\omega)=\int_{s+t}^{\infty} \biggl(v_1(\tau)w_{1x}(\tau)+v_2(\tau)w_{2x}(\tau)+H(\tau)+2\alpha\sqrt{|\tau-s|}+\frac{\alpha(2\alpha+1)}{2(\tau-s)} \biggr)d\tau.
\end{equation}
Here, to obtain the explicit expressions of $I_1(s+t;\alpha,\omega)$ in \eqref{eq:total integral w-2} and $I_2(s+t;\alpha,\omega)$ in \eqref{eq:total integral-v-2}, we use the differential identity in \eqref{eq:integral w-2} and \eqref{eq:integral-v-2}, as well as the asymptotics of $\left(\Phi^{(-s)}_0\right)_{11}$ and $\left(\Phi^{(-s)}_1\right)_{11}$ in \eqref{eq:asy-phi-0} and \eqref{eq:asy-phi-1}.

Although the exact expression of the integral $I_3(s+t;\alpha,\omega)$ is unavailable now, we may consider its derivative with respect to the parameter $\alpha$. From the Hamiltonian system \eqref{eq:Hamiltonian system} and \eqref{def:Hamiltonian-CPII}, we have
\begin{equation}\label{eq:differential alpha}
(v_1w_{1x}+v_2w_{2x}+H)_{\alpha}=(v_1w_{1\alpha}+v_2w_{2\alpha})_{x}+2w_2.
\end{equation}
The above formula implies
\begin{equation}\label{def:I-3-deff}
\frac{\partial}{\partial\alpha}I_3(s+t;\alpha,\omega)=-\left(v_1(s+t)w_{1\alpha}(s+t)+v_2(s+t)w_{2\alpha}(s+t)\right)+2I_1(s+t;\alpha,\omega).
\end{equation}
For $s<0$, the function $\Phi(z;x,s)$ is independent of the parameter $\omega$; see the model Riemann-Hilbert problem for $\Phi$ in Sec. \ref{sec:rhp-phi}. Then, $I_3(s;\alpha,\omega)$ is also independent of the parameter $\omega$ for negative $s$. For $\alpha=0$ and $\omega=1$, we have $v_2=0$ and $v_1(x)=y^2(x;0)$, where $y(x;0)$ is the classical Hastings-McLeod solution to $\textrm{P}_2$ equation; see Remark \ref{remark:p2HM}. Moreover, the Hamiltonian $H$ in \eqref{def:Hamiltonian-CPII} is reduced to the Hamiltonian $\mathcal{H}$ for the $\textrm{P}_2$ equation. By \eqref{eq:total integral H-1}, we have
\begin{align}
I_3(x;0,1)=-\int_x ^{\infty}\mathcal{H}(\tau)d\tau-\frac{1}{3}(v_1(x)w_1(x)+2x \mathcal{H}(x)) \nonumber
\end{align}
Note from \eqref{Hamilton-prime}, $\mathcal{H}'(x) = - v_1(x)$. An integration by parts gives us
\begin{equation}
  -\int_x ^{\infty}\mathcal{H}(\tau)d\tau = - \int_x^{\infty} (\tau - x ) \mathcal{H}'(\tau) d \tau = -\int_x^{\infty} (\tau - x ) y^2(\tau;0) d \tau,
\end{equation}
which is exactly the exponent of the Tracy-Widom distribution in \eqref{Tracy-Widom formula}. Therefore, using \eqref{large gap asy-airy}, as well as \eqref{eq:asy-v-1-negative infty} and \eqref{eq:asy-w-1-negative infty}, we have
\begin{equation} \label{def:I-initial value}
  I_3(x;0,1)=-\frac{|x|^3}{12}-\frac{1}{8}\ln|x|+\zeta'(-1)+\frac{1}{24}\ln2+\frac{1}{6}+o(1), \quad \textrm{as } x \to -\infty.
\end{equation}

Recalling the approximations in \eqref{eq:asy-v-1-negative infty}-\eqref{eq:asy-phi-1-negative infty}, we have asymptotics of the  integrals in \eqref{eq:total integral w-2}, \eqref{eq:total integral-v-2} and \eqref{def:I-3-deff}
\begin{align}\label{eq: I-1 asy}
I_1(s+t;\alpha,\omega)=-\frac{1}{3}|s|^{3/2}+t\sqrt{|s|}+(\alpha+\frac{1}{4})\ln|s|-\frac{2}{3}\mathrm{sgn}(t) |t|^{3/2}-(\alpha+\frac{1}{4})\ln |t|+O(1/s),
\end{align}
\begin{equation}\label{eq: I-2 asy}
I_2(s+t;\alpha,\omega)=2\alpha\left(-\frac{1}{2}\sqrt{|s|}+\mathrm{sgn}(t) \sqrt{|t|}-\frac{1}{2}\frac{t}{\sqrt{|s|}}-
\frac{\alpha}{2}\frac{1}{|s|}+O(s^{-3/2})\right),
\end{equation}
\begin{equation}\label{eq: I-3 diff asy}
\frac{\partial}{\partial\alpha}I_3(s+t;\alpha,\omega)=-\frac{2}{3}|s|^{3/2}+2t\sqrt{|s|}+2(\alpha+\frac{1}{4})\ln|s|
-\frac{4}{3}\mathrm{sgn}(t) |t|^{3/2}-2(\alpha+\frac{1}{4})\ln |t|+o(1).
\end{equation}
Integrating the above formula about $\alpha$, we have
\begin{equation}\label{eq: I-3 asy}
\begin{split}
  I_3(s+t;\alpha,\omega)&= I_3(s+t;0,\omega)-\frac{2\alpha}{3}|s|^{3/2}+2\alpha t\sqrt{|s|}+(\alpha^2+\frac{\alpha}{2})\ln|s| \\
  & \quad -\frac{4\alpha}{3}\mathrm{sgn}(t)  |t|^{3/2}-(\alpha^2+\frac{\alpha}{2})\ln |t|+o(1).
\end{split}
\end{equation}
Substituting the approximations \eqref{def:I-initial value}-\eqref{eq: I-2 asy} and \eqref{eq: I-3 asy} into \eqref{eq:total integral H-1}, we obtain the asymptotics
\begin{align}\label{eq:total integral H asy}\nonumber
& -\int_{s+t} ^{\infty}\biggl( H(\tau)+2\alpha\sqrt{\tau-s}+\frac{\alpha^2}{\tau-s} \biggr)d\tau =-\frac 1{12}|s+t|^3 +\frac 23 \alpha|s|^{\frac 32}-2\alpha|s|^{1/2}t\\
 &\hspace{3cm} -(\alpha^2+\frac {1}{8}) \ln|s+t|+\frac 43 \alpha \, \mathrm{sgn}(t) |t|^{\frac 32}+\alpha^2\ln|t|+c_0+o(1),
\end{align}
where the constant $c_0$ is given in \eqref{def:constant-0}.

This completes the proof of Theorem \ref{theorem-large gap asy-P34}.

\subsection{Large gap asymptotics: proof of Theorem \ref{theorem-large gap asy-P2}}

From Theorem \ref{theorem-large gap asy-P34} and the relation \eqref{Fredholm det-P2-P34}, we obtain the following asymptotic
 expansion for the Fredholm determinant of the $\textrm{P}_{2}$ kernel as $s\to +\infty$
\begin{eqnarray}
 \ln \det[I-K^{P2}_{\alpha,s}]&=& -\frac {2}{3}(s^2+\frac {t}{2})^3 +\frac 43 \alpha s^{3}+2\alpha st-(\alpha^2+\frac {3}{4}) \ln s +\frac{2\sqrt{2}}{3}\alpha \, \mathrm{sgn}(t) |t|^{3/2} \nonumber \\
 &&+(\frac{\alpha^2}{2}+\frac{1}{8})\ln|t| +c_1 +I^*+o(1),
\end{eqnarray}
where $c_1$ is given in \eqref{def:constant-1}, and the integral $I^*$ is given by
\begin{equation*}
 \begin{split}
   I^* &=  \int_{-2^{-\frac{1}{3}}t}^{+ \infty} (\tau + 2^{-\frac{1}{3}} t) \biggl( u(\tau, \alpha-\frac {1}{2},0) + u(\tau, \alpha+\frac {1}{2},0) -\frac{\alpha}{\sqrt{|\tau|}}+\frac{\frac{\alpha^2}{2}+\frac{1}{8}}{\tau^2} \biggr) d\tau \\
  &  =  - \int_{- \infty}^t (\tau - t) \biggl(  \frac{u(-2^{-\frac{1}{3}} \tau, \alpha-\frac {1}{2},0)}{2^{\frac{2}{3}}} + \frac{u(-2^{-\frac{1}{3}}\tau, \alpha+\frac {1}{2},0)}{2^{\frac{2}{3}}} -\frac{\alpha}{\sqrt{|2\tau|}}+\frac{\frac{\alpha^2}{2}+\frac{1}{8}}{\tau^2} \biggr) d\tau.
 \end{split}
\end{equation*}
Then, the above formula and \eqref{u1u2-y-HM} gives us the integral in \eqref{lag gap asy-P2}.

This completes the proof of Theorem \ref{theorem-large gap asy-P2}.

\section*{Acknowledgements}
 We are grateful to Tom Claeys for  useful discussions. Shuai-Xia Xu was partially supported by the National Natural Science Foundation
of China under grant number 11571376, GuangDong Natural Science
Foundation under grant number 2014A030313176. Dan Dai was partially supported by grants from the Research Grants Council of the Hong
Kong Special Administrative Region, China (Project No. CityU 11300814, CityU 11300115,
CityU 11303016).


\begin{appendices}

\section{Asymptotics of $v_i(x)$ as $x\rightarrow-\infty$ }\label{sec:asy of v1--infty} 

In this appendix, we derive the following asymptotics of $v_i(x)$ as $x\to -\infty$.
Similar to Section \ref{sec:asy of v1-infty}, the Deift-Zhou nonlinear steepest descend method is applied to obtain the asymptotics. Depending on the sign of $s$, we divide our computations into two parts.

\subsection{Case I: $s<0$}

\subsubsection{Nonlinear steepest descent analysis of $\Phi$ as $x \to -\infty$}

We first remove the exponential term in the large-$z$ expansion $e^{-\theta(z,x)\sigma_3}$ of $\Phi(z)$ in \eqref{Phi at infinity} by introducing the following transformation:
\begin{equation}\label{A}
A(z)=\left(
         \begin{array}{cc}
           1 & 0 \\
           ir_2 & 1 \\
         \end{array}
       \right)
\Phi(z)e^{\theta(z,x)\sigma_3},
\end{equation}
 \begin{figure}[h]
 \begin{center}
   \includegraphics[width=5cm]{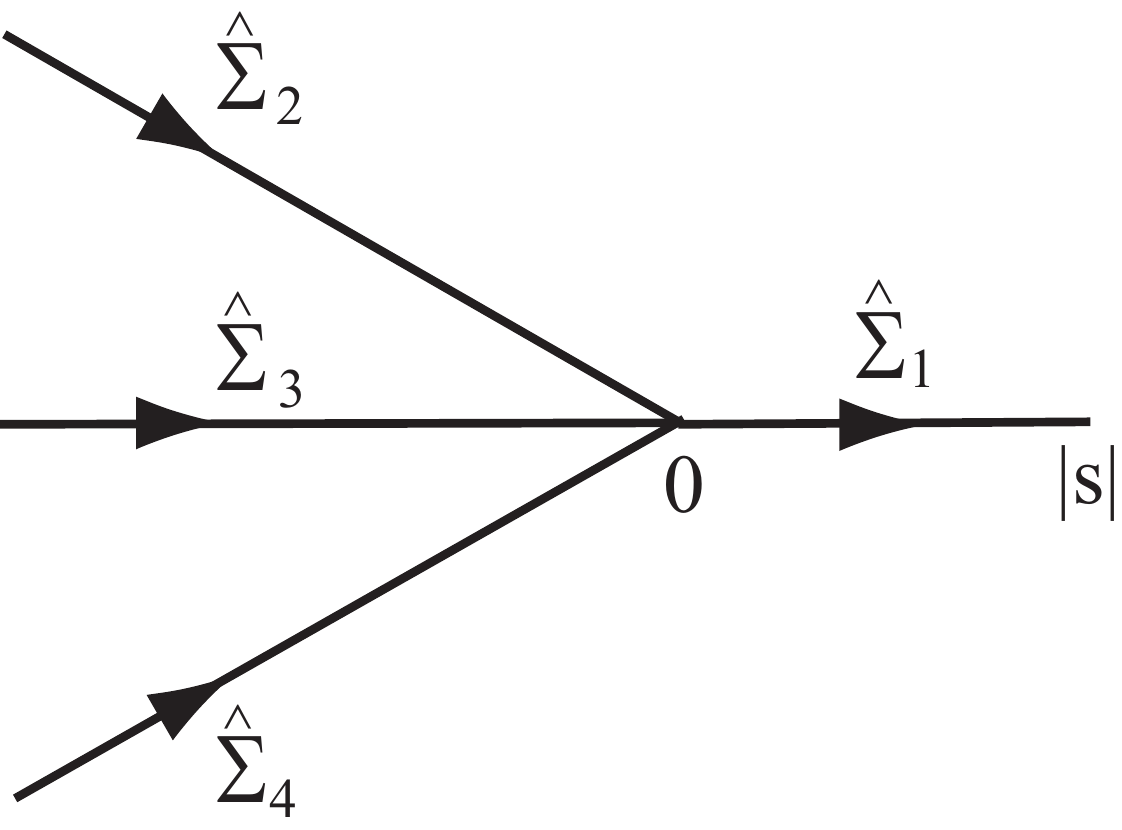} \end{center}
 \caption{\small{Regions and contours  for  $A$ with $s<0$.}}
 \label{Figure-s negative }
 \end{figure}
Then, we have a RH problem as follows.
\begin{rhp} The function $A(z)$ satisfies the following properties:
\begin{itemize}
  \item[(a)]   $A(z)$ is analytic in
  $\mathbb{C}\backslash \Sigma_A $,  where $\Sigma_A = \cup_{i=1}^4 \hat{\Sigma}_{i}$ is indicated in Fig. \ref{Figure-s negative };

  \item[(b)]  $A(z)$ satisfies the jump condition
  \begin{equation}\label{A-jumps}
     A_+(z)=A_-(z) J_{A,i}(z) \qquad \textrm{for } z\in\hat{\Sigma}_{i}
  \end{equation}
with
$$J_{A,1}=\left(\begin{array}{cc}
                                                              e^{2\pi i \alpha} & 0 \\
                                                            0 &  e^{-2\pi i \alpha}
                                                           \end{array}\right),  \quad  J_{A,2}=\left(\begin{array}{cc}
                                                             1 & 0 \\
                                                            e^{2\theta(z,x)} e^{2\pi i \alpha} & 1
                                                           \end{array}\right), $$
                                                           $$      J_{A,3}=\left(\begin{array}{cc}
                                                             0 & 1 \\
                                                            -1 & 0
                                                           \end{array}\right),\quad J_{A,4}=\left(\begin{array}{cc}
                                                             1 & 0 \\
                                                             e^{2\theta(z,x)}e^{-2\pi i \alpha} & 1
                                                           \end{array}\right);$$

  \item[(c)]   The asymptotic behavior of $A(z)$  at infinity:
  \begin{equation}\label{A at infinity}
  A(z)=
    \left (I+O\left (\frac 1 z\right )\right)
  z^{-\frac{1}{4}\sigma_3}\frac{I+i\sigma_1}{\sqrt{2}} \qquad \textrm{as } z \rightarrow \infty  ;
  \end{equation}

\item[(d)] $A(z)$ has the same  behavior near $z=0,-s$  as  $\Phi(z)$, given in \eqref{Phi at 0} and \eqref{Phi at -s}.
\end{itemize}
\end{rhp}

Because
\begin{equation} \label{theta-z-realpart}
  \Re \theta(z,x) = \Re \biggl(\frac{2}{3}z^{\frac{3}{2}}+x z^{\frac{1}{2}}\biggr) < 0, \qquad \textrm{for } z \in \mathbb{C}^{\pm},
\end{equation}
as $x\to -\infty$, the jumps $J_{A,i}(z)$ are exponentially close to the identity matrix except the ones on the real line. Neglecting the exponential small terms, we arrive at the following outer parametrix.

\begin{rhp} \label{rhp:outer2} The  function $A^{(\infty)}(z)$ satisfies the following properties:
\begin{itemize}
  \item[(a)]   $A^{(\infty)}(z)$ is analytic in
  $\mathbb{C}\backslash (-\infty,|s|]$;

  \item[(b)]  $A^{(\infty)}(z)$ satisfies the jump condition
  \begin{align}\label{A outside-jump}
  &A^{(\infty)}_+(z)=A^{(\infty)}_-(z) \begin{cases}
    \left(\begin{array}{cc}
                                                             0 &  1\\
                                                            -1& 0
                                                           \end{array}\right), &  z\in (-\infty,0), \vspace{0.2cm}  \\
                                     e^{2\pi \alpha i\sigma_3},& z\in (0,|s|);
  \end{cases}
\end{align}

  \item[(c)] At infinity, $A^{(\infty)}(z)$ satisfies the same asymptotics as $A(z)$ in \eqref{A at infinity}.

\end{itemize}
\end{rhp}

Similar to \textbf{RH problem \ref{rhp:outer1}}, a solution to the above RH problem can be constructed explicitly as
\begin{equation}\label{A outside solution}
A^{(\infty)}(z)=\left(
                                                                         \begin{array}{cc}
                                                                           1 & 0 \\
                                                                           2\alpha i |s|^{1/2} & 1 \\
                                                                         \end{array}
                                                                       \right)
z^{-\frac{1}{4}\sigma_3}\frac{I+i\sigma_1}{\sqrt{2}}
   h(z/|s|)^{\sigma_3},
\end{equation}
where $h(z)$ is defined in \eqref{h-function}.

Then, we turn to the local parametrix near the origin. Let $\theta(z,x)^2 = z (\frac{2}{3}z + x)^2$ be a conformal mapping in the neighbourhood of the origin. Then, similar to $S^{(0)}(\xi)$ in \eqref{parametrix local} of Section \ref{Sec:steepest-large-gap}, the local parametrix $A^{(0)}(z)$ is given explicitly as follows:
\begin{equation}\label{A-parametrix local}
A^{(0)}(z)=E_1(z)Z_0\left(\theta(z,x)^2\right)\left\{\begin{array}{ll}
                                                     e^{\theta(z,x)\sigma_3}e^{\pi i\alpha\sigma_3}, \quad &\arg z\in (0, \pi),\\
                                                     e^{\theta(z,x)\sigma_3}e^{-\pi i\alpha\sigma_3},\quad &\arg z\in (-\pi, 0),
                                                   \end{array}
\right.
\end{equation}
where the pre-factor $E_1(z)$ is an analytic function in $U_0:= \{z: \ |z| < \delta\}$ and $Z_0(z)$ is given in \eqref{phi related to bessel}.

To match the local parametrix $A^{(0)}(z)$ with the outer parametrix $A^{(\infty)}(z)$ in \eqref{A outside solution} on $\partial U_0$, we choose the analytic pre-factor $E_1(z)$ in \eqref{A-parametrix local} as
\begin{equation}\label{E-1}
E_1(z)= A^{(\infty)}(z)\left\{\begin{array}{ll}
                                                     e^{ -\pi i\alpha \sigma_3}\frac{I-i\sigma_1}{\sqrt{2}} z^{\frac {1}{4}\sigma_3}\left (|x|-\frac 23 z\right)^{\frac{1}{2}\sigma_3}, \quad& \arg z\in (0, \pi),\\[.2cm]
                                                     e^{ \pi i\alpha\sigma_3}\frac{I-i\sigma_1}{\sqrt{2}}z^{\frac {1}{4}\sigma_3}\left (|x|-\frac 23 z\right)^{\frac{1}{2}\sigma_3}, \quad&   \arg z\in (-\pi, 0),
                                                   \end{array}
\right. \textrm{for } z \in U_0.
\end{equation}
Then,  the following matching condition is fulfilled
\begin{equation}\label{mathching-1}
A^{(0)}(z)=\left(I+O(\frac 1x)\right)A^{(\infty)}(z) \qquad \textrm{as } x\to-\infty,
\end{equation}
uniformly for $z \in \partial U_0$.

With the outer and local parametrices constructed explicitly in \eqref{A outside solution} and \eqref{A-parametrix local}, we introduce the final transformation as follows:
\begin{equation}\label{B}
 B(z)= \begin{cases}
    A(z)(A^{(\infty)}(z))^{-1}, &  |z|>\delta, \\  A(z)(A^{(0)}(z))^{-1}, & |z|< \delta.
 \end{cases}
\end{equation}
By the matching condition \eqref{mathching-1}, one can verify that the jump of $B(z)$ is
$$J_B(z)=I+O(\frac 1x) \qquad \textrm{as } x\to-\infty,$$
which implies
\begin{equation}\label{B-estimate}
B(z)=I+O(\frac 1x) \qquad \textrm{as } x\to-\infty,
\end{equation}
uniformly  for $z$ in the complex plane; see the similar analysis in \textbf{RH problem \ref{rhp:c-final}}.

\subsubsection{Asymptotics of $v_i(x)$} \label{sec:vi-1}

Recall the transformation \eqref{A} and the representations of $v_i(x)$ in \eqref{v1-Phi'} and \eqref{v2-Phi'}, we have
\begin{eqnarray}
  v_1(x)&=&i\lim_{z\to 0}z(A'(z)A(z)^{-1})_{12}, \label{v-1 estimate} \\
  v_2(x)&=&i\lim_{z\to -s}(z+s)(A'(z)A(z)^{-1})_{12}. \label{v-2 estimate}
\end{eqnarray}
By the transformation \eqref{B}, we obtain
\begin{equation}\label{A-expression}
A(z)=B(z)A^{(0)}(z) \qquad \textrm{for } |z|< \delta.
\end{equation}
Note that $B(z)$ is analytic at the origin and satisfies the approximation \eqref{B-estimate}. Moreover, by \eqref{A-parametrix local},
\eqref{phi related to bessel} and \eqref{E-1} in  the expression of $A^{(0)}(z)$, the pre-factor $E_1(z)$   is analytic at the origin and the Bessel parametrix satisfies the following relation
$$\lim_{z\to 0}z Z_0'(z) Z_0(z)^{-1}=\left(
                                                                                                                       \begin{array}{cc}
                                                                                                                         0 & \frac {1}{2i} \\
                                                                                                                         0 & 0 \\
                                                                                                                       \end{array}
                                                                                                                     \right).$$
Thus, we have from \eqref{v-1 estimate} and \eqref{A-expression}
\begin{equation}\label{v-estimate-1}
v_1(x)=\left(B(0)E_1(0)\left(
                                                                                                                       \begin{array}{cc}
                                                                                                                         0 & \frac 1{2} \\
                                                                                                                         0 & 0 \\
                                                                                                                       \end{array}
                                                                                                                     \right)
E_1(0)^{-1}B(0)^{-1}\right)_{12}=\frac {1}{2}(E_1(0))_{11}^2(1+O(\frac {1}{x})).
\end{equation}
From \eqref{A outside solution} and \eqref{E-1}, we obtain
\begin{equation}\label{E-1-estimate}
(E_1(0))_{11}^2=|x|.
\end{equation}
Finally, we get the asymptotics
\begin{equation}\label{v-1-asy}
v_1(x)=-\frac {x}{2}(1+O(1/x)), \quad x\to-\infty.
\end{equation}

Similarly, by using the relation \eqref{B}, \eqref{v-2 estimate}, the definition of $A^{(\infty)}$ in \eqref{A outside solution} and the approximation \eqref{B-estimate}, we
have
\begin{equation}\label{v-2-asy}
v_2(x)=\frac {\alpha}{\sqrt{|s|}} +O(1/x), \quad x\to-\infty.
\end{equation}


\subsection {Case II: $s>0$  }

For $s>0$ and $\omega=0$, it is easy to see that $\Phi(z)$ is analytic at $z=s$ and $v_1(x)=0$. Therefore, in this section, we consider the case $\omega=e^{-2\pi i\beta}$ with $|\mathrm{Re}\beta|<\frac {1}{2}$.

\subsubsection{Nonlinear steepest descent analysis of $\Phi$ as $x \to -\infty$}

When $s<0$, the function $A(z)$ defined in \eqref{A} satisfies the following RH problem.
\begin{figure}[h]
 \begin{center}
   \includegraphics[width=5cm]{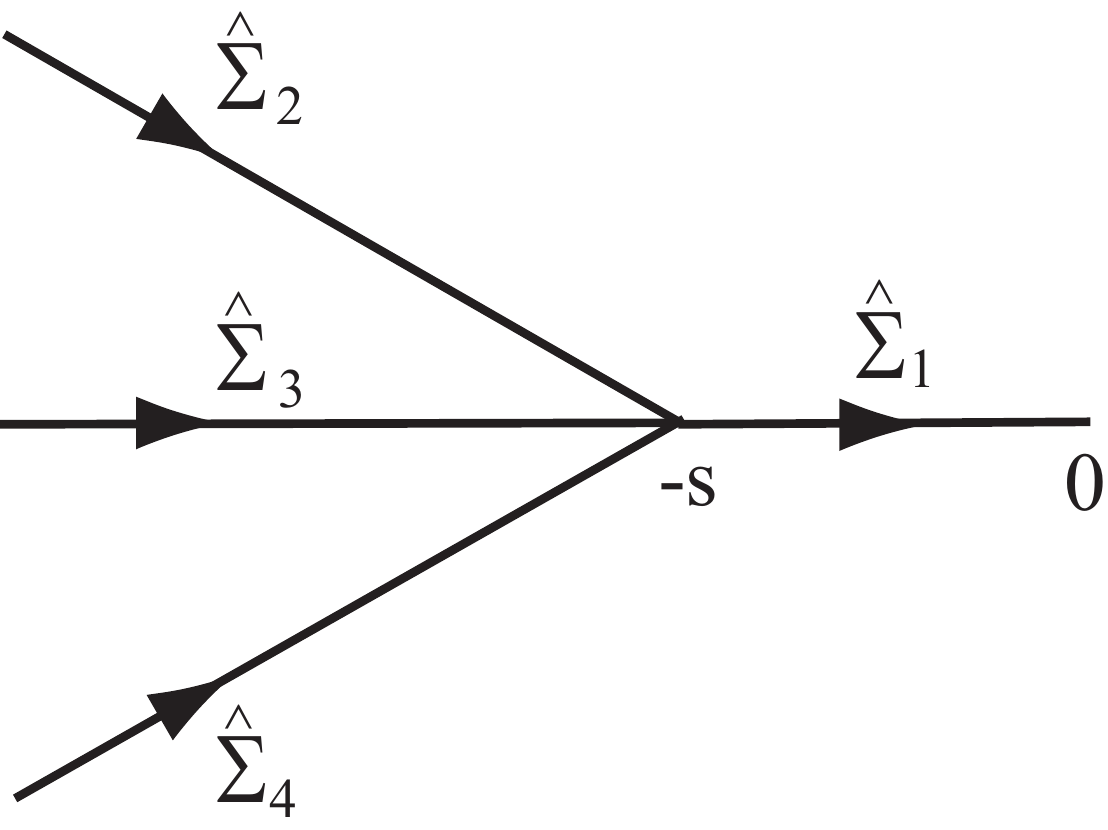} \end{center}
 \caption{\small{Regions and contours  for   $A$ with $s>0$.}}
 \label{Figure-s positive }
 \end{figure}
\begin{rhp} The function $A(z)$ satisfies the following properties:
\begin{itemize}
  \item[(a)]   $A(z)$ is analytic in
  $\mathbb{C}\backslash \hat{\Sigma}_{i}$, where the contours $\hat{\Sigma}_{i}$ are indicated in Fig. \ref{Figure-s positive };

  \item[(b)]  $A(z)$ satisfies the jump condition
  \begin{equation}\label{A-III-jumps}
  A_+(z)=A_-(z) J_{A,i}(z) \qquad \textrm{for } z\in\hat{\Sigma}_{i}
  \end{equation}
with
$$J_{A,1}=\left(\begin{array}{cc}
                                                             e^{g_{0+}(z)-g_{0-}(z)} & e^{-2\pi i\beta}  \\
                                                            0 &   e^{g_{0-}(z)-g_{0+}(z)}
                                                           \end{array}\right),  \quad  J_{A,2}=\left(\begin{array}{cc}
                                                             1 & 0 \\
                                                            e^{2\theta(z,x)} e^{2\pi i \alpha} & 1
                                                           \end{array}\right), $$
                                                           $$      J_{A,3}=\left(\begin{array}{cc}
                                                             0 & 1 \\
                                                            -1 & 0
                                                           \end{array}\right),\quad J_{A,4}=\left(\begin{array}{cc}
                                                             1 & 0 \\
                                                             e^{2\theta(z,x)}e^{-2\pi i \alpha} & 1
                                                           \end{array}\right);$$

  \item[(c)]   The asymptotic behavior of $A(z)$  at infinity
  \begin{equation}\label{A-IIII at infinity}
  A(z)=
    \left (I+O\left (\frac 1 z\right )\right)
  z^{-\frac{1}{4}\sigma_3}\frac{I+i\sigma_1}{\sqrt{2}}  \qquad \textrm{as } z \rightarrow \infty ;
  \end{equation}

\item[(d)] The asymptotic behavior of $A(z)$  at $z=0,-s$ is the same as that of $\Phi(z)$.
\end{itemize}
\end{rhp}

Based on the factorization
$$\left(\begin{array}{cc}
                                                             e^{g_{0+}(z)-g_{0-}(z)} & e^{-2\pi i\beta}  \\
                                                            0 &   e^{g_{0-}(z)-g_{0+}(z)}
                                                           \end{array}\right)=\left(\begin{array}{cc}
                                                             1 & 0\\
                                                             e^{2\pi i\beta} e^{2g_{0-}(z)} &  1
                                                           \end{array}\right)\left(\begin{array}{cc}
                                                             0 & e^{-2\pi i\beta} \\
                                                           -e^{2\pi i\beta} &  0
                                                           \end{array}\right)\left(\begin{array}{cc}
                                                             1 & 0\\
                                                             e^{2\pi i\beta} e^{2g_{0+}(z)} &  1
                                                           \end{array}\right),$$
we introduce the second transformation
\begin{equation}\label{B-II} B(z)=\left\{\begin{array}{cc}
                                               A(z)\left(\begin{array}{cc}
                                                             1 & 0\\
                                                               e^{2\pi i\beta} e^{2\theta(z,x)} &  1
                                                           \end{array}\right), & \mbox{$z$ in the lower lens-shaped region;}\\
                                              A(z)\left(\begin{array}{cc}
                                                             1 & 0\\
                                                             -   e^{2\pi i\beta}e^{2\theta(z,x)} &  1
                                                           \end{array}\right), & \mbox{$z$ in the upper lens-shaped region;}\\
                                               A(z), & \mbox{otherwise}.
                                             \end{array}
\right.
  \end{equation}
  \begin{figure}[h]
 \begin{center}
   \includegraphics[width=6cm]{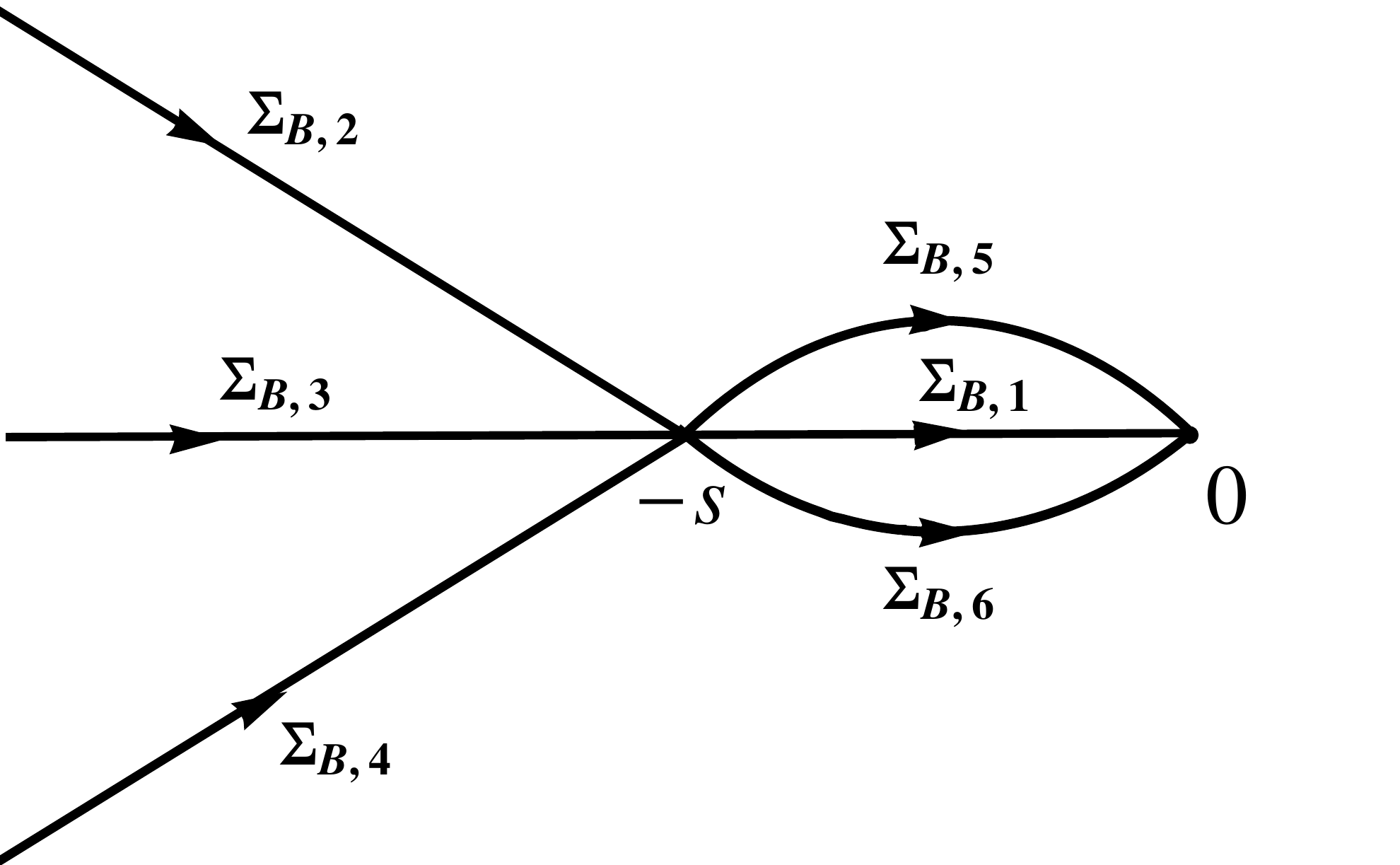} \end{center}
 \caption{\small{Regions and contours  for   $B$.}}
 \label{Figure-B-1 }
 \end{figure}
Then, $B(z)$ satisfies  the following RH problem.
\begin{rhp} The function $B(z)$ satisfies the following properties:
  \begin{itemize}
  \item[(a)]   $B(z)$ is analytic in
  $\mathbb{C}\backslash \Sigma_{B,i}$; see Fig. \ref{Figure-B-1 } for the contours $ \Sigma_{B,i}$;

  \item[(b)]  $B(z)$ satisfies the jump condition
  \begin{equation}\label{B-III-jumps}
  B_+(z)=B_-(z) J_{B,i}(z) \qquad \textrm{for } z\in\Sigma_{B,i}
  \end{equation}
with
$$J_{B,1}=\left(\begin{array}{cc}
                                                             0 &   e^{-2\pi i\beta} \\
                                                            -  e^{2\pi i\beta} &  0
                                                           \end{array}\right),  \quad  J_{5}=J_{6}=\left(\begin{array}{cc}
                                                             1 & 0 \\
                                                            e^{2\pi i\beta}e^{2\theta(z,x)} & 1
                                                           \end{array}\right), $$
                                                           $$  J_{B,2}=\left(\begin{array}{cc}
                                                             1 & 0 \\
                                                            e^{2\theta(z,x)} e^{2\pi i \alpha} & 1
                                                           \end{array}\right),     J_{B,3}=\left(\begin{array}{cc}
                                                             0 & 1 \\
                                                            -1 & 0
                                                           \end{array}\right), J_{B,4}=\left(\begin{array}{cc}
                                                             1 & 0 \\
                                                             e^{2\theta(z,x)}e^{-2\pi i \alpha} & 1
                                                           \end{array}\right);$$

  \item[(c)] At infinity, $B(z)$ satisfies the same asymptotics as $A(z)$ in \eqref{A-IIII at infinity};

\item[(d)] $B(z)$  has the same behavior near  $z=0,-s$ with $\Phi(z)$ given in \eqref{Phi at 0} and \eqref{Phi at -s}.
\end{itemize}
\end{rhp}

Due to \eqref{theta-z-realpart}, as $x\to -\infty$, the jumps are close to the identity matrix except the ones on the real line.
So, we arrive at the following outer parametrix.
\begin{rhp} The function $B^{(\infty)}(z)$ satisfies the following properties:
\begin{itemize}
  \item[(a)]   $B^{(\infty)}(z)$ is analytic in
  $\mathbb{C}\backslash (-\infty,0]$;

  \item[(b)]  $B^{(\infty)}(z)$ satisfies the following jump condition
  \begin{align}\label{B outside-jump}
  &B^{(\infty)}_+(z)=B^{(\infty)}_-(z) \left(\begin{array}{cc}
                                                             0 &  1\\
                                                            -1& 0
                                                           \end{array}\right),\quad  z\in (-\infty,-s),\nonumber\\
&B^{(\infty)}_+(z)=B^{(\infty)}_-(z)\left(\begin{array}{cc}
                                                             0 &  e^{-2\pi i\beta}\\
                                                            -e^{2\pi i\beta}& 0
                                                           \end{array}\right) ,\quad z\in (-s,0);
\end{align}

  \item[(c)] At infinity, $B^{(\infty)}(z)$ satisfies the same asymptotics as $A(z)$ in \eqref{A-IIII at infinity};

\end{itemize}
\end{rhp}

To construct the outer parametrix, let us define the following scalar function $h_2(z)$ by
\begin{equation}\label{h1-function}
 h_2(z)=\left(\frac {\sqrt{z}+i\sqrt{s}}{\sqrt{z}-i\sqrt{s}}\right)^{\beta},\quad z\in \mathbb{C}\setminus (-\infty,0],
\end{equation}
where the power function $z^c, c \notin \mathbb{Z},$ takes the principle branch with the branch cut along  $(-\infty,0)$.
Then, $h_2(z)$ satisfies the following jump condition
\begin{equation}
h_{2+}(x)h_{2-}(x)=\left\{\begin{array}{c}
                                        e^{2\pi i\beta}, \quad x\in(-s,0) \\
                                         1, \quad x\in(-\infty,-s).
                                      \end{array} \right.
\end{equation}
With the function $h_2(z)$, one solution to the above RH problem is explicitly given by
\begin{equation}\label{B-infty- solution}
 B^{(\infty)}(z)=\left(
                                                                         \begin{array}{cc}
                                                                           1 & 0 \\
                                                                          2\beta \sqrt{s} & 1 \\
                                                                         \end{array}
                                                                       \right)
z^{-\frac{1}{4}\sigma_3}\frac{I+i\sigma_1}{\sqrt{2}}
 h_2(z)^{\sigma_3}.
\end{equation}

In the neighborhood of the origin,  the local parametrix is similar to \eqref{A-parametrix local} and given explicitly as follows:
\begin{equation}\label{B-0}
  B^{(0)}(z)=E_2(z)Z_0(\theta(z,x)^2)e^{(\frac {2}{3}z^{\frac {3}{2}}+xz^{\frac {1}{2}})\sigma_3}e^{\pi i \beta\sigma_3},
\end{equation}
where  $Z_0$ is defined in terms of  the Bessel functions \eqref{phi related to bessel}.
 The analytic pre-factor $E_2(z)$ is given by
\begin{equation}\label{E-2}
 E_2(z)= B^{(\infty)}(z)\left\{\begin{array}{ll}
                                                     e^{ -\pi i\beta\sigma_3}\frac{I-i\sigma_1}{\sqrt{2}} z^{\frac {1}{4}\sigma_3}\left (|x|-\frac 23 z\right)^{\frac{1}{2}\sigma_3}, \quad& \arg z\in (0, \pi),\\[.2cm]
                                                     e^{- \pi i\beta\sigma_3}\frac{I-i\sigma_1}{\sqrt{2}} z^{\frac {1}{4}\sigma_3}\left (|x|-\frac 23 z\right)^{\frac{1}{2}\sigma_3}, \quad&   \arg z\in (-\pi, 0),
                                                   \end{array}
\right.
\end{equation}
for $|z|<\delta$. With the properties of the Bessel functions, one can verify the  following matching condition
\begin{equation}\label{Matching -2}
B^{(0)}(z)=(I+O(1/x))B^{(\infty)}(z),
\end{equation}
 uniform for $|z|=\delta$  as $x\to -\infty$.

 In the neighborhood of $-s$, we seek a local parametrix of the following form
\begin{equation}\label{B-1}
B^{(1)}(z)=E_3(z)\hat{B}^{(1)}(f_0(z))e^{\frac {1}{2}\pi i \beta\sigma_3}e^{\theta(z,x)\sigma_3},
 \end{equation}
where $f_0(z)$ is the  conformal mapping near $-s<0$:
 \begin{equation}\label{B-1-conforml mapping}
 \begin{split}
   f_0(z)&:=(|x|i z^{\frac {1}{2}}-\frac {2i}{3}z^{\frac {3}{2}})+(|x|s^{\frac {1}{2}}+\frac {2}{3}s^{\frac {3}{2}})\\
   &=(s^{1/2}+\frac {1}{2}|x|s^{-1/2})(z+s)+O((z+s)^2),
   \end{split}
    \end{equation}
 with $\arg z\in(0,2\pi)$ and   $E_3(z)$ is analytic for $|z+s|<\frac {s}{2}$.
 Let
\begin{equation}\label{B-tilde}
\tilde{B}^{(1)}(z)=\hat{B}^{(1)}(z)\left\{\begin{array}{c}
                                                                            e^{-\pi i\alpha\sigma_3}, \quad \arg z\in(0,\frac {\pi}{2}) \\
                                                                             e^{\pi i\alpha\sigma_3},\quad \arg z\in(\frac {3\pi}{2},2\pi)\\
                                                                            I, \quad \arg z\in(\frac {\pi}{2},\frac {3\pi}{2}).
                                                                           \end{array}
\right.
    \end{equation}
Then, $\tilde{B}^{(1)}$ satisfies  the following RH problem.
\begin{figure}[h]
 \begin{center}
   \includegraphics[width=5cm]{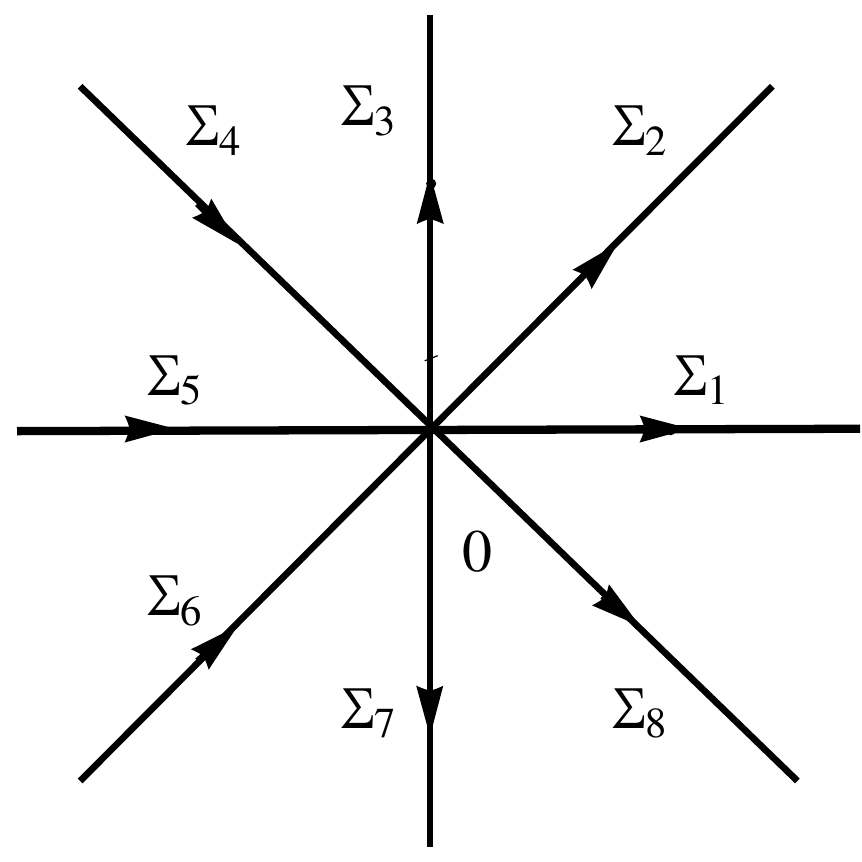} \end{center}
 \caption{\small{Regions and contours  for   $\tilde{B}$.}}
 \label{Figure-CHF}
 \end{figure}
\begin{rhp} The function $\tilde{B}^{(1)}(z)$ satisfies the following properties:
  \begin{itemize}
  \item[(a)]   $\tilde{B}^{(1)}(z)$ is analytic in
  $\mathbb{C}\backslash \Sigma_{i}$, where $\Sigma_{i}$ are indicated in Fig. \ref{Figure-CHF};

  \item[(b)] $\tilde{B}^{(1)}(z)$ satisfies the following jump condition
  \begin{equation}
    \tilde{B}_+^{(1)}(z)=\tilde{B}_-^{(1)}(z) J_i(z), \qquad \textrm{for } z \in \Sigma_i
  \end{equation}
  with
  \begin{equation*}
    J_1 = \left(\begin{array}{cc}
                                                             0 &   e^{-\pi i\beta} \\
                                                            -  e^{\pi i\beta} &  0
                                                           \end{array}\right), \quad J_2 = \left(\begin{array}{cc}
                                                             1 & 0 \\
                                                            e^{ \pi i(\beta-2\alpha)} & 1
                                                           \end{array}\right), \quad J_3 = J_7 = e^{\pi i\alpha\sigma_3},
  \end{equation*}
  \begin{equation*}
    J_4 = \left(\begin{array}{cc}
                                                             1 & 0 \\
                                                            e^{ -\pi i(\beta-2\alpha)} & 1
                                                           \end{array}\right), \quad J_5 = \left(\begin{array}{cc}
                                                             0 &   e^{\pi i\beta} \\
                                                            -  e^{-\pi i\beta} &  0
                                                           \end{array}\right), \quad J_6 = \left(\begin{array}{cc}
                                                             1 & 0 \\
                                                            e^{- \pi i(\beta+2\alpha)} & 1
                                                           \end{array}\right),
  \end{equation*}
  and $J_8 = \left(\begin{array}{cc}
                                                             1 & 0 \\
                                                            e^{\pi i(\beta+2\alpha)} & 1
                                                           \end{array}\right);$

\item[(c)] As $z\to 0$,
$$\tilde{B}^{(1)}(z)=\left(
                       \begin{array}{cc}
                         O(|z|^{\alpha}) & O(|z|^{-|\alpha|}) \\
                         O(|z|^{\alpha}) & O(|z|^{-|\alpha|}) \\
                       \end{array}
                     \right), \quad \alpha\neq 0,$$
and $$\tilde{B}^{(1)}(z)=\left(
                       \begin{array}{cc}
                         O(1) & O(\ln z) \\
                         O(1) & O(\ln z) \\
                       \end{array}
                     \right), \quad \alpha=0.$$
\end{itemize}
\end{rhp}
The above RH problem is the same as the one in \cite[(4.25)-(4.31)]{DIK1} up to a rotation.
The solution  $\tilde{B}^{(1)}(z)$ can be constructed explicitly in terms of
the confluent hypergeometric functions $\psi(a,b,z)$
\begin{align*}
\tilde{B}^{(1)}(z)&=C_0\left(\begin{array}{ll}
(2e^{\pi i/2}z)^\alpha\psi(\alpha+\beta,1+2\alpha,2e^{\pi i/2}z)e^{i\pi(\alpha+2\beta)}e^{-iz}\\
-\frac{\Gamma(1+\alpha+\beta)}{\Gamma(\alpha-\beta)}(2e^{\pi i/2}z)^{-\alpha}\psi(1-\alpha+\beta,1-2\alpha,2e^{\pi i/2}z)e^{i\pi(-3\alpha+\beta)}e^{-iz}\end{array}\right.\nonumber\\
&\quad\quad\left.\begin{array}{rr}
-\frac{\Gamma(1+\alpha-\beta)}{\Gamma(\alpha+\beta)}(2e^{\pi i/2}z)^\alpha\psi(1+\alpha-\beta,1+2\alpha,2e^{-\pi i/2}\zeta)e^{i\pi(\alpha+\beta)}e^{iz}\\
(2e^{\pi i/2}z)^{-\alpha}\psi(-\alpha-\beta,1-2\alpha,2e^{-\pi i/2}z)e^{-i\pi\alpha}e^{iz}\end{array}\right),
\end{align*}
for $z$ in the sector with boundary $\Sigma_1$ and $\Sigma_2$, and the constant matrix $C_0$ is
 $$
 C_0=2^{\beta \sigma_3}e^{\beta\pi i\sigma_3/2}
\left(\begin{array}{cc}e^{-i\pi(\alpha+2\beta)} & 0\\ 0 & e^{i\pi(2\alpha+\beta)}\end{array}\right);
$$
see \cite{DIK1,ik}. The expression of the solution in the other sectors is then determined by using the jump conditions. Note that, from the properties of the confluent hypergeometric functions $\psi(a,b,z)$, the asymptotics of $\tilde{B}^{(1)}(z)$ as $z\to\infty$ is given by
\begin{equation}\label{CHF at infinity}
\begin{split}
  \tilde{B}^{(1)}(z)=  &  \left [I+\frac {1}{z} \left(
                                                        \begin{array}{cc}
                                                         - \frac {i(\alpha^2-\beta^2)}{2} & -i2^{2\beta-1}\frac {\Gamma(1+\alpha-\beta)}{\Gamma(\alpha+\beta)}e^{i\pi(\alpha-\beta)} \\
                                                           i2^{-2\beta-1}\frac {\Gamma(1+\alpha+\beta)}{\Gamma(\alpha-\beta)}e^{-i\pi(\alpha-\beta)} &  \frac {i(\alpha^2-\beta^2)}{2} \\
                                                        \end{array}
                                                      \right)\right.\\
&\left. +O\left (\frac 1 {z^2}\right )\right ]
 z^{-\beta \sigma_3}e^{-iz\sigma_3},
\end{split}
\end{equation}
in the region $\arg z \in(0,\pi/2)$. 
Moreover, $\tilde{B}^{(1)}(z)$ satisfies the following differential equation
\begin{equation}\label{CHF-equation}
\frac{d}{dz}\tilde{B}^{(1)}(z)=\left(-i\sigma_3+\frac {1}{z}\left(
                                                                      \begin{array}{cc}
                                                                        -\beta & 2^{2\beta}\frac {\Gamma(1+\alpha-\beta)}{\Gamma(\alpha+\beta)}e^{i\pi(\alpha-\beta)}  \\
                                                                        2^{-2\beta}\frac {\Gamma(1+\alpha+\beta)}{\Gamma(\alpha-\beta)}e^{-i\pi(\alpha-\beta)} & \beta \\
                                                                      \end{array}
                                                                    \right)\right)
\tilde{B}^{(1)}(z).\end{equation}
We take the analytic pre-factor $E_3(z)$ in \eqref{B-1} as
\begin{equation}\label{E-3}
E_3(z)=\left\{\begin{array}{ll}
B^{(\infty)}(z)e^{-\left (\alpha+\frac {\beta}2\right )\pi i\sigma_3}e^{i(\frac {2}{3}|s|^{\frac{3}{2}}-x|s|^{\frac{1}{2}} )\sigma_3}f_0(z)^{\beta \sigma_3},& \Im z>0,\\
B^{(\infty)}(z)\left(\begin{array}{cc}0 & 1\\-1 & 0\end{array}\right)
e^{-\left (\alpha+\frac {\beta}2\right )\pi i\sigma_3}e^{i(\frac {2}{3}|s|^{\frac{3}{2}}-x|s|^{\frac{1}{2}} )\sigma_3} f_0(z)^{\beta \sigma_3}, & \Im z<0,  \end{array}\right.
\end{equation}
with $\arg f_0(z)\in (0,2\pi)$. Particularly, we have
\begin{equation}\label{E-3-s}
E_3(-s)=s^{-\frac {1}{4}\sigma_3}\left(
         \begin{array}{cc}
           1 & 0 \\
          2\beta & 1\\
         \end{array}
       \right)\frac {1}{\sqrt{2}}\left(
                \begin{array}{cc}
                  1 & 1 \\
                  -1 & 1 \\
                \end{array}
              \right)e^{-(\alpha-\frac {\beta}{2}+\frac  {1}{4})\pi i\sigma_3}e^{i(\frac {2}{3}s^{3/2}+|x|s^{1/2})\sigma_3}(2\sqrt{s}|x-2s|)^{\beta\sigma_3}.
\end{equation}
Using the expressions of $B^{(\infty)}(z)$ in \eqref{B-infty- solution}, the asymptotics of $\tilde{B}^{(1)}(z)$ in \eqref{CHF at infinity} and the definition of $B^{(1)}(z)$  in \eqref{B-1}, we have the following matching condition
\begin{equation}\label{Matching-CHG}
B^{(1)}(z)B^{(\infty)}(z)^{-1}=I+x^{-\mathrm{Re}\beta\sigma_3}O(1/x)x^{\mathrm{Re}\beta\sigma_3}
=I+O(x^{2|\mathrm{Re}\beta|-1}).
\end{equation}

In the final transformation, we define
\begin{equation}\label{C}
C(z)=\begin{cases}
  B(z)(B^{(0)}(z))^{-1}, & |z|<\delta \\
  B(z)(B^{(1)}(z))^{-1}, & |z+s|< \delta \\
  B(z)(B^{(\infty)}(z))^{-1},& \mbox{ otherwise}.
\end{cases}
\end{equation}
By the matching conditions \eqref{Matching -2} and \eqref{Matching-CHG}, one can verify that the jump of $C(z)$ is
$$J_C(z)=I+O(x^{2|\mathrm{Re}\beta|-1}),$$
where $2|\mathrm{Re}\beta|-1<0$. Therefore, we have
\begin{equation}\label{C-estimate}
C(z)=I+O(x^{2|\mathrm{Re}\beta|-1}),
\end{equation}
uniformly  for $z$ in the complex plane; see the similar analysis in \textbf{RH problem \ref{rhp:c-final}}.

\subsubsection{Asymptotics of $v_i(x)$} \label{sec:vi-2}

Recall the transformations \eqref{A} and \eqref{B-II}, and the representations of $v_i(x)$ in \eqref{v1-Phi'} and \eqref{v2-Phi'}, we have
\begin{equation}\label{v-estimate}
 v_1(x)=i\lim_{z\to 0^{+}}z(B'(z)B(z)^{-1})_{12},
 \end{equation}
and
\begin{equation}\label{v-2-estimate-1}
 v_2(x)=i\lim_{z\to -s}(z+s)(B'(z)B(z)^{-1})_{12}.
 \end{equation}
From the transformation  \eqref{C}, we get
\begin{equation}\label{A-II-expression}
B(z)=C(z)B^{(0)}(z), \qquad \textrm{for $|z|< \delta$.}
\end{equation}
Thus, the asymptotics of $v_1(x)$ follows from \eqref{B-0} and the approximation \eqref{C-estimate}
\begin{equation}\label{v-1-asy-case II}
v_1(x)=-\frac {x}{2}(1+O(x^{2|\mathrm{Re}\beta|-1})), \quad x\to-\infty.
\end{equation}
Combining \eqref{B-1}, \eqref{CHF-equation} and \eqref{E-3}, we get from \eqref{v-2-estimate-1}
\begin{equation}\label{v-2-estimate-2}
 v_2(x)=i\left(C(-s)E_3(-s)\left(
                                                                      \begin{array}{cc}
                                                                        -\beta & 2^{2\beta}\frac {\Gamma(1+\alpha-\beta)}{\Gamma(\alpha+\beta)}e^{i\pi(\alpha-\beta)}  \\
                                                                        2^{-2\beta}\frac {\Gamma(1+\alpha+\beta)}{\Gamma(\alpha-\beta)}e^{-i\pi(\alpha-\beta)} & \beta \\
                                                                      \end{array} \right)
 E_3^{-1}(-s)C^{-1}(-s)\right)_{12}.
 \end{equation}
Some straightforward computations give us
\begin{equation}\label{v-2-asy-case II}
v_2(x)=\frac{1}{\sqrt{s}}\left(i\beta+\frac {1}{2}\frac {\Gamma(1+\alpha-\beta)}{\Gamma(\alpha+\beta)}e^{i\tilde{\vartheta}(x,s,\alpha,\beta)}+\frac {1}{2}\frac {\Gamma(1+\alpha+\beta)}{\Gamma(\alpha-\beta)}e^{-i\tilde{\vartheta}(x,s,\alpha,\beta)}\right)(1+O(x^{2\mathrm{Re} \beta-1})),
\end{equation}
where $\tilde{\vartheta}(x,s,\alpha,\beta)=-2xs^{1/2}+\frac {4}{3}s^{3/2}-\alpha\pi -6i\beta\ln 2-i\beta\ln(s)-2i\beta \ln|\frac {x}{2}-s| $.

\medskip

Finally, summarizing the asymptotics of $v_i(x)$ as $x\to - \infty $ in \eqref{v-1-asy}, \eqref{v-2-asy}, \eqref{v-1-asy-case II} and \eqref{v-2-asy-case II}, we obtain the following results.

{\thm{\label{theorem-v-2}Under the same conditions as in Theorem \ref{theorem-v-1}, the asymptotic behaviors of $v_i(x)$ to the coupled $\textrm{P}_{2}$ equations \eqref{int-equation v} are given by
\begin{equation}\label{int-v-1-asy-infty}
v_1(x,s;2\alpha,\omega)=-\frac {x}2 \biggl( 1+o(1) \biggr),
\end{equation}
as $x\to -\infty$; and
\begin{equation}
v_2(x,s;2\alpha,\omega)= \begin{cases}
    \frac {1}{\sqrt{s}} \biggl( i\beta+\frac {\Gamma(1+\alpha-\beta)}{2\Gamma(\alpha+\beta)}e^{i\tilde{\vartheta}(x,s,\alpha,\beta)} & \mbox{if}~ s>0, \\
  \quad \quad +\frac {\Gamma(1+\alpha+\beta)}{2\Gamma(\alpha-\beta)}e^{-i\tilde{\vartheta}(x,s,\alpha,\beta)} \biggr)(1+O(x^{2\mathrm{Re} \beta-1})), \\
  \frac {\alpha}{ \sqrt{|s| } }+O(1/x), & \mbox{if}~ s<0,
\end{cases}
\end{equation}
as $x\to -\infty$, where $\omega=e^{-2\pi i\beta}$, $|\Re \beta|<1/2$ and $\tilde{\vartheta}(x,s,\alpha,\beta)=-2xs^{1/2}+\frac {4}{3}s^{3/2}-\alpha\pi -6i\beta\ln 2-i\beta\ln(s)-2i\beta \ln|\frac {x}{2}-s| $. }}

\end{appendices}


\end{document}